\documentclass[runningheads,11pt]{llncs}

\usepackage{listings}
\usepackage{xcolor}

\usepackage[T1]{fontenc}
\usepackage[edges]{forest}

\usepackage{amsfonts}
\usepackage{graphicx}
\usepackage[utf8]{inputenc} 
\usepackage{url}            
\usepackage[numbers,sort]{natbib}
\usepackage{hyperref}       
\usepackage{nicefrac}       
\usepackage{microtype}      
\hypersetup{
     colorlinks  = true,
     urlcolor    = teal,
	 citecolor   = teal,
	 linkcolor   = red
}

\makeatletter
\def\set@curr@file#1{\def\@curr@file{#1}} 
\makeatother
\usepackage[load-configurations=version-1]{siunitx} 

\usepackage{algorithm}
\usepackage[noend]{algorithmic}
\usepackage{enumitem}
\usepackage{multicol}

\usepackage{dsfont}
\usepackage{mathtools}
\usepackage{amsmath}
\usepackage{xspace}

\makeatletter
\renewenvironment{proof}[1][\proofname]{\par
\normalfont \topsep6\p@\@plus6\p@\relax
\trivlist
\item\relax
{\itshape
#1\@addpunct{.}}\hspace\labelsep\ignorespaces
}{%
\qed\endtrivlist\@endpefalse
}
\makeatother

\makeatletter
\let\original@algocf@latexcaption\algocf@latexcaption
\long\def\algocf@latexcaption#1[#2]{%
  \@ifundefined{NR@gettitle}{%
    \def\@currentlabelname{#2}%
  }{%
    \NR@gettitle{#2}%
  }%
  \original@algocf@latexcaption{#1}[{#2}]%
}
\makeatother

\newcommand{\ug}{U_>}
\newcommand{\ul}{U_<}

\newcommand{\suff}{\textnormal{suffix}}

\newcommand{\prefix}{\textit{prefix}}
\newcommand{\safe}{\textit{MIR}}
\newcommand{\pathto}[2]{#1\stackrel{\pi}{\rightsquigarrow}#2}
\newcommand{\pth}[3]{#1\stackrel{#2}{\rightsquigarrow}#3}
\newcommand{\pathtorho}[2]{#1\stackrel{\rho}{\rightsquigarrow}#2}

\renewcommand{\above}{\textnormal{pos}}
\newcommand{\below}{\textnormal{neg}}

\renewcommand{\refeq}[1]{(\ref{#1})}
\newcommand{\tmu}[1]{\tilde \mu_{#1}}
\newcommand{\sigr}{\sigma^{\textnormal{right}}}
\newcommand{\sigl}{\sigma^{\textnormal{left}}}
\newcommand{\setmuskip}[2]{#1=#2\relax}

\DeclarePairedDelimiter{\ceil}{\lceil}{\rceil}

\newenvironment{proofof}[1]{\begin{proof}[\textnormal{\textbf{Proof of~#1}}]}{\end{proof}} 

\makeatletter
\let\original@algocf@latexcaption\algocf@latexcaption
\long\def\algocf@latexcaption#1[#2]{%
  \@ifundefined{NR@gettitle}{%
    \def\@currentlabelname{#2}%
  }{%
    \NR@gettitle{#2}%
  }%
  \original@algocf@latexcaption{#1}[{#2}]%
}
\makeatother

\newcount\Comments  
\newcount\Includeappendix 
\newcount\Putacknowledgement 

\definecolor{darkgreen}{rgb}{0,0.6,0}
\newcommand{\kibitz}[2]{\ifnum\Comments=1{\color{#1}{#2}}\fi}

\newcommand{\gre}{\textsc{Greedy}\xspace}
\newcommand{\mainalg}{\SEGB}

\makeatletter
\newcommand{\RemoveAlgoNumber}{\renewcommand{\fnum@algocf}{\AlCapSty{\AlCapFnt\algorithmcfname}}}
\newcommand{\RevertAlgoNumber}{\algocf@resetfnum}
\makeatother

\newcommand{\mP}{\mathcal{P}}

\newcommand{\mI}{\mathcal{I}}

\newcommand{\mH}{\mathcal{H}}

\newcommand{\mS}{\mathcal{S}}

\newcommand{\mU}{SW}
\newcommand{\ind}{\mathds{1}}

\newcommand{\R}{\mathbb{R}}

\newcommand{\defeq}{\stackrel{\text{def}}{=}}

\newcommand{\tupbracket}[1]{\left\langle {#1} \right\rangle}

\newtheorem{assumption}{Assumption}

\newtheorem{claim}{Claim}

\newtheorem{observation}{Observation}

\newcommand\bl[1]{\boldsymbol{ #1 } }

\newcommand\abs[1]{\left| #1  \right|}
\DeclareMathOperator*{\argmin}{arg\,min} 
\DeclareMathOperator*{\argmax}{arg\,max} 

\DeclareMathOperator{\E}{\mathbb{E}}

\newcounter{algorithmpolicy}
\makeatletter
\newenvironment{algorithmpolicy}[1][htb]{%
  \let\c@algorithm\c@algorithmpolicy
    \renewcommand{\ALG@name}{Policy}
   \begin{algorithm}[#1]%
  }{\end{algorithm}}
\makeatother

\newcommand{\nonl}{\renewcommand{\nl}{\let\nl\oldnl}}
\newcommand{\ise}{\textsc{IRSR}}

\newcommand{\ALG}{{\small\textnormal{\textsf{ALG}}}}
\newcommand{\OPT}{{\small\textnormal{\textsf{OPT}}}}
\newcommand{\OGP}{{\small\textnormal{\textsf{OGP}}}}
\newcommand{\SEGB}{{\small\textnormal{\textsf{IREGB}}}}
\newcommand{\ICSEGB}{{\small\textnormal{\textsf{BIC-IREGB}}}}

\usepackage{times}

\makeatletter
\renewcommand\normalsize{%
  \@setfontsize\normalsize{11}{13.6}%
  \abovedisplayskip 5.5\p@ \@plus2\p@ \@minus4\p@
  \abovedisplayshortskip \z@ \@plus3\p@
  \belowdisplayshortskip 6\p@ \@plus3\p@ \@minus3\p@
  \belowdisplayskip \abovedisplayskip
  \let\@listi\@listI}
\makeatother
\usepackage[a4paper, left=1.1in, right=1.1in, top=1.1in, bottom=1.1in]{geometry}

\begin{document}

\title{Near-Linear MIR Algorithms for Stochastically-Ordered Priors}


\author{%
  Gal Bahar\inst{1}\and
  Omer Ben{-}Porat\inst{1}\thanks{Corresponding author.} \and
  Kevin Leyton{-}Brown\inst{2} \and
  Moshe Tennenholtz\inst{1}
}

\authorrunning{G.\ Bahar et al.}

\institute{%
  Technion – Israel Institute of Technology, Haifa, Israel\\
  \email{bahar@campus.technion.ac.il}, \email{\{omerbp,moshet\}@technion.ac.il}
  \and
  University of British Columbia, Vancouver, Canada\\
  \email{kevinlb@cs.ubc.ca}
}

\maketitle
\setcounter{footnote}{0}
\begin{abstract}%

With the rise of online applications, recommender systems (RSs) often encounter constraints in balancing exploration and exploitation. Such constraints arise when exploration is carried out by agents whose utility must be taken into account when optimizing overall welfare. Recent work suggests that recommendations should be \emph{mechanism-informed individually rational} (MIR)~\cite{Fiduciary}. Specifically, if agents have a default arm they would use, relying on the RS should yield each agent at least the reward of the default arm, conditioned on the information available to the RS. Under the MIR constraint, striking a balance between exploration and exploitation becomes a complex planning problem. To that end, \citet{Fiduciary} propose an approximately optimal yet inefficient planning algorithm that runs in $O(2^K K^2 H^2)$, where $K$ is the number of arms and $H$ is the size of the support of the reward distributions. In this paper, we make a significant improvement for a special yet practical case, removing both the dependence on $H$ and the exponential dependence on $K$. We assume a stochastic order of the rewards (e.g., Gaussian with unit variance, Bernoulli, etc.), and devise an asymptotically optimal algorithm with a runtime of $O(K \log K)$. Our technique is based on formulating a Goal Markov Decision Process (GMDP), establishing an optimal dynamic programming procedure, and then unveiling its crux---fleshing out a simple index-based structure that facilitates efficient computation. Additionally, we present an incentive-compatible version of our algorithm.

\end{abstract}

\section{Introduction}\label{sec:intro}

Multi-armed bandits (MABs) represent a fundamental problem in online learning. In this framework, the goal is to maximize total rewards by balancing the exploitation of known high-reward arms with the exploration of potentially better ones. With the rise of online applications, MAB scenarios often involve human-facing settings, where arms or rounds represent agents whose needs must be considered. For instance, several lines of work consider content providers as the arms~\cite{hu2023incentivizing,BenPoratT18,immorlica2024clickbait}, requiring algorithms to incentivize the creation of high-quality content. In other works, beyond the exploration-exploitation dilemma, algorithms must also address fairness~\cite{MatthewKearnsMorgensternRothNIPS2016,liu2017calibrated}, long-term welfare~\cite{ben2023learning,googleOmer20}, and strategic behavior~\cite{braverman2019multi,feng2020intrinsic}, among other concerns.

A recent line of work initiated by \citet{Kremer2014} considers MAB as a mechanism design problem~\cite{nisan1999algorithmic, nisan2007algorithmic}. In such settings, a recommender system suggests actions to agents without enforcing its choices. As agents are self-interested and may avoid high-risk, high-reward arms, there is a need to \emph{incentivize exploration} to ensure accurate reward estimates for all arms. Previous works~\cite{mansour2015bayesian,mansour2020bayesian,cohen2019optimal} achieve incentive compatibility (IC) by leveraging \emph{information asymmetry}: While each agent possesses some prior information, mechanisms accumulate additional information by observing the actions and rewards of previous agents. Although IC typically serves as a benchmark for a mechanism's trustworthiness, as it guarantees that agents maximize their expected utility by following the mechanism's recommendations, mechanisms that employ information asymmetry still suffer from an inherent trust issue. Indeed,  agents who receive recommendations that are deliberately suboptimal (i.e., exploration) may feel misled. 

One intuitive remedy is to modify the mechanism beyond mere IC by promising that an agent is never worse off by following a recommendation. This brings us to the classical notion of Individual Rationality (IR), or participation constraints \cite{tadelis2013game}. IR requires that the expected utility from following the recommendation be at least as high as that of the agent's default action---that is, the action they would select if the mechanism did not exist. Yet the standard IR definition inherits the same information-asymmetry pitfall: Because the mechanism holds more information than any single agent, it can propose actions that look rational from the agent's perspective but are sub-optimal relative to the mechanism's more comprehensive information.

To address this issue, \citet{Fiduciary} propose a stronger guarantee, \emph{Mechanism-informed Individual Rationality} (MIR).\footnote{\citet{Fiduciary} label this guarantee simply ``IR.'' To avoid confusion with the conventional IR constraint, we henceforth call it mechanism-informed IR.} MIR requires that every recommendation's expected reward, \emph{based on all the information available to the mechanism}, meet or exceed the agent's default-arm reward \emph{ex-ante}. This stronger guarantee is crucial because it closes the information gap, treating the agent as if they possessed the mechanism's accumulated information. Despite the attractiveness of MIR, relying on the mechanism's accumulated information limits the ability to learn through exploration. Consequently, MIR mechanisms must adopt a careful exploration policy, maximizing their ability to explore risky arms while still satisfying the MIR guarantee by mixing low-risk and high-risk arms. 

Due to this challenge, prior work \cite{Fiduciary} faces significant practical limitations. Indeed, the runtime of the MIR mechanism of \citet{Fiduciary} is $O(2^K K^2 H^2)$, where $K$ is the number of arms and $H$ is the support size of the (discrete) reward distributions. This high runtime complexity severely limits the mechanism's practicality, particularly when dealing with many arms or continuous reward distributions. For the latter, discretization approaches are either computationally infeasible or yield suboptimal rewards.

In this paper, we revisit MIR constraints in Bayesian MABs with static rewards.  For priors that are \emph{stochastically ordered}, we design $\SEGB$: An index-based algorithm that is MIR by construction and is asymptotically optimal among all MIR algorithms. Crucially, $\SEGB$ runs in $O(K\log K)$ time---reducing the exponential $2^{K}K^{2}$ dependence on the number of arms down to near-linear and removing the $H^2$ term entirely, so it works with continuous reward distributions. Finally, we use $\SEGB$ as a black-box to develop a MIR and IC algorithm.

\subsection{Our Contribution}
We consider a MAB setting with a Bayesian prior and static rewards; that is, the reward of each arm is initially unknown and is realized only once.\footnote{This assumption allows us to focus on the novel mechanism-informed IR viewpoint proposed in this paper. A similar approach is taken by several early works on incentivizing exploration \cite{Bahar2016,Fiduciary}, and some have strictly stronger assumptions, i.e., $\{-1,1\}$ rewards \cite{cohen2019optimal}. Despite the limitation, our problem remains technically non-trivial.} Some arms may be risky, having negative expected rewards but also offering high potential rewards. Additionally, we assume that there is a \textit{default arm} that agents would pick if they did not use the algorithm. In each round, the algorithm picks a distribution over the arms, termed a \textit{portfolio}. The algorithm then samples and selects an arm based on the portfolio's weights. We aim to maximize social welfare, defined as the sum of rewards.

Without any further constraints, maximizing welfare is straightforward: Pull each arm once, then repeatedly choose the best one. Due to our static rewards assumption, one round of exploration per arm is enough to reveal the highest reward. However, we restrict the algorithm to select only MIR portfolios. A portfolio is MIR if its expected reward is at least as good as that of the default arm \emph{conditioned on the information available to the algorithm} (no information asymmetry). Since rewards are realized only once, careful planning is crucial. Arms with a high expected value extend the set of MIR portfolios, enabling the exploration of a priori inferior arms.

Our main technical contribution is an asymptotically optimal MIR algorithm under a stochastic order assumption on rewards (see Definition~\ref{def:bayesian safety}).  Stochastic order is common in practical scenarios, e.g., Bernoulli rewards or Gaussian reward distributions with a common variance. We introduce an auxiliary Goal Markov Decision Process (GMDP)~\cite{barto1995learning}. In our GMDP, the action sets are convex polytopes of all MIR portfolios that mix a priori superior arms (better than the default arm) with a priori inferior arms. Despite the complex action space, we reveal a simple index-based optimal policy, akin to the seminal result of \citet{weitzman1978optimal} for Pandora's box \cite{gergatsouli2024weitzman, atsidakou2024contextual,boodaghians2020pandora, berger2023pandora}. We then use the optimal GMDP policy to develop $\SEGB$ (see Algorithm~\ref{alg:alg of pi}), a MIR and asymptotically optimal algorithm. Informally stated,
\begin{theorem}[Informal version of Corollary~\ref{cor: alg is asym opt}]
Given a problem instance with a stochastic order over rewards, $\SEGB$ has a runtime of $O(K \log K)$, where $K$ is the number of arms, and asymptotically achieves optimal welfare among all MIR algorithms.
\end{theorem}
We analyze the algorithm's convergence rate and demonstrate that it is fast in expectation. Additionally, we assume agents are strategic and aim to design an IC and MIR mechanism. To achieve this, we apply the hidden exploration technique of \citet{mansour2015bayesian} to modify $\SEGB$, making it both MIR and IC.

\subsection{Related Work}\label{subsec:related}
Mechanism-informed individual rationality can be viewed as a safety constraint, connecting to various strands of research on safe exploration. For instance, Bandits with knapsack~\cite{badanidiyuru2013bandits} address a MAB problem with a global budget, aiming to maximize total rewards before exhausting resources. Another approach focuses on stage-wise safety, ensuring that regret performance remains above a threshold set by a baseline strategy in each round~\cite{kazerouni2017conservative,wu2016conservative}. Notably, in these lines of work, constraints apply to cumulative resource consumption or reward throughout the algorithm's run. Conversely, \citet{amani2019linear} apply a reward constraint in each round. Their work leaves the set of safe decisions uncertain due to the inherent uncertainty in the learning process, aiming to minimize regret while learning the safe decision set. Similar to \citet{amani2019linear}, we also consider a stage-wise constraint. It is worth mentioning that safety has been studied beyond MABs~\cite{wachi2020safe,Moldovan:2012}. For example, \citet{Moldovan:2012} introduce  an algorithm enabling safe exploration in Markov Decision Processes to avoid fatal absorbing states. 

Closely related to our work is the research on incentivizing exploration, pioneered by \citet{KremerMP13}, motivated by using a bandit-like setting for recommendations. Since users are selfish, algorithms must incentivize exploration (see \cite{slivkins2019introduction} for an introduction and overview). Subsequent works consider regret minimization \cite{mansour2015bayesian,mansour2020bayesian}, heterogeneous agents \cite{chen18a,immorlica2019bayesian}, social networks \citep{Bahar2016,bahar2019recommendation}, and extensions to Markov Decision Processes \cite{simchowitz2024exploration}. Similarly to several previous works \cite{Fiduciary,bahar2019recommendation,cohen2019optimal,KremerMP13}, we assume each arm is associated with a fixed value, which is initially unknown and sampled from a known distribution. Given this, our problem becomes the careful planning of MIR exploration.

\citet{Fiduciary} is the work most closely related to ours within the incentivizing exploration line of research. We inherit their model and also aim to develop an optimal MIR mechanism. \citet{Fiduciary} assume rewards are integers in a bounded set of size $H$, construct an auxiliary GMDP model, and devise a dynamic programming-based planning policy with a runtime of $O(2^K K^2 H^2)$. We adopt the same GMDP approach (Section~\ref{sec:infinite}) but aim to reduce the exponential dependence on $K$ and remove the dependence on $H$. We achieve a significant improvement of $O(K \log K)$ for rewards with stochastic order, a special yet common case. Importantly, the techniques we use, such as the Equivalence Lemma (Lemma~\ref{lemma:equivalence body}), are novel and were not examined by \citet{Fiduciary}. 

\section{Problem Statement}\label{sec:problem statement}
This section formally defines the Mechanism-informed Individually Rational Sequential Recommendation problem ($\ise$ for shorthand). We consider a set $A$ of $K$ arms, $A=\{a_1,\dots a_K\}$. The reward of arm $a_i$ is a random variable $X(a_i)$, and $(X(a_i))_{i=1}^K$ are mutually independent. The rewards are static, i.e., they are realized only once but initially unknown. If the reward $X(a_i)$ is realized before round $t$ and is hence known, we use $x(a_i)$ to denote its value. We denote by $\mu({a_i})$ the expected value of $X(a_i)$, i.e., $\mu(a_i)=\E\left[X(a_i)\right]$. We augment the set of arms with a \textit{default arm}, which we denote by $a_0$. For simplicity, unless stated otherwise, we assume that the reward of the default arm $a_0$ is always 0; hence, $\mu(a_0)=0$.\footnote{The selection of zero as the threshold is arbitrary; all our results hold for any bounded distribution for $X(a_0)$ with minor modifications.} Consequently, we let $A^+ = A\cup\{a_0\}$. As will become apparent later on, it is convenient to distinguish arms with positive expected rewards from those with negative expected rewards. To that end, we let 
$\above(A)=\{a_i\in A:\mu(a_i) > 0 \}$. Analogously, $\below(A)=\{a_i\in A:\mu(a_i)< 0 \}$. We assume for simplicity that there are no arms with an expected reward of 0. 

There are $T$ agents arriving sequentially. We denote by $a^t$ the action of the agent arriving at round $t$. The reward of the agent arriving at round $t$ is denoted by $r^t$, and is a function of the arm she chooses. For instance, by selecting arm $a_l \in A$, the agent arriving at round $t$ obtains $r^t = r^t(a_l)=X(a_l)$. 
Agents are fully aware of the distribution of $(X(a_i))_{i=1}^K$, and each agent cares about her own reward, which she wants to maximize. 

A mechanism is a recommender system that interacts with agents. The input to the mechanism at round $t$ consists of all the information acquired up to round $t$ by the previous $t-1$ agents, which we denote by $\mI_t$. Namely, $\mI_1$ encodes the prior information only, $\mI_2$ encodes both the prior information and the reward of the arm selected at round 1, and so on. The output of the mechanism is a distribution over arms. Formally, a randomized mechanism is a mapping $ \bigcup_{t=1}^T \left(A^+\times \R_+ \right)^{t-1} \rightarrow \Delta(A^+)$, where $\Delta(A^+)$ is the set of probability distributions over the elements of $A^+$. 

Notice that a randomized mechanism selects a distribution of arms in every round, which we term \textit{portfolios}. A portfolio is an element from $\Delta(A^+)$---it is a distribution over the arms and the default arm. Whenever the mechanism picks a portfolio $\bl p \in \Delta(A^+)$, Nature (i.e., a third-party) flips coins according to $\bl p$ to realize one arm from $A^+$. We let $\bl p^t$ denote the portfolio the mechanism selects at round $t$, and $m^t=m^t(\mI_t)$ the recommended arm, i.e., $m^t\sim \bl p^t$. Using this formulation, the expected reward at time $t$ of an agent choosing $m^t$ is determined by Nature's coin flips and the randomness of the rewards. Formally, 
\[
\E[r^t(m^t)\mid \mI_t]=\sum_{a_i \in A^+} \bl p^t(a_i)\E\left[X(a_i)\mid \mI_t\right].
\]
We consider both \textit{non-strategic agents} and \textit{strategic agents} schemes. In the former, agents always follow the recommendation of the mechanism, namely $a^t=m^t$. Note that this is a standard underlying assumption of classical multi-armed Bandit models. In the latter, the mechanism recommends actions to agents but cannot compel them to follow these recommendations, potentially resulting in $a^t \neq m^t$. In this scheme, we say that a mechanism is Bayesian incentive compatible (BIC) if following its recommendations is a dominant strategy in the Bayesian sense. That is, when given a recommendation, an agent's best response is to follow the recommended arm. Formally,
\begin{definition}[Bayesian Incentive Compatibility]\label{def:ic}
A mechanism is Bayesian incentive compatible (BIC) if $\forall t \in \{1,\ldots,T\}$, for every information $\mI_t$  and for all action pairs $a_l,a_i \in A$,
\begin{equation}\label{eq: ic const}
\E[X(a_l)-X(a_i)\mid m^t(\mI_t)=a_l]\geq   0.
\end{equation}
\end{definition}
We stress that every agent can decide which arm to take \textit{after seeing the realized recommended arm $m^t$} and not the portfolio $\bl p^t$. Unless stated otherwise, we refer to the non-strategic agents scheme. We treat the strategic agents scheme in Section~\ref{sec:ic body}.

\paragraph{Mechanism-Informed Individual Rationality}
In this paper, we focus on \textit{ex-ante mechanism-informed individually rational} recommendations, abbreviated as MIR. A portfolio $\bl p^t$ is MIR at time $t$ if, given the information the mechanism acquired up to round $t$, its expected reward is greater than or equal to zero, which is the reward of the default arm. Formally,
\begin{definition}[Ex-ante Mechanism-Informed Individual Rationality (MIR)]\label{def:bayesian safety}
A portfolio $\bl p$ is MIR with respect to information $\mI$ if
\begin{equation}\label{ineq:IR def}
\E_{a\sim \bl p}[X(a) \mid \mI]=\sum_{a_i\in A}\bl p (a_i)\E\left[X(a_i)\mid \mI\right]\geq \E\left[X(a_0)\mid \mI\right]= 0.
\end{equation}
\end{definition}
We stress that Definition~\ref{def:bayesian safety} implies only ex-ante MIR, as Nature may realize an arm with a negative expectation; nevertheless, we shall continue referring to it as MIR. If a mechanism selects MIR portfolios for every possible information $\mI$, we say the mechanism is MIR. By focusing on MIR mechanisms, we ensure that every agent will receive at least the reward of the default arm in expectation, independent of the other agents. The definitions of MIR and BIC are orthogonal and significantly differ in the events they condition on. The BIC Inequality~\eqref{eq: ic const} conditions on the agent's information, which includes the prior information and the recommendation of the mechanism. In contrast, the MIR Inequality~\eqref{ineq:IR def} conditions on all the information the mechanism has acquired. Thereby, MIR could be viewed as a standard for how much and in what way the system may leverage its informational advantage.

We note that forbidding the mechanism from exploring arms with a negative expected reward (equivalent to demanding \emph{ex-post} MIR) trivializes the problem as it prevents exploring any such arms. The MIR constraint is a compromise: It allows learning about a priori inferior arms, but only when they are mixed with a priori superior arms.

\paragraph{Social Welfare} When agents follow the mechanism, we denote the \textit{social welfare} achieved by a mechanism $\ALG$ by $\mU_T(\ALG)=\E\left[\frac{1}{T}\sum_{t=1}^T r^t\right]$. The global objective of the mechanism is to maximize social welfare, subject to selecting MIR portfolios in every round. If agents are strategic, we narrow our attention to BIC mechanisms only, which justifies using the same formulation of social welfare.

To conclude, we represent an instance of the $\ise$ problem by the tuple 
\[
\tupbracket{K, A, (X(a_i))_i, (\mu(a_i))_i}.
\]
Notice that the horizon $T$ is not part of the description since we often discuss a particular instance with varying $T$. When the instance is known from the context, we denote the highest possible social welfare achieved by any algorithm by $\OPT_T$, where the subscript emphasizes the dependence on the number of rounds $T$. Additionally, we let $\OPT_{\infty}$  denote $\lim_{T\rightarrow \infty} \OPT_T$.

\subsection{Example}
Before we proceed, we illustrate our setting and notation with an example.
\begin{example}\label{example with normal}
Let $K=4$, and let{\setmuskip{\thickmuskip}{0mu}
$X(a_1)\sim N(2,1)$, $X(a_2)\sim N(1,1)$, $X(a_3)\sim N(-1,1)$, and $X(a_4)\sim N(-2,1)$, where $N(\mu,\sigma^2)$ denotes the normal distribution with a mean $\mu$ and a variance $\sigma^2$. 
}%
Notice that $\above(A)=\{a_1,a_2\}$ and  $\below(A)=\{a_3,a_4\}$. In the first round, the mechanism can select, for example, a portfolio that comprises $a_1$ w.p. 1. It can also select the portfolio that mixes $a_1$ w.p. $\frac{1}{2}$ with $a_2$ w.p. $\frac{1}{2}$, and infinitely many other MIR portfolios.

To demonstrate the technical difficulty of maximizing social welfare, assume that $X(a_1),X(a_2) <0$ and $X(a_3),X(a_4)>0$ (this information is \textit{not} available to the mechanism). Indeed, this occurs with positive non-negligible probability under the distributional assumptions of this example. Consider the case where the mechanism selects the portfolio $a_1$ w.p. 1 in the first round and $a_2$ w.p. 1 in the second round. Notice that these portfolios are MIR: Given the information $\mI_1$ in the first round, selecting $a_1$ w.p. 1 is MIR, and regardless of $x(a_1)$, selecting $a_2$ w.p. 1 in the second round is also MIR. After observing that $a_1$ and $a_2$ have negative rewards (since we momentarily assume $X(a_1)<0$ and $X(a_2) <0$), the only MIR portfolio in subsequent rounds is selecting the default arm $a_0$ w.p.~1.

However, there are better ways to act in the first round. Arms with positive expected rewards are essential for enabling the exploration of arms with negative expected rewards. Consider $\bl p^1$ such that $\bl p^1(a_1)=\frac{1}{3},\bl p^1(a_3)=\frac{2}{3}$ and $\bl p^1(a_2)=\bl p^1(a_4)=0$. Notice that $\bl p^1$ is MIR w.r.t. $\mI_1$, since $\bl p^1(a_1)\mu(a_1)+\bl p^1(a_3)\mu(a_3)=0$. If the mechanism selects $\bl p^1$ in the first round, it will already discover the positive reward of arm $a_3$ in the first round w.p. $\frac{2}{3}$. It is thus evident that selecting $\bl p^1$ in the first round results in higher social welfare than selecting $a_1$ w.p. 1.  
\end{example}

This example illustrates that any mechanism seeking to maximize social welfare should leverage $a_1$ and $a_2$ with their a priori positive rewards to explore the a priori inferior arms $a_3$ and~$a_4$.

\subsection{Stochastic Order}
Our main technical contribution, which we present in the next section, assumes that some of the rewards are stochastically ordered. We say that a random variable $X$ stochastically dominates (or, has first-order stochastic dominance over) a random variable $Y$ if for every $x\in (-\infty ,\infty ),$ $\Pr(X\geq x)\geq \Pr(Y\geq x)$. Observe that this dominance immediately implies that $\E\left[X\right] \geq \E\left[Y\right]$.
\begin{assumption}\label{assumption:dominance}
The rewards of the arms in $\below(A)$ are stochastically ordered.
\end{assumption}
Notice that this assumption applies only to a priori inferior arms, i.e., the arms in $\below(A)$, and not to all arms. There are many natural cases for this assumption: Unit-variance (or any fixed variance) Gaussian, as in Example~\ref{example with normal}, Bernoulli, log-normal, truncated normal, etc. When a formal statement relies on Assumption \ref{assumption:dominance}, we mention it explicitly.

\section{An Auxiliary Goal Markov Decision Process}\label{sec:infinite}
In this section, we present an auxiliary GMDP problem that helps us focus on the optimal exploration order. We describe its usefulness in Subsection~\ref{subsec:bernoulli trials} and its formal representation in Subsection~\ref{subsec:aux GMDP}. We devote Subsection~\ref{subsec:p valid inefficient} to taking the first step towards efficient computation, as we narrow down the action space. Then, in Subsection~\ref{subsec:optimal GMDP policy}, we present our main technical contribution, which is featured in Policy~\ref{policy:pi star}: An optimal index-based GMDP policy. Subsection~\ref{subsec:stopchastic} further argues for the non-triviality of our results, demonstrating that even stochastically ordered rewards present significant challenges.

\subsection{Motivation for GMDP: Bernoulli Trials}\label{subsec:bernoulli trials}
Assume momentarily that there are infinitely many rounds. In this case, the following interesting property arises.
\begin{observation}\label{obs:eventually will explore}
Assume that during the execution we discover that $X(a_i)>0$ for some arm $a_i\in A$. We can then use $a_i$ to explore all other arms in finite time.
\end{observation}
We call mixing an observed arm with a positive realized reward and an unobserved arm with a negative expected reward a \emph{Bernoulli Trial}. Observation~\ref{obs:eventually will explore} suggests that if we discover an arm with a positive reward, we can use Bernoulli trials to reveal all arms in finite time almost surely. 

To illustrate, recall Example \ref{example with normal}. Assume that we are in the third round, have already explored $a_1$ and $a_2$, and discovered that $X(a_1)=x(a_1)>0$ and that $X(a_2)<0$. Furthermore, we have not yet explored $a_3$ and $a_4$ (one of which may still hide the highest realized reward among all arms). Consider the portfolio 
\begin{align*}
\bl p(a) =
\begin{cases}
\frac{-\mu(a_3)}{x(a_1)-\mu(a_3)} & \textnormal{if } a=a_1\\
\frac{x(a_1)}{x(a_1)-\mu(a_3)} & \textnormal{if } a=a_3\\
0 & \textnormal{otherwise}
\end{cases}.
\end{align*}
Observe that
\[
\sum_{a_i \in A^+} \bl p(a_i)\E\left[X(a_i)\mid \mI_3\right]
=x(a_1)\cdot \frac{-\mu(a_3)}{x(a_1)-\mu(a_3)} + \mu(a_3)\cdot  \frac{x(a_1)}{x(a_1)-\mu(a_3)}   =0;
\]
thus, $\bl p$ is MIR. If we pick this portfolio, Nature flips coins to choose either $a_1$ or $a_3$, with the latter being selected with positive probability. This Bernoulli trial might result in selecting $a_1$, but we can repeat it until Nature picks $a_3$. The number of rounds required to explore $a_3$ follows the Geometric distribution with a success probability of $\nicefrac{x(a_1)}{x(a_1)-\mu(a_3)}$ in each Bernoulli trial. Indeed, by executing Bernoulli trials until the first success, we can guarantee exploring $a_3$. Zooming out from Example~\ref{example with normal} to the general case, Observation~\ref{obs:eventually will explore} suggests that once a positive reward is realized, we can execute enough Bernoulli trials to explore all the arms in finite time.

Observation \ref{obs:eventually will explore} calls for modeling that abstracts the setting once a positive reward is realized. More specifically, we can focus on the case where all unexplored arms are revealed instantly after a positive reward is discovered. To that end, we analyze the induced constrained Goal Markov Decision Process (GMDP), which we present in the next subsection. Our goal is to find the optimal policy for this GMDP, and then translate it to an asymptotically optimal algorithm for the corresponding $\ise$ instance.\footnote{We keep the terms mechanism and algorithm for solutions to $\ise$ and use the term policy for solutions of the GMDP.}

\subsection{Constructing the GMDP}\label{subsec:aux GMDP}
We construct the GMDP as follows:
\begin{itemize}
    \item Every state is characterized by the set of unobserved arms $s \subseteq A$, and we denote the set of all states by $\mS=2^A$. The initial state is $s_0=A$. 
    \item In every state $s$, the feasible actions are the set of MIR portfolios w.r.t. the prior information, i.e., from
    $
    \safe(s)=\left\{\bl p \in \Delta(s) : \sum_{a\in s}\bl p(a)\mu(a) \geq 0 \right\},
    $
    where $\Delta(s)$ is the set of all distributions over the elements of $s$. If $\safe(s)$ is empty, then we say that $s$ is \textit{terminal}. 
    \item Given a non-terminal state $s$ and an action (portfolio) $\bl p$, the transition probability is given by
    \[
    \Pr(s'\mid s, \bl p)=
    \begin{cases}
\bl p(a) & \textnormal{if $s' = s \setminus \{a\}$}\\
0 & \textnormal{otherwise}
\end{cases}.
    \]
    Namely, if the realized arm is $a$, the GMDP transitions to the state $s\setminus \{ a\}$. In particular, we can reach a terminal state if we ran out of arms with a positive expected value, i.e., $\above(s)=\emptyset$, or if we have explored all arms, i.e., $s=\emptyset$.
    \item Rewards are obtained in terminal states solely. The reward of a terminal state $s$ is
\[
R(s) \defeq
\begin{cases}
\max_{a\in A} X(a) & \textnormal{if $ \safe(s)=\emptyset$ and $\max_{a'\in A\setminus s} X(a')>0$}\\
0 & \textnormal{otherwise}
\end{cases}.
\]
In other words, once we reach a terminal state $s$, we sample $X(a)$ for every $a\in A$. If we reached $s$, then we have explored all arms $A\setminus s$. Following our intuition of Bernoulli trials, if at least one of $(X(a))_{a\in A\setminus s}$ is positive, we can use that arm to explore all other arms in finite time. Therefore, the reward we get is the maximum reward among \textit{all} arms $\max_{a\in A} X(a)$, even those we have not explored. 
\end{itemize}

A (stationary) \textit{policy} is a mapping from states to portfolios. Due to standard arguments, which are further elaborated in {\ifnum\Includeappendix=0{the appendix}\else{Section \ref{sec:thm1 outline}}\fi}, there exists an optimal policy that is also stationary. Hence, we limit our attention to stationary policies only.  Let $W(\pi,s)$ be the state-value function: The expected reward of a policy $\pi$ starting at the initial state $s$. Namely, 
\begin{align}
W(\pi,s) = 
\begin{cases}
R(s) & \textnormal{if }\safe(s)=\emptyset\\
\sum_{a\in s}\Pr(s\setminus \{a\}\mid s, \pi(s)) W(\pi,s\setminus \{a\}) & \textnormal{otherwise}
\end{cases}.
\end{align}
Moreover, let $W^{\star}(s)$ denote the highest possible reward of any policy, i.e., $W^{\star}(s)=\sup_{\pi}W(\pi,s)$; hence, $W^{\star}(A)$ symbolizes the optimal expected reward when starting from $s_0=A$. Later on, in Section~\ref{sec:policy to algorithm}, we show the connection between the optimal $\ise$ solution and this GMDP. For now, we focus on finding an optimal GMDP policy.

\subsection{$\mP$-valid Policies and an Inefficient Solution via Dynamic Programming}\label{subsec:p valid inefficient}

Next, we take a step toward finding an optimal policy by revealing the neat structure of the problem. To that end, we highlight the following family of portfolios that mix at most two arms. Consider a pair of arms, $a_i\in \above(A)$, and $a_j\in \below(A)$. The portfolio
\begin{align}\label{eq:blp from body}
\bl p_{i,j}(a) \defeq 
\begin{cases}
\frac{-\mu(a_j)}{\mu(a_i)-\mu(a_j)} & \textnormal{if } a=a_i\\
\frac{\mu(a_i)}{\mu(a_i)-\mu(a_j)} & \textnormal{if } a=a_j\\
0 & \textnormal{otherwise}
\end{cases}
\end{align}
mixes $a_i$ and $a_j$ while maximizing the probability of exploring $a_j$ (the a priori inferior arm). The reader can verify that $\bl p_{i,j}$ is indeed MIR, yielding an expected value of precisely zero. For completeness, for every $a_i \in \above(A)$, we also define $\bl p_{i,i}$ as a deterministic selection of $a_i$, e.g., $\bl p_{i,i}(a)=1$ if $a=a_i$, and zero otherwise. Next, we define $\mP,\mP'$ such that
\[
\mP\defeq \{\bl p_{i,j}\mid a_i\in \above(A), a_j\in \below(A) \}, \qquad \mP'\defeq\{\bl p_{i,i}\mid a_i\in \above(A) \}.
\]
We use the above to focus on a narrow class of portfolios, which we call \emph{$\mP$-valid}.
\begin{definition}[$\mP$-valid Portfolio]\label{def:p valid}
A portfolio $\bl p$ is \emph{$\mP$-valid} with respect to a state $s\in \mS$ if
\begin{itemize}
\item $\bl p \in \safe(s)$, and
\item If $\below(s)\neq \emptyset$, then $\bl p\in \mP$; else, if $\below(s)= \emptyset$, then $\bl p\in \mP'$.
\end{itemize} 
\end{definition}
We say a policy is \emph{$\mP$-valid} if, for all states $s\in \mS$, $\pi(s)$ is a \emph{$\mP$-valid} portfolio. Notice that $\mP$-valid policies are MIR by definition. Furthermore, they mix at most two actions from $A$ in each state; hence, there are  $O(K^2)$ actions $\mP$-valid policies can take in every state. This is in sharp contrast to the set of actions $\safe(s)$ for a state $s$, which is generally a convex polytope with infinitely many actions. 

Due to the linearity of the value function $W(\pi,s)$ in $\pi(s)$, the GMDP exhibits a nice structural property, as captured by the following Proposition \ref{prop:main optimal p valid}. 
\begin{proposition}\label{prop:main optimal p valid}
Fix any arbitrary $\ise$ instance and observe its induced GMDP instance. There exists an optimal policy which is $\mP$-valid.
\end{proposition}
The proof of this proposition appears in {\ifnum\Includeappendix=0{the appendix}\else{Section~\ref{sec:aux}}\fi}. While the GMDP is defined using a continuous action space in each state (the convex polytope $\safe(s)$), Proposition~\ref{prop:main optimal p valid} allows us to narrow down the search to $\mP$-valid portfolios. Particularly, we can find an optimal policy with dynamic programming using a bottom-up approach in time $O(2^K K^2)$, as we explain next.

For every $s \in \mS$, compute $W^{\star}(s)$ by exhaustively searching over $\mP$-valid portfolios and using previously obtained $W^{\star}(s')$ for $s' \subset s$. Since $s$ leads to states of the form $s/{a}$ for $a\in s$, we can compute the expected reward for each $\mP$-valid portfolio and pick the best portfolio. There are $2^K$ states and at most $O(K^2)$ $\mP$-valid portfolios in every state, each taking $O(1)$ computations to assess; therefore, this by itself guarantees finding an optimal policy in time $O(2^K K^2)$. Noticeably, this runtime could be intractable for large $K$. 
In the next subsection, we show how to substantially reduce the computation.

\subsection{Efficiently Computing an Optimal GMDP Policy}\label{subsec:optimal GMDP policy}

\begin{algorithmpolicy}[H]
\caption{Optimal GMDP Policy ($\OGP$)\label{policy:pi star}}
\begin{algorithmic}[1]
\REQUIRE a state $s\subseteq A$
\ENSURE a $\mP$-valid portfolio or $\emptyset$
\IF{$s$ is terminal \label{policy:if terminal}} {
\STATE \textbf{return} $\emptyset$ \label{policy:return empty}
}
\ELSE\label{policy:non terminal}{
\STATE pick any arbitrary $a_i \in \above(s)$\label{policy:pick arbitrary}
\IF{$\below(s)=\emptyset$\label{policy:if no below}} {
		\RETURN $\bl p_{i,i}$ \label{policy:return double above}
}
\ELSE\label{policy:if has below}{
\STATE pick $a_{j^\star}\in \argmax_{a_j \in \below(s)} \mu(a_j)$\label{policy:pick below}
		\RETURN $\bl p_{i,j^\star}$\label{policy:return mix}
}
\ENDIF
}
\ENDIF
\end{algorithmic}
\end{algorithmpolicy}

We start this subsection by presenting an optimal GMDP policy (hereinafter $\OGP$ for shorthand), which we formalize via Policy~\ref{policy:pi star}, and argue for its optimality later. Given a state $s$, $\OGP(s)$ operates as follows. If $s$ is terminal, it returns the empty set (Lines~\ref{policy:if terminal}--\ref{policy:return empty}). Otherwise, if $s$ is non-terminal, we enter the ``else'' clause in Line~\ref{policy:non terminal}, and pick any arbitrary arm $a_i$ from $\above(s)$. Then, we have two cases. If $\below(s)$ is empty, $\OGP$ returns $\bl p_{i,i}$ (Line~\ref{policy:return double above}). Else, it picks the best arm from $\below(s)$ in terms of expected reward, which we denote by $a_{j^\star}$ (Line~\ref{policy:pick below}), and returns $\bl p_{i,j^\star}$ (Line~\ref{policy:return mix}). Overall, $\OGP$ returns $\emptyset$ if $s$ is terminal or a $\mP$-valid portfolio if $s$ is non-terminal.
\begin{theorem}\label{thm:optimal policy}
Fix any arbitrary $\ise$ instance satisfying Assumption~\ref{assumption:dominance} and observe its induced GMDP instance. For any state $s\subseteq A$, $\OGP$ prescribed in Policy~\ref{policy:pi star} satisfies $W(\OGP,A)=W^{\star}(A)$.
\end{theorem}
\begin{proof}[\textnormal{\textbf{Proof Overview of Theorem \ref{thm:optimal policy}}}]
We defer the full proof to {\ifnum\Includeappendix=0{the appendix}\else{Section~\ref{sec:thm1 outline}}\fi} and outline it below.
To prove this theorem, we focus on the dynamic programming explained at the end of the previous subsection. We unveil its crux and flesh out the index structure $\OGP$ uses. A crucial ingredient of the analysis is the probability of reaching the empty state, representing the case in which we explored all arms. More formally, for every state $s\in \mS$ and a policy $\pi$, we denote by $Q(\pi,s)$ the probability starting at $s\subseteq A$, following the policy $\pi$ and reaching the empty state. Note that, similarly to $W$ the state-value function, $Q$ is defined recursively: Namely, if $\pi(s)=\bl p_{i,j}$ for a non-terminal state $s$, then \[ Q(\pi,s)=\bl p_{i,j}(a_i)Q(\pi,s\setminus\{a_i\})+\bl p_{i,j}(a_j)Q(\pi,s\setminus\{a_j\}). \]
We prove a rather surprising feature of $Q$, as captured in the following lemma.
\begin{lemma}[Equivalence Lemma]\label{lemma:equivalence body}
For every two $\mP$-valid policies $\pi,\rho$ and every state $s\in \mS$, it holds that $Q(\pi,s)=Q(\rho,s)$.
\end{lemma}
The Equivalence Lemma asserts that $Q$ is policy-independent, a property that significantly simplifies the analysis. The proof of this lemma is intricate and relies on a careful factorization of the elements of $Q$.  Interestingly, the arguments we have mentioned so far, including the Equivalence Lemma, are independent of Assumption \ref{assumption:dominance}.

Equipped with the Equivalence Lemma, we use the recursive structure of $Q$ and the reward function $W$ to prove the optimality of $\OGP$ inductively over the sizes of $\above(s)$ and $\below(s)$. This step makes use of further insights into the dynamic programming procedure, such as monotonicity in the set of available arms.
\end{proof}

We exemplify $\OGP$ for the induced GMDP of Example~\ref{example with normal} in Figure~\ref{fig:tree small}. 
Recall that 
$X(a_1)\sim N(2,1)$, $X(a_2)\sim N(1,1)$, $X(a_3)\sim N(-1,1)$, and $X(a_4)\sim N(-2,1)$. Consequently, $\above(A)=\{a_1,a_2\}$, $\below(A)=\{a_3,a_4\}$, and $\mu(a_3)=-1 > -2=\mu(a_4)$. Every node is a state---the root is $s_0=A$. In $s_0$, we select the portfolio $\bl p_{1,3}$, which mixes $a_1\in \above(A)$ and $a_3 \in \below(A)$ (selecting $\bl p_{2,3}$ would also be optimal, since Line~\ref{policy:pick arbitrary} allows us to pick any $a_i\in\above(A)$). The split follows from Nature's coin flips: Left if the realized action is $a_1$ (w.p. $\bl p_{1,3}(a_1)$), and right if the action is $a_3$ (w.p. $\bl p_{1,3}(a_3)$). The leaves $\{a_3,a_4\}$, $\{a_4\}$, and $\{\}$ are the terminal states. 

\begin{figure}[t]
\centering
    \includegraphics[scale=0.8]{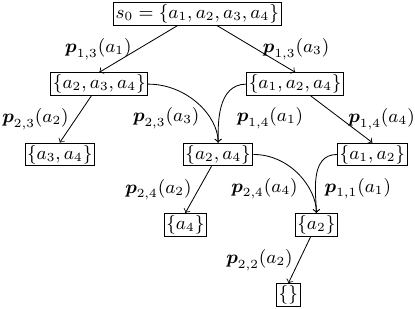}
    \small
    \caption{Illustration of $\OGP$ for Example~\ref{example with normal}.\label{fig:tree small}}
\end{figure}

\subsection{Stochastic Order and Non-triviality}\label{subsec:stopchastic}
Before leveraging $\OGP$ to an optimal $\ise$ algorithm, it is important to describe why Theorem~\ref{thm:optimal policy} is non-trivial to prove, even with the limiting Assumption~\ref{assumption:dominance}. Consider a state $s\in\mS$ and three arms $a_i,a_j ,a_{\tilde j}$ such that
\[
a_i\in \above(s), a_j \in \argmax_{a_j \in \below(s)} \mu(a_{j'}), \mu(a_{\tilde j})< \mu(a_j).
\]
Namely, $a_i$ has a positive expected value, while $a_j, a_{\tilde j}$ have negative expected values and $a_j$ stochastically dominates $a_{\tilde j}$. The action $\bl p_{i,j}$, which mixes $a_i$ with the minimal element $a_j$ according to the stochastic order on $\below(s)$, is weakly superior to  $\bl p_{i,{\tilde{j}}}$ if 

{
\thinmuskip=.2mu
\medmuskip=0mu plus .2mu minus .2mu
\thickmuskip=1mu plus 1mu
\[
\Big(\bl p_{i,j}(j)W^{\star}(s\setminus \{a_j\})+\bl p_{i,j}(i)W^{\star}(s\setminus\{a_i\})\Big)-\left(\bl p_{i,{\tilde j}}({\tilde j})W^{\star}(s\setminus \{a_{\tilde j}\})+\bl p_{i,{\tilde j}}(i)W^{\star}(s\setminus\{a_i\})\right) \geq 0.
\]}%
Rearranging,
\begin{align}\label{eq:why hard body}
&\bl p_{i,j}(j)W^{\star}(s\setminus \{a_j\})-\bl p_{i,{\tilde j}}({\tilde j})W^{\star}(s\setminus \{a_{\tilde j}\})+\left(\bl p_{i,j}(i)-\bl p_{i,{\tilde j}}(i)\right)W^{\star}(s\setminus\{a_i\}) \geq 0.
\end{align}%
By our selection of $a_j,a_{\tilde j}$, we know that $\bl p_{i,j}(i)-\bl p_{i,{\tilde j}}(i) \leq 0$; hence, the third term is non-positive. Interestingly, as we show in~{\ifnum\Includeappendix=0{the appendix}\else{Claim \ref{claim:ass is not for W} in Section \ref{sec:aux}}\fi}, stochastic order does not imply that $W^{\star}(s\setminus \{a_j\}) \geq W^{\star}(s\setminus \{a_{\tilde j}\})$; thus, it is unclear that the expression
{
\begin{align*}
&\bl p_{i,j}(j)W^{\star}(s\setminus \{a_j\})-\bl p_{i,{\tilde j}}({\tilde j})W^{\star}(s\setminus \{a_{\tilde j}\}),
\end{align*}}%
which appears in Inequality \eqref{eq:why hard body}, is non-negative. Therefore, we cannot prove Inequality  \eqref{eq:why hard body} without revealing the structure of $W^{\star}$, even when Assumption \ref{assumption:dominance} holds.

\section{From GMDP to an Approximately Optimal $\ise$ Algorithm}\label{sec:policy to algorithm}
In this section, we leverage $\OGP$, which is the optimal policy for the GMDP presented in the previous section, to an approximately optimal algorithm, which we call $\SEGB$ and is outlined in Algorithm~\ref{alg:alg of pi}. We start by addressing the infinite number of rounds, and then present a finite-time analysis. 

\subsection{The Case of $T\rightarrow \infty$}
Recall that the GMDP of Section~\ref{sec:infinite} mimics the reward any algorithm could obtain, since it abstracts the need for Bernoulli trials. In other words, 
\begin{observation}\label{obs: U leq W*}
For any algorithm $\ALG$, it holds that $\lim_{T \rightarrow \infty }\mU_T(\ALG) \leq W^{\star}(A)$.
\end{observation}

To see this, note that for any $\ALG$, there exists a modified algorithm that is deterministic and uses portfolios of at most two arms. Then, there is a modified algorithm that first performs exploration (mixed portfolios of unexplored arms), then Bernoulli trials, and finally exploitation. Those modifications can only increase $\ALG$'s social welfare reward, with an asymptotic value of $W^{\star}(A)$.

Furthermore, given any policy $\pi$ for the GMDP, we can construct a \emph{wrapper algorithm} $\ALG(\pi)$ for the corresponding $\ise$ instance by wrapping $\pi$ with Bernoulli trials. At the beginning, $\ALG(\pi)$ picks portfolios according to $\pi$ and updates the state until it reaches a terminal state. Then, if a positive reward is realized, $\ALG(\pi)$ executes Bernoulli trials until full exploration. Otherwise, it selects the default arm forever (we elaborate more when we explain Algorithm~\ref{alg:alg of pi}). It is straightforward to see that
\begin{observation}\label{obs: U get W}
For any policy $\pi$, there exists an algorithm $\ALG(\pi)$ such that 
\[
\lim_{T \rightarrow \infty }\mU_T(\ALG(\pi)) = W(\pi,A).
\]
\end{observation}
An immediate corollary of Observations~\ref{obs: U leq W*} and~\ref{obs: U get W} is a sandwich argument. For an optimal policy $\pi^\star$, namely, $W(\pi^\star,A)=W^{\star}(A)$, it holds that
\begin{equation}\label{eq:sandwith}
\OPT_{\infty} \leq W(\pi^\star,A) = \lim_{T\rightarrow \infty}\mU_T(\ALG(\pi^\star))  \leq  \OPT_{\infty};  
\end{equation}
thus, wrapping the optimal GMDP policy results in optimal welfare. 

We are ready to present $\SEGB$, which is implemented in Algorithm~\ref{alg:alg of pi}. Following the notation of Observation~\ref{obs: U get W}, $\SEGB=\ALG(\OGP)$, namely $\SEGB$ wraps $\OGP$. $\SEGB$ picks portfolios according to $\OGP$ until it reaches a terminal state (Lines~\ref{algpi:while}--\ref{algpi:upd s}). Then, if $\SEGB$ realized a positive reward (Line~\ref{algpi:pos then geo}), it executes Bernoulli trials to explore all arms (Line~\ref{algpi:geo}), and exploits the best arm (Line~\ref{algpi:exploit best}). Otherwise, it selects the default arm $a_0$ forever. Inequality~\eqref{eq:sandwith} and Theorem~\ref{thm:optimal policy} suggest the following corollary.
\begin{corollary}\label{cor: alg is asym opt}
Fix any arbitrary $\ise$ instance satisfying Assumption~\ref{assumption:dominance}. It holds that 
\[\lim\limits_{T \rightarrow \infty }\mU_T(\SEGB) = \OPT_{\infty}.
\]
\end{corollary}

\begin{algorithm}[t]
\renewcommand{\algorithmiccomment}[1]{\texttt{\color{blue}{\##1}}}

\caption{MIR Exploration via GMDP and Bernoulli Trials ($\SEGB$) \label{alg:alg of pi}}
\begin{algorithmic}[1]
\REQUIRE the $\OGP$ policy
\STATE $s\gets A$
\WHILE[$s$ is not a terminal state] {$\OGP(s)\neq \emptyset$\label{algpi:while}}{
\STATE select $\OGP(s)$, and denote the realized action by $a_k$\label{algpi:play with ogp}
\STATE $s\gets s\setminus \{a_k\}$\label{algpi:upd s}
}
\ENDWHILE
\IF[A positive reward is revealed]{$x_{a_k}>0$ for some $a_k\in A$\label{algpi:pos then geo}} {
        \STATE $a_{k^\star} \gets \argmax_{a_i} x(a_i)$ \COMMENT{best among all the explored arms} 
		\STATE execute Bernoulli trials mixing $a_{k^\star}$ with every other unexplored arm until all are revealed \label{algpi:geo}
		\STATE select the best arm forever \label{algpi:exploit best}
}
\ENDIF
\STATE \textbf{else}, select the default arm $a_0$ forever \label{algpi:exploit a0}
\end{algorithmic}
\end{algorithm}

\subsection{Convergence Rate}
Despite $\SEGB$'s optimality when $T\rightarrow \infty$, its welfare is lower than $\OPT_\infty$ for finite $T$. In this subsection, we analyze $\SEGB$'s performance for finite time and bound its convergence rate. Our goal is to show the gap between $\mU_T(\SEGB)$ and $\OPT_\infty$. The analysis of this subsections assumes that $\abs{X_i}\leq H$ for some $H\in \mathbb R^+$ for all $i\in\{1,\dots K\}$.

We begin by noting that in special cases, $\SEGB$ can be slightly modified to achieve $\OPT_\infty$ even for a finite $T$. To illustrate, consider $(X(a_i))_i$ that are supported on $\{x^-,x^+\}$, with $x^+>0$ and $x^- <0$. In such a case, revealing a positive reward of $x^+$ suggests that we need not explore any further (since the other rewards cannot outperform $+1$); thus, the Bernoulli trials in Line~\ref{algpi:geo} become redundant. Formalizing this intuition,
\begin{proposition}\label{prop:bernoulli opt}
Fix any arbitrary $\ise$ instance such that $(X(a_i))_i \in \{x^-,x^+\}$ w.p. 1, i.e., the rewards are almost surely supported on only two values. Let $\SEGB'$ be a modified version of $\SEGB$ that exploits once a reward of $x^+$ is realized. Then, there exists $T_0$ such that whenever $T\geq T_0$, $\mU_T(\SEGB') =\OPT_T$.
\end{proposition}
However, $\SEGB$ cannot achieve $\OPT_\infty$ in finite time due to two factors. First, the Bernoulli trials in Line~\ref{algpi:geo} might take several rounds, thereby delaying the potentially higher rewards hidden by risky arms. Consequently, the number of rounds that can benefit from these high rewards is limited (unlike in the $T\rightarrow\infty$ case). Second, the Bernoulli trials can be unfruitful in expectation. To illustrate, assume that there are only a few rounds left and one unexplored arm from $\below(A)$. In such a case, the one-time cost of exploring this arm outweighs the value it would provide in the remaining rounds. Deciding whether to mix such arms in Bernoulli trials or to exploit the best-seen arm constitutes another barrier. We leave this technical challenge for future work. Instead, our goal is to bound the time it takes to reveal all arms and address the potential performance gap it causes. 

Given an $\ise$ instance, let the random variable $\delta$ denote the highest positive reward among the arms of $\above(A)$, i.e., $\delta = \max_{a_i\in \above(A): X(a_i)>0} X(a_i)$. Recall that in Line~\ref{algpi:geo} of Algorithm~\ref{alg:alg of pi}, we execute Bernoulli trials with the best-seen arm. The term $\delta$ constitutes a lower bound on the reward of that arm $a_{k^\star}$. Furthermore, let $\eta = \max_{a_i\in A:\mu(a_i)<0}\abs{\mu(a_i)}$. Notice that $\eta$ quantifies the highest absolute value of an arm in $\below(A)$. The success probability of each Bernoulli trial, which leads to exploring one additional arm, is at least 
$ \frac{\delta}{\delta+\eta}$. Consequently, after $K(1+{\eta}\E\left[{\frac{1}{\delta}}\right])$ Bernoulli trials in expectation, we would explore all arms. We formalize this intuition through Proposition~\ref{prop:i-d bounds} below.
\begin{proposition}\label{prop:i-d bounds}
Fix any arbitrary $\ise$ instance satisfying Assumption~\ref{assumption:dominance}, and let $\delta>0$ and $\eta$ be the quantities defined above. There exists $T_0$ such that whenever $T>T_0$, it holds that
$
\mU_T(\SEGB) \geq \left(  1-\frac{K(1+\eta\cdot \E\left[{\frac{1}{\delta}}\right]}{T}\right) \OPT_\infty.
$
\end{proposition}
To illustrate, assume that $(X(a_i))_i$ are arbitrarily distributed in the discrete set $\{-H,\dots, H\}$. In such a case, as long as $\delta > 0$, $\frac{\eta}{\delta} \leq H$ holds almost surely; thus,  Proposition~\ref{prop:i-d bounds} suggests that $\SEGB$ is optimal up to a multiplicative factor of $\frac{K(H+1)}{T}$ or an additive factor of $\frac{KH(H+1)}{T}$.

\section{Strategic Agents and Bayesian Incentive Compatibility}\label{sec:ic body}
In this section, we address the case of \emph{strategic agents}: Agents only observe the mechanism's recommendation, but are free to select any arm. In Subsection~\ref{subsec:enabling}, we lay out several observations and assumptions that enable exploration. Then, Subsection~\ref{subsec:BIC ALG} presents $\ICSEGB$, which is implemented in Algorithm~\ref{mainicalg}. To simplify our analysis, we shall assume that the rewards are bounded. Namely, there exists $H \in \mathbb R$ such that for every $a_i \in A^+$, it holds that $\abs{X(a_i)} \leq H$ almost surely.

\subsection{Enabling Exploration and Harmless Mechanisms}\label{subsec:enabling}
In the case of strategic agents, agents would naturally try to avoid exploration. However, as mechanism designers, we are still interested in maximizing social welfare and, consequently, in exploring risky arms. To allow such an exploration, we must have the standard assumption below. 
\begin{assumption}\label{assumption ic body support}
For every pair $i,j$ such that $i,j\in A^+$, it holds that $\Pr(X_i < \mu_j)>0$.
\end{assumption}
If Assumption \ref{assumption ic body support} does not hold for some pair $i,j$, agents would never agree to explore arm $a_j$ even though it could hide a higher reward than $a_i$.

The main technique used in previous work~\cite{mansour2015bayesian,Mansour2016Slivkins} is \emph{hidden exploration}: Crafting the mechanism in such a way that agents cannot determine whether they are exploring or exploiting. As a result, a crucial element that affects BIC is the agents' knowledge about the \emph{order of arrival}. In the BIC constraint in Definition~\ref{def:ic}, the expectation is conditioned on all the information the agent has about their arrival order. For instance, assume that agents have a uniform belief about their order of arrival and fix any a priori inferior arm $a$. Notice that there is at most one round of exploration for each arm $a$ and potentially many exploitation rounds (under the event that $a$ is observed to be the best arm). Consequently, given enough agents, the probability of exploitation is higher than exploration from an agent's standpoint.
\begin{proposition}\label{prop:ic for uniform}
Fix any arbitrary $\ise$ instance, and assume uniform belief over arrival orders. There exists $T_0$ such that for any $T>T_0$, $\SEGB$ is BIC.
\end{proposition}
The more common and challenging case, which is the standard in prior work~\cite{Kremer2014,mansour2015bayesian}, is the informative order, where each agent knows their round number exactly. Namely, the agent arriving at time $l$ knows that she is the $l$'th agent. We focus on this setting for the remainder of this section.

It turns out that we must have stronger information asymmetry regarding the default arm $a_0$, which we previously assumed has a known reward of zero.\footnote{Assuming $X(a_0)=0$ almost surely is without loss of generality, as all of our results in Sections~\ref{sec:infinite} and~\ref{sec:policy to algorithm} extend beyond it. In this section, assuming the default arm has a deterministic reward is \emph{with} loss of generality, and hence we elaborate.} To illustrate the complexity, we highlight the following property.
\begin{definition}[Harmless Mechanism]
    A mechanism is called \emph{harmless} if, for any round $t$, information $\mI_t$ and any realized arm it selects $m^t$, it holds that $\Pr(X(m^t)>0\mid I_t)>0$.
\end{definition}
In other words, if an arm $a$ is observed to have a negative reward, the mechanism will avoid recommending $a$ in subsequent rounds almost surely. If a mechanism is not harmless, we say it is \emph{harmful}. What can a mechanism gain from being harmful? The following observation answers this question.
\begin{observation}\label{obs: harmless not exploring}
Assume that $X(a_0)=0$ almost surely and that $\max_{a_i\in \above(A)}X_i <0$. If a MIR and BIC mechanism explores the arms of $\below(A)$, it is necessarily harmful.
\end{observation}
We now explain why this observation holds. Fix any harmless mechanism and assume w.l.o.g. that the arms are sorted in decreasing order of expectation. The first agent, knowing she is the first, will only agree to pull $a_1$, the a priori optimal arm. Therefore, any BIC mechanism will recommend the portfolio $\bl p_{1,1}$. Next, consider the second agent. She knows that the only arm that has been explored so far is $a_1$. Thus, if the mechanism recommends anything but $a_1$, she will pick $a_2$. Since $X(a_1)<0$ due to the assumptions of the observation ($\max_{a_i\in \above(A)}X_i <0$), any harmless mechanism cannot mix $a_1$ in a portfolio. Thus, the second agent pulls $a_2$. This argument holds inductively for all arms in $\above(A)$.  From here on, we focus on harmless mechanisms only. 

In light of Observation~\ref{obs: harmless not exploring}, we need a stronger form of information asymmetry regarding the default arm; hence, we shall make the standard assumption that the default arm $a_0$ is a priori superior~\cite{Fiduciary,Kremer2014}. This reflects cases where agents are generally informed about the attractiveness of the arms and prefer the default arm over the others. 
\begin{assumption}\label{assumption order}
The expected values of the rewards satisfy $\mu(a_0)>\mu(a_1)\geq \mu(a_2)\cdots \geq \mu(a_K)$.
\end{assumption}
For completeness, we also remark that the MIR constraint from Inequality~\eqref{ineq:IR def} becomes
\[
\E_{a\sim \bl p}[X(a) \mid \mI]=\sum_{a_i\in A}\bl p (a_i)\E\left[X(a_i)\mid \mI\right]\geq\E\left[X(a_0)\mid \mI\right].
\]
Notably, the right-hand side is no longer zero. Due to Assumption~\ref{assumption order}, any MIR mechanism must recommend the default arm to the first agent.

\subsection{The $\ICSEGB$ Algorithm}\label{subsec:BIC ALG}

\begin{algorithm}[tb]
\caption{Incentive Compatible $\mainalg$ \label{mainicalg} ($\ICSEGB$)}
\begin{algorithmic}[1]
\REQUIRE the $\mainalg$ algorithm
\STATE initialize an instance of $\mainalg$ and update it after every recommendation\label{alg-ic:initialize} 
\STATE recommend the default arm $a_0$ to the first agent, observe $x(a_0)$ \label{alg-ic:first agent}
\IF{$x(a_0) < \mu_K$ \label{alg-ic:very low reward}} 
    \STATE recommend using \gre to agents $2,\dots, K+1$\label{alg-ic: r1 is bad}%
    \ELSE
    \STATE recommend $a_0$ to agents $2,\dots, K+1$  \label{alg-ic: r1 after one}
\ENDIF
\STATE split the remaining rounds into consecutive phases of $B$ rounds each \label{alg-ic:split}
\FOR {phase $k=1,\dots$ \label{alg-ic:for loop}}
    \IF{$\mainalg$ exploits (Lines \ref{algpi:exploit best}-\ref{algpi:exploit a0} in   $\mainalg$) \label{alg-ic:if exploits}}
        \STATE follow $\mainalg$ \label{alg-ic:exploits} 
    \ELSE {\label{alg-ic:else block explore}
    \STATE pick one agent $l(k)$ from the $B$ agents in this phase uniformly at random and recommend her according to $\mainalg$ \label{alg-ic: pick to explore}}
    \STATE as for the rest of the agents,
    \IF{an arm $a_i$ with $x(a_i)>x(a_0)$ was revealed \label{alg-ic:was revealed}}
    \STATE recommend using \gre \label{alg-ic: recommend as greedy if observed}
    \ELSE
        \STATE recommend $a_0$\label{alg-ic: recommend a on in phase} 
    \ENDIF
    \ENDIF
\ENDFOR
\end{algorithmic}
\end{algorithm}
In this subsection, we present a Bayesian incentive-compatible mechanism with approximately optimal welfare. The algorithm, which we call $\ICSEGB$ and is implemented in Algorithm~\ref{mainicalg}, uses $\mainalg$ to conduct exploration under the MIR and BIC constraints. From a technical standpoint, our approach builds on the techniques introduced by~\citet{mansour2015bayesian}. We divide rounds into \textit{phases}: In each phase, there is at most one exploration round (following $\mainalg$). The other rounds are either recommending the default arm $a_0$ or greedy exploitation. We denote by $\gre$ the algorithm that picks the best arm in expectation over the information it has (could be an already observed arm).

The algorithm works as follows. It initializes an $\mainalg$ instance in Line~\ref{alg-ic:initialize}. For the first agent to arrive, it recommends the default arm and observes its realized reward in Line~\ref{alg-ic:first agent}. Then, if the default arm realizes a very low reward, we enter the ``if'' clause in Line~\ref{alg-ic:very low reward} and continue by following $\gre$ for the next $K$ agents (Line~\ref{alg-ic: r1 is bad}). Otherwise, in Line~\ref{alg-ic: r1 after one}, we recommend the default arm for those agents. Line~\ref{alg-ic:split} splits the next rounds into phases, each with $B$ rounds (we determine $B$ later on). The for loop in Line~\ref{alg-ic:for loop} contains two possible behaviors. If $\mainalg$ exploits (the ``if'' clause in Line~\ref{alg-ic:if exploits}), we exploit as well. Otherwise, we enter the ``else'' clause in Line~\ref{alg-ic:else block explore}. One of the agents in this phase, who we note by $l(k)$ and is chosen uniformly at random, will get a recommendation according to $\mainalg$ (Line~\ref{alg-ic: pick to explore}). The other agents will either get the recommendation of $\gre$ or the default arm. The former refers to the ``if'' clause in Line~\ref{alg-ic:was revealed}: As long as we discover an arm $a_i$ with $x(a_i)>x(a_0)$, which means that $\mainalg$ can explore all arms using Bernoulli trials. This completes the description of the algorithm.

To analyze the social welfare of $\ICSEGB$, we need to determine the phase length $B$; thus, we introduce the following quantities $\xi$ and $\gamma$. Assumption~\ref{assumption ic body support} hints that there exist $\xi>0$ and $\gamma>0$ such that every arm $i \in A^+$ has a chance of at least $\gamma$ to be greater than any other arm by at least $\xi$. Formally, for all $i\in A^+$, it holds that $\Pr(\forall i'\in A^+\setminus \{i\}:\mu_i -X_{i'} > \xi  )>\gamma$. We use these quantities and set $B= \ceil*{\frac{H}{\xi \gamma}}+1$. We are ready to summarize the properties of $\ICSEGB$.
\begin{theorem}\label{theorem: ic fee}
Under Assumptions~\ref{assumption:dominance},\ref{assumption ic body support} and~\ref{assumption order}, $\ICSEGB$ satisfies MIR and BIC. In addition,
\[
\mU_T(\ICSEGB) \geq \left(  1-O\left(\frac{K \eta H \E\left[\frac{1}{\delta}\right] }{T \xi \gamma}  \right)\right) \OPT_\infty.
\]
\end{theorem}

\section{Conclusion and Discussion}\label{sec:discussion}
Our model emphasizes the MIR requirement, ensuring that agents receive, in expectation, at least as much as they would have received without using the mechanism, based on the mechanism's accumulated information. 
In Section~\ref{sec:infinite}, we introduced an auxiliary GMDP and achieved our main contribution: An index-based optimal policy that only requires sorting arms according to their expected values. Later, in Section~\ref{sec:policy to algorithm}, we leveraged this optimal GMDP policy to develop an approximately optimal algorithm and analyzed its convergence rate. Finally, in Section~\ref{sec:ic body}, we addressed incentive compatibility.

We see considerable scope for future work. On the technical side, one possible follow-up is to relax Assumption~\ref{assumption:dominance}; more details can be found in Subsection~\ref{subsec:conjecture}. Another possible direction is to consider stochastic bandits with Bayesian priors, moving beyond the assumption of static rewards. The MIR constraint is readily extendable to this case as well. More conceptually, MIR concerns only the expected value. However, in many settings, other factors, such as variance, should also be considered. In such cases, our constraint can be generalized to $\sum_{a_i\in A}\bl p (a_i)\E\left[f(X(a_i))\mid \mI\right]\geq 0$, for some function $f$. This more general formulation can express risk-aversion or risk-seeking behavior, as well as more complex quantities that depend on the reward distribution. Our results from Section~\ref{sec:infinite} hold for any general $f$, as long as it agrees with the stochastic order. Finally, the MIR constraint could be applied to other explore-exploit models, such as Markov Decision Processes.


\subsection{Open Problem: Relaxing Assumption~\ref{assumption:dominance}}\label{subsec:conjecture}
Unfortunately, our results hold only under the stochastic order assumption, Assumption \ref{assumption:dominance}. This limitation is not merely a byproduct of our proof technique; there are examples where $\OGP$ ceases to be optimal due to incorrect planning (see~{\ifnum\Includeappendix=0{the appendix}\else{the proof of Proposition \ref{prop:index with ugeq one} in the appendix}\fi}).

To illustrate why $\OGP$ leads to sub-optimal performance in the absence of stochastic order, consider three crucial aspects of selecting an arm $a_j$ as part of a portfolio (a randomized action). First, its expected value $\mu(a_j)$ determines the likelihood of exploring it. Second, the probability of obtaining a positive value, $\Pr(X_{a_j}>0)$,  influences the exploration of all other arms via Bernoulli trials. Third, the potential for exceptionally high rewards if $X_{a_j}$ turns out to be positive, $\E(\max_{a_l } X_{a_l}\mid X_{a_j}>0 )$, can generate significant regret if not explored. Under Assumption~\ref{assumption:dominance}, ordering the arms in $\below(A)$ by each of these factors yields the same order. Without stochastic order, the exploration index, and whether such an index even exists, is unclear. We conjecture that a different index-like policy is optimal.

\begin{conjecture}
Let $A$ be an arbitrary set of arms, not necessarily satisfying Assumption \ref{assumption:dominance}. For every state $U\subseteq A$, let $f^U$ be the real-valued function, $f^U:\below(U)\rightarrow \mathbb R$ defined as
\[
f^U(a_{j})=\frac{\Pr(X_{a_j} > 0 )\E(\max_{a_l \in {U} } X_{a_l}\mid X_{a_j}>0 )}{\abs{\mu(a_j)}}.
\] 
Define a policy $\pi$ such that for every state $U$, $\pi(U)=\bl p_{i,j}$, where $a_i \in \above(U)$ is an arbitrary arm with a positive expected value and $a_j  \in \argmax_{a_{j'}\in \below(U)} f^U(a_{j'})$. Then, $W(\pi,s_0)=\OPT$. 
\end{conjecture}
The base cases we proved for Theorem \ref{thm:optimal policy} support this conjecture. 

\section*{Acknowledgments}
The work of Omer Ben-Porat was supported by the Israel Science Foundation (ISF; Grant No. 3079/24).


{\ifnum\Includeappendix=1{ 
\appendix

\section{Proof Outline for Theorem \ref{thm:optimal policy}}\label{sec:thm1 outline}
In this section, we outline the proof of Theorem~\ref{thm:optimal policy}. 
To allow this section to be self-explanatory, we reiterate some definitions that appear in the body of the paper. We begin with several notations and definitions we use extensively in the proof.
\subsection{Preliminaries}
We denote the set of all states by $\mS=2^A$. A \textit{policy} is a mapping from previous states and actions to a randomized action. Formally, let $\mH$ be the set of histories, $\mH=\cup_{k=0}^K\left(\mS \times \Delta(A)\right)^k$ be a tuple comprising pairs of states and randomized actions taken. A policy $\pi$ is a function $\pi:\mH \times \mS \rightarrow \Delta(A)$. We say that a policy is \textit{MIR} if for every $h\in \mH,s\in \mS$, $\pi(h,s)\in \safe(s)$. From here on, we consider MIR policies solely. Given a policy $\pi$ and a pair $(h,s)$, we let $W(\pi,h,s)$ denote the expected reward of $\pi$ when starting from $s$ after witnessing $h$. Namely,
\begin{align}\label{eq:W elaborated}
W(\pi,h,s) = 
\begin{cases}
R(s) & \textnormal{if }\safe(s)=\emptyset\\
\sum_{a\in s}\pi(h,s)(a)W(\pi,h\oplus(s,a),s\setminus \{a\}) & \textnormal{otherwise}
\end{cases}.
\end{align}
For every state $s$, let $W^\star(s)=\sup_{\pi'}W(\pi ',\emptyset,s)$.\footnote{As we show in Proposition \ref{prop:optimal p valid}, there exists a policy that attains this supremum.} While policies may depend on histories, it often suffices to consider \emph{stationary} policies.
\begin{definition}[Stationary]
A MIR policy $\pi$ is \textit{stationary} if for every two histories $h,h' \in \mH$ and a state $s \in \mS$, $\pi(h,s)=\pi(h',s)$.
\end{definition}
Since there is a finite set of states and the action sets are convex, there exists 
an optimal stationary policy; hence, from here on we address stationary policies solely. When discussing stationary policies, we thus neglect the dependency on $h$, writing $\pi(s)$. For stationary policies, the definition of $W$ is much more intuitive: Given a stationary policy $\pi$ and a state $s$,
\begin{align}\label{eq:w with terminal}
W(\pi,s)=\sum_{\substack{s_t\in \mS:\\ s_t \textnormal{ is terminal}}}\Pr(\pathto{s}{s_t})R(s_t),
\end{align}
where $\pathto{s}{s_t}$ indicates the event that, starting from $s$ and following the actions of $\pi$, the GMDP terminates at $s_t$. 

An additional useful notation is the following. For every state $s\in \mS$, we denote by $Q(\pi,s)$ the probability starting at $s\subseteq A$, following the policy $\pi$ and exploring all arms. Formally,
\[
Q(\pi,s)= \Pr(\pathto{s}{\emptyset}),
\]
Note that $Q$ is defined recursively: Namely, if $\pi(s)=\bl p_{i,j}$ for a non-terminal state $s$, then
\[
Q(\pi,s)=\bl p_{i,j}(a_i)Q(\pi,s\setminus\{a_i\})+\bl p_{i,j}(a_j)Q(\pi,s\setminus\{a_j\}).
\]
It will sometimes be convenient to denote $Q(\pi,\above(s),\below(s))$ for $Q(\pi,s)$, thereby explicitly stating the two distinguished sets of arms.

\subsection{Binary Structure}\label{subsec:bin}
A structural property of the above GMDP is that in every terminal state $s_t$, $\above(s_t)=\emptyset$, or otherwise we could explore more arms; thus, intuitively, the arms in $\above(A)$ provide us ``power'' to explore the arms of $\below(A)$. Following this logic, in every state we should aim to explore arms from $\below(A)$ and not those of $\above(A)$, subject to satisfying the MIR constraint.

Recall the definition of $\bl p_{i,j}$ and $\bl p_{i,i}$ from Equation~\eqref{eq:blp from body} in Subsection~\ref{subsec:optimal GMDP policy}. Next, we define $\mP,\mP'$ such that
\[
\mP\defeq \{\bl p_{i,j}\mid a_i\in \above(A), a_j\in \below(A) \}, \qquad \mP'\defeq\{\bl p_{i,i}\mid a_i\in \above(A) \}.
\]
Notice that $\mP\cup \mP'$ includes $O(K^2)$ actions, while $\safe(s)$ for a state $s$ is generally a convex polytope  with infinitely many actions. Further, in every non-terminate state $s$, ${\safe(s) \cap (\mP \cup \mP') \neq \emptyset}$. Next, we remind the reader of the definition of $\mP$-valid policies.
\begin{definition}[Mirroring Definition~\ref{def:p valid}]
A MIR policy $\pi$ is $\mP$-valid if for every non-terminal state $s\in \mS$,
\begin{itemize}
\item if $\below(s)\neq \emptyset$, then $\pi(s)\in \mP$;
\item else, if $\below(s)= \emptyset$, then $\pi(s)\in \mP'$.
\end{itemize} 
\end{definition}
Observe that $\mP$ is a strict subset of all MIR actions in the state $A$, which incorporate mixes of at most two arms. However, the set of MIR actions $\safe(s)$ for $s\subseteq A$ may include distributions mixing several elements of $A$. Due to the convexity of $W(\pi,s)$ in $\pi(s)$ (see the elaborated representation of $W$ in Equation \refeq{eq:W elaborated}), the GMDP exhibits a nice structural property, as captured by the following Proposition~\ref{prop:optimal p valid}. 
\begin{proposition}[Mirroring Proposition~\ref{prop:main optimal p valid}]\label{prop:optimal p valid}
There exists an optimal policy that is $\mP$-valid.
\end{proposition}
The proof of Proposition \ref{prop:optimal p valid} appears in Section~\ref{sec:aux}. 
Due to Proposition \ref{prop:optimal p valid}, we shall focus on $\mP$-valid policies. Such policies are easy to visualize using trees, as we exemplify next.\footnote{In Figure~\ref{example with normal} we illustrated the optimal policy using a graph that is not a tree. However, a tree structure  serves better the presentation of our technical statements.}
\begin{example}\label{example with four}
We reconsider Example~\ref{example with normal}, but neglect the actual distributions (as we only care about the expected values). Let $A=\{a_1,a_2,a_3,a_4\}$, with $\above(A)=\{a_1,a_2\}$ and $\below(A)=\{a_3,a_4\}$. Consider the tree description in Figure \ref{fig:tree example}. The root of the tree is the set of all arms. At the root, the policy picks $\bl p_{1,3}$. The outgoing left edge represents the case the realized action is $a_1$, which happens w.p. $\bl p_{1,3}(a_1)$. In such a case, the new state is $\{a_2,a_3,a_4\}$. With the remaining probability, $\bl p_{1,3}(a_3)$, the new state will be $\{a_1,a_2,a_4\}$. Leaves of the tree are terminal states, where no further exploration could be done. For instance, in the leftmost leaf, $\{a_3,a_4\}$, the only arms explored are $\{a_1,a_2\}$. The two highlighted nodes represent the same state. Since the presented policy is $\mP$-valid, it is stationary; hence, the policy acts exactly the same in these two nodes and their sub-trees.
\end{example}
Notice that the tree in Figure \ref{fig:tree example} represents only the \textit{on-path} states, i.e., states that are reachable with positive probability, while policies are functions from the entire space of states, including off-path states; thus, two different policies can be described using the same tree. Nevertheless, the tree structure is convenient and will be used extensively in our analysis. When we define a policy using a tree, we shall also describe its behavior at \textit{off-path} states.

The policy exemplified in Figure \ref{fig:tree example} has an additional combinatorial property: In every state $s$, it takes an action according to some order of the arms. This property is manifested in the following Definitions \ref{def:right ordered} and \ref{def:left ordered}.
\begin{definition}[Right-ordered policy]\label{def:right ordered}
A $\mP$-valid policy $\pi$ is right-ordered if there exist a bijection $\sigr_\pi: \below(A)\rightarrow [\abs{\below(A)}]$ such that in every state $s$ with $\below(s) \neq \emptyset$, $\pi(s)=\bl p_{i, {j^*}}$ where $a_i \in \above(s)$ and $a_{j^*} = \argmin_{a_j \in \below(s)} \sigr_\pi(a_j)$.
\end{definition}
\begin{definition}[Left-ordered policy]\label{def:left ordered}
A $\mP$-valid policy $\pi$ is left-ordered if there exist a bijection $\sigl_\pi: \above(A)\rightarrow [\abs{\above(A)}]$ such that in every state $s$ with $\below(s) \neq \emptyset$, $\pi(s)=\bl p_{{i^*}, j}$ where $a_j \in \below(s)$ and $a_{i^*} = \argmin_{a_i \in \above(s)} \sigl_\pi(a_i)$.
\end{definition}
In addition, we say that a policy is \textit{ordered} if it is right-ordered and left-ordered. To illustrate, observe the example in Figure \ref{fig:tree example}. The tree depicts an ordered policy, with $\sigl=(a_1,a_2)$ and $\sigr=(a_3,a_4)$. Notice that ordered policies are well-defined for off-path states.

\begin{figure}[t]
\centering
\includegraphics[scale=0.9]{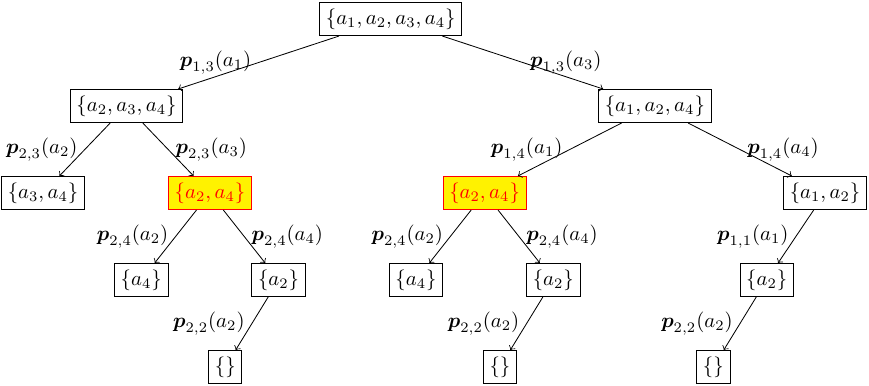}
\caption{The policy described in Example \ref{example with four}.  
Every node represents a state (the mapping is onto, but not one-to-one). Outgoing left edges imply the coin flips resulted in an arm from $\above(v)$, and outgoing right edges imply an arm from $\below(v)$. Leaves correspond to terminal states, where no action could be taken. \label{fig:tree example}
}
\end{figure}

\subsection{Proof Overview}\label{subsec:results}
We are ready to prove Theorem~\ref{thm:optimal policy}. The main tool in our analysis is Lemma \ref{lemma:equivalence}. Lemma \ref{lemma:equivalence} reveals a rather surprising feature of $Q$: $Q$ is policy independent. 
\begin{lemma}[Mirroring Lemma~\ref{lemma:equivalence body}]\label{lemma:equivalence}
For every two $\mP$-valid policies $\pi,\rho$ and every state $s\in \mS$, it holds that $Q(\pi,s)=Q(\rho,s)$.
\end{lemma}
The proof of Lemma \ref{lemma:equivalence} appears in~{\ifnum\Includeappendix=0{the appendix}\else{Section \ref{sec:proof of lemma}}\fi}. We stress that this lemma holds regardless of Assumption~\ref{assumption:dominance}. Next, we leverage Lemma \ref{lemma:equivalence} to prove the main technical result of the paper.
\begin{theorem}
\label{thm:holy grail}
Let $\pi^\star$ be a right-ordered, $\mP$-valid policy with $\sigr_{\pi^\star}$ ordered in decreasing expected value. Under Assumption \ref{assumption:dominance}, for every state $s\in \mS=2^A$, it holds that $W(\pi^\star,s)=W^\star(s)$.
\end{theorem}
In particular, Theorem \ref{thm:holy grail} implies that $W(\SEGB,s_0)=W^\star(s_0)$ and that Theorem~\ref{thm:optimal policy} holds,  since $\SEGB$ is right-ordered in decreasing expected value. The formal proof of Theorem \ref{thm:holy grail} is relegated to~{\ifnum\Includeappendix=0{the appendix}\else{Section~\ref{sec:proof of thm}}\fi}. 

\section{Proof of Lemma \ref{lemma:equivalence}}\label{sec:proof of lemma}
\begin{proofof}{Lemma \ref{lemma:equivalence}}
We prove the lemma by a two-dimensional induction on the number of arms in $\above(s)$ and $\below(s)$. We prove four base cases in Section \ref{sec:base for lemma}:
\begin{itemize}
\item $\abs{\above(s)}=1$ and $\abs{\below(s)}\geq 2$ (Proposition \ref{prop:case of one strong}).
\item $\abs{\above(s)}\geq 2 $ and $\abs{\below(s)} = 1$ (Proposition \ref{prop:case of one}).
\item $\abs{\above(s)}\geq 2$ and $\abs{\below(s)} = 2$ (Proposition \ref{prop:case of two strong}).
\item $\abs{\above(s)}=2$ and $\abs{\below(s)}\geq 2$ (Proposition \ref{prop:case of two}).
\end{itemize}
While the first two cases are almost immediate, the other two are technical and require careful attention. Next, assume the statement holds for all states $s\in \mS$ such that $\abs{\above(s)}\leq K_1$, $\abs{\below(s)}\leq K_2$ and $\abs{\above(s)}+\abs{\below(s)}< K_1+K_2$.

Let $U\in\mS$ denote a state with $\abs{\above(U)}=K_1$ and $\abs{\below(U)}=K_2$. For abbreviation, let $\ug\defeq\above(U),\ul\defeq\below(U)$. Further, define $Q^*(U) = \sup_{\pi} Q(\pi,U)$, \footnote{This supremum is attained since there are only finitely many $\mP$-valid policies.} and for every $a_i \in \ug, a_j\in \ul$ let 
\[
Q^*_{i,j}(\ug,\ul)\defeq\bl p_{i,j}(a_j)Q^*(\ug,\ul\setminus \{a_j\}) +\bl p_{i,j}(a_i)Q^*(\ug\setminus\{a_i\},\ul).
\]
Next, let $(a_{i^*},a_{j^*})\in \argmax_{a_i\in \ug,a_j\in \ul}Q^*_{i,j}(\ug,\ul)$, and assume by contradiction that there exists a pair $(a_{\tilde i}, a_{\tilde j})$ such that 
\begin{equation}\label{eq:contradiction of lemma}
Q^*_{{i^*},{j^*}}(U) > Q^*_{{\tilde i}, {\tilde j}}(U).
\end{equation}
\paragraph{Step 1} Fix arbitrary $a_{i'}$ and $a_{j'}$ such that  $a_{i'} \in \ug$ and $a_{j'} \in \ul$. We will show that 
\begin{equation}\label{eq:step 1 goal}
Q^*_{{i'},{j^*}}(U)=Q^*_{{i'},{j'}}(U).
\end{equation}
We define the ordered policy $\pi$ such that $\sigr_\pi=(a_{i'},\dots)$, i.e., $\sigr_\pi$ first explores $a_{i'}$ and then the rest of the arms of $\ug$ in some arbitrary order; and, $\sigl_\pi=(a_{j^*},a_{j'},\dots)$. In addition, we define $\rho$ such that $\sigl_{\rho}=\sigl_\pi$, and $\sigr_\rho=(a_{j'},a_{j^*},\dots)$. Due to the inductive assumption, we have 
\begin{align}\label{eq:policies suffice}
Q^*_{{i'},{j^*}}(U) &=  \bl p_{{i'},{j^*}}(a_{j^*})Q^*(\ug,\ul\setminus \{a_{j^*}\}) +\bl p_{{i'},{j^*}}(a_{i'})Q^*(\ug\setminus \{a_{i'}\},\ul) \nonumber \\
&= \bl p_{{i'},{j^*}}(a_{j^*})Q(\pi,\ug,\ul\setminus \{a_{j^*}\})+\bl p_{{i'},{j^*}}(a_{i'})Q(\pi,\ug\setminus \{a_{i'}\},\ul)\\
&=Q(\pi,U).\nonumber
\end{align}

Similarly, $Q^*_{{i'},{j'}}(U)  = Q(\rho,U)$; hence, proving that $Q(\pi,U) =Q(\rho,U)$ entails Equality (\ref{eq:step 1 goal}). Next, let $\suff(\sigl_\pi)$ be the set of all non-empty suffices of $\sigl_\pi$. Being left-ordered suggests that on-path\footnote{These are terminal states that $\pi$ reaches to with positive probability.} terminal states with all arms of $\ul$ explored of $\pi$ are of the form $(Z,\emptyset)$, where $Z\in \suff(\sigl_\pi)$. Next, we factor $Q(\pi,U)$ recursively as follows: We factor $Q(\pi,U)$ into two terms, like in Equation \refeq{eq:policies suffice}. Following, for each term obtained, we ask whether the corresponding state excludes $\{a_{j^*},a_{j'}\}$. If the answer is yes, we stop factorizing it, and move to the other terms. We do this recursively, until we cannot factor anymore, or we reached a terminal state. Using this factorizing process, we have \footnote{We stop factorizing if both $a_{j^*},a_{j'}$ were observed; thus, $Z$ will never be the empty set.} 
\begin{align*}
Q(\pi,U) &= \alpha \cdot Q(\pi,\emptyset,\ul)+\beta\cdot Q(\pi,\emptyset,\ul\setminus \{a_{j^*} \}) + \sum_{Z \in \suff(\sigl_\pi)} c^\pi_Z \cdot Q(\pi,Z,\ul\setminus \{a_{j^*},a_{j'}\}),
\end{align*}
for $\alpha=\Pr(\pathto{s}{(\emptyset,\ul)})$ and $\beta=\Pr(\pathto{s}{(\emptyset,\ul\setminus\{a_{j^*}\})})$ such that $\alpha+\beta+\sum_{Z\in \suff(\sigl_\pi)}c^\pi_Z =1$ and $\alpha,\beta,c^\pi_Z\in [0,1]$ for every $Z \in \suff(\sigl_\pi)$. In this representation, $\alpha$ is the probability of reaching the terminal $(\emptyset,\ul)$, while $\beta$ is the probability of reaching the terminal state $(\emptyset,\ul\setminus\{a_{j^*}\})$. For these two terminal states, we know that $Q^*(\emptyset,\ul)= Q^*(\emptyset,\ul\setminus \{a_{j^*} \})=0$; hence,
\begin{align}\label{eq:pi j^* to j}
Q(\pi,U) &= \sum_{Z \in \suff(\sigl_\pi)} c^\pi_Z \cdot Q(\pi,Z,\ul\setminus \{a_{j^*},a_{j'}\}).
\end{align}
Following the same factorization process for $\rho$, we get
\begin{align}\label{eq:rho j^* to j}
Q(\rho,U) &= \sum_{Z \in \suff(\sigl_\rho)} c^\rho_Z \cdot Q(\rho,Z,\ul\setminus \{a_{j^*},a_{j'}\}).
\end{align}
Next, we want to simplify the coefficients $\left(c^\pi_Z\right)_Z$. We remark that $c^\pi_Z$ is not simply the probability of reaching $(Z,\ul\setminus \{a_{j^*},a_{j'}\})$ from $s$, i.e., $\Pr(\pathto{s}{(Z,\ul\setminus \{a_{j^*},a_{j'}\})})$. To clarify, consider a strict suffix $Z$, $1\leq \abs{Z}< \abs{\above(A)}$, and the suffix $Z'=Z \cup\{a_l\}$ for the minimal element $a_l \in \ug \setminus Z$ according to $\sigl_\pi$,i.e., $a_l = \argmin_{a\in \ug \setminus Z}\sigl_\pi(a)$. In the factorization process that produced Equation (\ref{eq:pi j^* to j}), once we got the term $Q(\pi,Z',\ul\setminus \{a_{j^*},a_{j'}\})$, we stopped factorizing any further; thus, $c^\pi_Z$ does not include the probability of reaching a node associated with $(Z',\ul\setminus \{a_{j^*},a_{j'}\})$ and then following the left edge to $(Z,\ul\setminus \{a_{j^*},a_{j'}\})$. However, this probability is taken into account in $\Pr(\pathto{s}{(Z,\ul\setminus \{a_{j^*},a_{j'}\})})$. Rather, $c^\pi_Z$ is the probability of reaching any node $v$ in the tree induced by $\pi$ with the following property: $v$ represents the state $(Z,\ul\setminus \{a_{j^*},a_{j'}\})$, while $a_{j'}$ does not belong to the state represented by the parent of $v$.  In the tree interpretation, $v$ should also be a \textit{right child of its parent} (for instance, the left highlighted node in the tree in Figure \ref{fig:tree example}). The following Proposition \ref{prop:coef c} describes $\left(c^\pi_Z\right)_Z$ in terms of $Q$.
\begin{proposition}\label{prop:coef c}
For every $Z\in \suff(\sigl_\pi)$, let $a_{i(Z)} = \argmin_{a_i\in Z} \sigl_\pi(a_i)$. It holds that
\[
c^\pi_Z = Q(\pi,\ug\setminus Z \cup \{a_{i(Z)}\}, \{a_{j^*},a_{j'}\})-Q(\pi,\ug\setminus Z, \{a_{j^*},a_{j'}\}).
\]
\end{proposition}
The proof of Proposition \ref{prop:coef c} appears at the end of this proof. Notice that for every $Z$, $c^\pi_Z$ includes values of $Q$ with less arms than $U$ (besides, perhaps, the case where $\abs{\ug}=2$ and $\abs{Z}=1$ obtaining $Q(\pi,\ug, \{a_{j^*},a_{j'}\})$, but we cover this case in the bases cases); consequently, due to the inductive step
\begin{align}\label{c pi is rho}
c^\pi_Z = Q(\rho,\ug\setminus Z \cup \{a_{i(Z)}\}, \{a_{j^*},a_{j'}\})-Q(\rho,\ug\setminus Z, \{a_{j^*},a_{j'}\})=c^\rho_Z,
\end{align}
where the last equality follows from mirroring Proposition \ref{prop:coef c} for $(c^\rho_Z)_Z$. Ultimately,
{\thinmuskip=.2mu
\medmuskip=0mu plus .2mu minus .2mu
\thickmuskip=1mu plus 1mu
\begin{align*}
Q(\pi,U) &\stackrel{\textnormal{Eq. (\ref{eq:pi j^* to j})}}{=}\sum_{Z \in \suff(\sigl_\pi)}c^\pi_Z \cdot Q(\pi,Z,\ul\setminus \{a_{j^*},a_{j'}\})\stackrel{\textnormal{Eq. (\ref{c pi is rho})}}{=}\sum_{Z \in \suff(\sigl_\pi)}c^\rho_Z \cdot Q(\pi,Z,\ul\setminus \{a_{j^*},a_{j'}\}) \nonumber\\
&\stackrel{\textnormal{Ind. step}}{=}\sum_{Z \in \suff(\sigl_\pi)}c^\rho_Z \cdot Q(\rho,Z,\ul\setminus \{a_{j^*},a_{j'}\}) \stackrel{\sigl_{\rho}=\sigl_\pi}{=}\sum_{Z \in \suff(\sigl_\rho)}c^\rho_Z \cdot Q(\rho,Z,\ul\setminus \{a_{j^*},a_{j'}\}) \nonumber\\
&\stackrel{\textnormal{Eq. (\ref{eq:rho j^* to j})}}{=}Q(\rho,U) .
\end{align*}}
This completes Step 1.

\paragraph{Step 2}
Fix arbitrary $a_{i'}$ and $a_{j'}$ such that $a_{i'} \in \ug$ and $a_{j'} \in \ul$. We will show that 
\begin{equation}\label{eq:step 2 goal}
Q^*_{{i^*},{j'}}(U)=Q^*_{{i'},{j'}}(U).
\end{equation}
We follow the same technique as in the previous step. Let $\pi$ be an ordered policy such that $\sigl_\pi=(a_{i^*},a_{i'},\dots )$, i.e., $\sigl_\pi$ ranks $a_{i^*}$ first, $a_{i'}$ second and then follows some arbitrary order on the remaining arms, and $\sigr_\pi=(a_{j'},\dots)$. In addition, we define the ordered policy $\rho$ with $\sigl_\rho=(a_{i'},a_{i^*},\dots )$, where the dots refer to any arbitrary order on the remaining elements of $\ug$, and $\sigr_\rho =\sigr_\pi=(a_{j'},\dots)$. Using the inductive step and the same arguments as in Equation (\ref{eq:policies suffice}), it suffices to show that $Q(\pi,U)=Q(\rho,U)$. We factor $Q(\pi,U)$ recursively such that 
\begin{align}\label{eq: q pi with d}
Q(\pi,U)=Q(\pi,\{a_{i^*},a_{i'}\},\ul)+ \sum_{Z \in \suff(\sigr_\pi)} d^\pi_Z \cdot Q(\pi,\ug \setminus \{a_{i^*},a_{i'}\},Z),
\end{align}
and similarly
\begin{align}\label{eq: q rho with d}
Q(\rho,U)=Q(\rho,\{a_{i^*},a_{i'}\},\ul)+ \sum_{Z \in \suff(\sigr_\rho)} d^\rho_Z \cdot Q(\rho,\ug \setminus \{a_{i^*},a_{i'}\},Z).
\end{align}
Next, we claim that
\begin{proposition}\label{prop:coef d}
For every $Z\in \suff(\sigr_\pi)$, let $a_{j(Z)} = \argmin_{a_j\in Z}\sigr_\pi(a_j)$. It holds that
\[
d^\pi_Z = Q(\pi,\{a_{i^*},a_{i'}\},\ul \setminus Z )-Q(\pi,\{a_{i^*},a_{i'}\},\ul \setminus Z \cup \{a_{j(Z)}\}).
\]
\end{proposition}
The proof of Proposition \ref{prop:coef d} appears at the end of this proof. Notice that for every $Z$, $d^\pi_Z$ includes values of $Q$ with less arms than $U$ (besides, perhaps, the case where $\abs{\ul}=2$ and $\abs{Z}=1$ obtaining $Q(\pi,\{a_{i^*},a_{i'}\},\ul\})$, but we cover this case in the bases cases); consequently, due to the inductive step
\begin{align}\label{eq: d pi is rho}
d^\pi_Z = Q(\rho,\{a_{i^*},a_{i'}\},\ul \setminus Z )-Q(\rho,\{a_{i^*},a_{i'}\},\ul \setminus Z \cup \{a_{j(Z)}\})=d^\rho_Z,
\end{align}
where the last equality follows from mirroring Proposition \ref{prop:coef d} for $(d^\rho_Z)_Z$. Ultimately, by rearranging Equation (\ref{eq: q pi with d}) and invoking the inductive step, Equation (\ref{eq: d pi is rho}) and the fact that $\sigr_\rho =\sigr_\pi$, we get
\begin{align*}
Q(\pi,U)=Q(\rho,\{a_{i^*},a_{i'}\},\ul)+ \sum_{Z \in \suff(\sigr_\rho)} d^\rho_Z \cdot Q(\rho,\ug \setminus \{a_{i^*},a_{i'}\}),Z)\stackrel{\textnormal{Eq. (\ref{eq: q rho with d})}}{=}Q(\pi,U),
\end{align*}
implying Equation (\ref{eq:step 2 goal}) holds.
\paragraph{Step 3} We are ready to prove the lemma. Fix arbitrary $a_{\tilde i}$ and $a_{\tilde j}$ such that $a_{\tilde i} \in \ug$ and  $a_{\tilde j} \in \ul$. By the previous Step 1 and Step 2, we know that
\[
Q^*_{{i^*},{j^*}}(U)\stackrel{\textnormal{Step 1}}{=}Q^*_{{i^*},{\tilde j}}(U)\stackrel{\textnormal{Step 2}}{=}Q^*_{{\tilde i},{\tilde j}}(U),
\]
which contradicts Equation (\ref{eq:contradiction of lemma}); hence, the lemma holds.
\end{proofof}

\section{Additional Statements for Lemma \ref{lemma:equivalence}}
\begin{proofof}{Proposition \ref{prop:coef c}}
\begin{figure}
\centering
\includegraphics[scale=0.9]{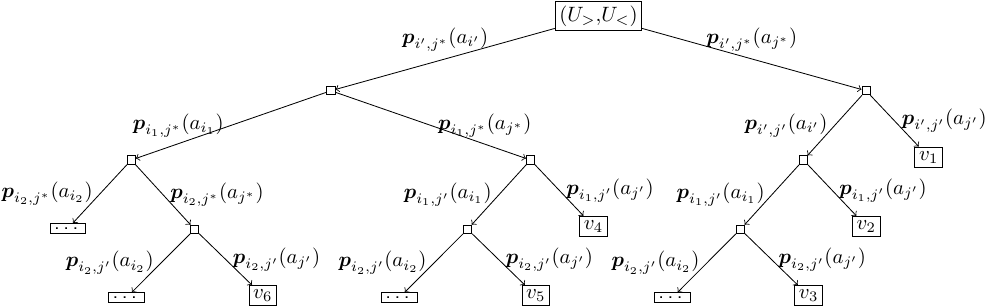}
\caption{Illustration for Proposition \ref{prop:coef c}. The tree depicts $T(\pi)$. Nodes $v_1$ to $v_6$ are nodes whose sub-trees were pruned in the construction of $T$. Let $Z=\ug\setminus \{a_{i'},a_{i_1}\}$ and $Z^c= \{a_{i'},a_{i_1}\}$. The minimal element of $Z$, denoted $a_{i(Z)}$ in the proof, is $a_{i_2}$. The corresponding $c^\pi_Z$ is the probability to reach one of $\{v_3,v_5,v_6\}$, namely, $c^\pi_Z = \Pr(\{v_3,v_5,v_6\})$. In the tree $T$, we ignore sub-trees of nodes $v$ labeled with ``$\dots $'' since these do not contribute to $c^\pi_Z$. Observe that the probability of reaching $v_i$, for $i\in\{1,\dots, 6\}$ is the same in $T(\pi)$ and $T$. Finally, notice that $Q(\pi,\ug\setminus Z \cup \{a_{i(Z)}\}, \{a_{j^*},a_{j'}\})=\Pr(\{v_1,v_2,\dots,v_6 \})$, and $Q(\pi,\ug\setminus Z, \{a_{j^*},a_{j'}\})=\Pr(\{v_1,v_2,v_4\})$. Combining, we get that $c^\pi_Z = Q(\pi,\ug\setminus Z \cup \{a_{i(Z)}\}, \{a_{j^*},a_{j'}\})-Q(\pi,\ug\setminus Z, \{a_{j^*},a_{j'}\})=\Pr(\{v_3,v_5,v_6\})$ as required.
\label{fig:tree illustation}}
\end{figure}
To prove this claim, we focus on the tree induced by $\pi$, $T(\pi)$. It is convenient to discuss a modified version $T(\pi)$ obtained by pruning, and this is feasible since even if prune nodes from $T(\pi)$ it still remains Markov chain. We illustrate the proof of this claim in Figure \ref{fig:tree illustation}.

We factorize $Q(\pi,U) $ recursively (see Equation (\ref{eq:pi j^* to j})) until we hit a node associated with a state that excludes $\{a_{j^*},a_{j'}\}$, or a leaf. This factorization can be illustrated as follows: We traverse $T(\pi)$, from right to left. Every node we visit, we ask whether that node includes $\{a_{j^*},a_{j'}\}$. If it does not, we prune its sub-tree (i.e., it becomes a leaf) while leaving it intact. Denote the obtained tree by $T$, and let $V(T)$ be its set of nodes. Observe that
\begin{observation}\label{obs:two types}
Every leaf $v$ in $V(T)$ satisfies exactly one property: 
\begin{enumerate}[leftmargin=0cm,itemindent=.5cm,labelwidth=\itemindent,labelsep=0cm,align=left]
\item[]\textit{Type 1:} $\above(state(v))=\emptyset$, or
\item[]\textit{Type 2:} $\below(state(v))=\ul\setminus \{j^*,j'\}$ with $\below(state(v)) \subset \below(state(parent(v)))$.
\end{enumerate}
\end{observation}
Leaves of type 1 are associated with terminal states of the MDP (see Subsection~\ref{subsec:aux GMDP}). Leaves of type 2 are those whose sub-trees were pruned during the traversal. Moreover, $\below(state(v)) \subset \below(state(parent(v)))$ holds in every such a leaf $v$, since otherwise we would have pruned its parent. Due to Observation \ref{obs:two types}, every node $v$ with $state(v)=(Z,\ul\setminus \{a_{j^*},a_{j'}\})$ is of type 2; therefore,
\begin{align}\label{eq:c with T}
c^\pi_Z = \sum_{\substack{v\in V(T):state(v)=\\(Z,\ul\setminus \{a_{j^*},a_{j'}\})}}\Pr\left(\pth{root(T)}{}{v}\right).
\end{align}

Next, fix an arbitrary non-empty $Z$, $Z \subseteq \ug$, and $\Psi$ be the set of all non-empty suffixes of $\ug\setminus Z$. Consider $T$ and its root $root(T)$. Notice that $Q(\pi,\ug\setminus Z, \{a_{j^*},a_{j'}\})$ is the probability of reaching a (type 2) leaf $v$ such that $\above(v)=\psi\cup Z$ for some $\psi \in \Psi$. This is true since $\pi$ is ordered, and every path from $root(T)$ to such a leaf $v$ does not include any action from $Z$; hence, we can compare the probability of reaching it to off-path behavior of $\pi$. Further, $Q(\pi,\ug\setminus Z \cup \{a_{i(Z)}\}, \{a_{j^*},a_{j'}\})$ is the probability of reaching a (type 2) leaf $v$ such that $\above(v)=\psi\cup Z$ for some $\psi \in \Psi$ or $\above(v)=Z$; hence,
\[
Q(\pi,\ug\setminus Z \cup \{a_{i(Z)}\}, \{a_{j^*},a_{j'}\})-Q(\pi,\ug\setminus Z, \{a_{j^*},a_{j'}\})
\]
is precisely the right-hand-side of Equation \refeq{eq:c with T}.
\end{proofof}

\begin{proofof}{Proposition \ref{prop:coef d}}
Fix $Z\in \suff(\sigr_\pi)$, and let $a_{j(Z)} = \argmin_{a_j\in Z}\sigr_\pi(a_j)$. Let $T(\pi)$ denote the tree induced by $\pi$. Observe that
\begin{observation}\label{obs: for d}
The coefficient $d^\pi_Z$ is the probability to get to a node $v$ in $T(\pi)$ such that
\begin{enumerate}
\item $state(v)=(\ug\setminus \{a_{i^*},a_{i'} \},Z )$, and 
\item $state(parent(v))=(\ug\setminus \{a_{i^*}\},Z )$. 
\end{enumerate}
\end{observation}
The first condition is immediate, due to the way we factorize $Q$ in Equation \refeq{eq: q pi with d}. To see why the second condition holds, notice that $state(parent(v))$ must be a strict superset of $state(v)$; hence, $state(parent(v))$ could be either $(\ug\setminus \{a_{i^*}\},Z )$ or $(\ug\setminus \{a_{i^*},a_{i'}\},Z\cup\{ a\} )$ for $a\in \ul\setminus Z$, but then it would contribute to $d^\pi_{Z \cup \{a\}}$, namely, to another summand in Equation \refeq{eq: q pi with d}. 

Denote by $V$ the set of all nodes that satisfy the conditions of Observation \ref{obs: for d}. Due to the way we constructed $\pi$, the paths from the root of $T(\pi)$ to any node in $V$ consist of actions that involve the arms $\{a_{i^*},a_{i'},a_{j(Z)}\}\cup (\ul \setminus Z)$ solely; hence, we can focus on the \text{off-path} tree whose root is $s_0'\defeq\{a_{i^*},a_{i'}\}\cup (\ul \setminus Z)\cup \{a_{j(Z)}\}$, and the actions are precisely as in the tree induced by $\pi$ (according to the order of $\pi$). Denote this new tree by $T'$, and let
\[
V' \defeq \left\{v\in nodes(T)\mid state(v)=\{a_{j(Z)}\}  \right\}.
\] 
Due to this construction, 
\begin{observation}\label{obs: for d two}
The coefficient $d^\pi_Z$ is the probability to get to a node that belongs to $V'$  in $T'$.
\end{observation}
The observation follows from the one-to-one correspondence between the nodes and path in $T(\pi)$ and their counterparts in the off-path tree $T'$. 

In $T'$,  $Q(\pi,s'_0)$ is the probability of starting at $s'_0$ and reaching the leaf with no arms (terminal state $\emptyset$), i.e., exploring $\ul \setminus Z$ \textit{and} $a_{j(Z)}$. In contrast, $Q(\pi,\{a_{i^*},a_{i'}\}\cup (\ul \setminus Z))$ is the probability of starting at $s'_0$ and reaching a node (internal or terminal) $v$ with $state(v)= \cap \left(\ul \setminus Z\right) =\emptyset$, namely, exploring $\ul \setminus Z$. Such a node $v$ leads to a leaf with probability 1; hence, paths from $v$ terminate in leaves corresponding to either state $\emptyset$ or  $\{a_{j(Z)}\}$. Consequently,
\[
Q(\pi,\{a_{i^*},a_{i'}\}\cup (\ul \setminus Z)) - Q(\pi,\{a_{i^*},a_{i'}\}\cup (\ul \setminus Z)\cup \{ a_{j(Z)}\})
\]
is the probability of starting at $s_0'$, and reaching a terminal node that belongs to $V'$.
\end{proofof}

\section{Base Cases for Lemma \ref{lemma:equivalence}}\label{sec:base for lemma}
\begin{proposition}\label{prop:case of one strong}
Let $\abs{\ug}=1$ and $\abs{\ul} \geq 2$. For any pair of policies $\pi,\rho$, it holds that $Q(\pi,\ug,\ul)=Q(\rho,\ug,\ul)$.
\end{proposition}
\begin{proofof}{Proposition \ref{prop:case of one strong}}
Let $\tilde \mu(a)\defeq \abs{\mu(a)}$, and denote $\ug=\{a_{i_1}\}$  and $\ul=\{a_{j_1},\dots a_{j_k}\}$ for $k=\abs{\ul}$. The probability of reaching the empty terminal state under any $\mP$-valid policy is
\begin{align*}
\prod_{l=1}^k\frac{\tilde \mu(a_{j_l})}{\tilde \mu(a_{j_l})+\tilde \mu(a_{i_1})},
\end{align*}
i.e., the probability of successfully exploring $\ul$. Due to multiplication associativity, the above expression is invariant of the way we order its elements. Finally, by definition of $Q$, this implies that $Q(\pi,\ug,\ul)=Q(\rho,\ug,\ul)$.
\end{proofof}

\begin{proposition}\label{prop:case of one}
Let $\abs{\ug}\geq 2$ and $\abs{\ul} = 1$. For any pair of policies $\pi,\rho$, it holds that $Q(\pi,\ug,\ul)=Q(\rho,\ug,\ul)$.
\end{proposition}
\begin{proofof}{Proposition \ref{prop:case of one}}
Let $\tilde \mu(a)\defeq \abs{\mu(a)}$, and denote $\ug=\{a_{i_1},\dots a_{i_k}\}$ for $k=\abs{\ug}$ and $\ul=\{a_{j_1}\}$. The probability of reaching the terminal state $(a_{j_1})$ under any $\mP$-valid policy is
\begin{align*}
\prod_{l=1}^k\frac{\tilde \mu(a_{i_l})}{\tilde \mu(a_{i_l})+\tilde \mu(a_{j_1})},
\end{align*}
i.e., the probability of failing to explore $a_{j_1}$. Due to multiplication associativity, the above expression is invariant of the way we order its elements. Finally, by definition of $Q$, this implies that $1-Q(\pi,\ug,\ul)=1-Q(\rho,\ug,\ul)$; hence, $Q(\pi,\ug,\ul)=Q(\rho,\ug,\ul)$
\end{proofof}

\begin{proposition}\label{prop:case of two strong}
Let $U$ be an arbitrary state, such that $\ug\defeq\above(U)=2$ and $\ul\defeq\below(U) \geq 2$. For any pair of $\mP$-valid policies $\pi$ and $\rho$, it holds that $Q(\pi,U)=Q(\rho,U)$.
\end{proposition}
\begin{proofof}{Proposition \ref{prop:case of two strong}}
We prove the claim by induction, with Proposition \ref{prop:case of one} serving as the base case. Assume the claim holds for $\abs{\ul}=k-1$. It is enough to show that if $\abs{\ul}=k$, for any $a_i\in \ug,a_j\in \ul$, $Q^*_{i,j}(U)=Q^*(U)$. Assume that $Q^*_{{i_1},{j_1}}(U)=Q^*(U)$, and fix any $a_{i'}\in \ug, a_{j'}\in \ul$. 
\paragraph{Remark} We do not use Assumption \ref{assumption:dominance} here.
\paragraph{Step 1} Assume that $i'=i_1$ and $j' \neq j_1$. W.l.o.g. $j'=j_2$. We construct two policies, $\pi$ that ordered $\ul$ as $\sigr_\pi=(a_{j_1},a_{j_2},\dots, a_{j_{k}})$, and $\rho$ that orders $\ul$ as $\sigr_\rho=(a_{j_2},a_{j_1},\dots,a_{j_{k}})$. Both policies order $\ug$ according to $\sigl_\pi=\sigl_\rho=(a_{i_1},a_{i_2})$. Due to the inductive step and our assumption that  $Q^*_{i_1,j_1}(U)=Q^*(U)$, we have that $Q(\pi,U)=Q^*(U)$, and
\begin{align*}
&Q(\pi,\ug,\ul) = \underbrace{\prod_{l=1}^{k} \bl p_{{i_1},{j_l}}(a_{j_l})}_{\lambda(\pi)}
+ \underbrace{\sum_{f=1}^{k} \left(\prod_{l=1}^{f-1} \bl p_{{i_1},{j_l}}(a_{j_l})\right) \bl p_{{i_1},{j_f}}(a_{i_1}) \left(\prod_{l=f}^{k} \bl p_{{i_2},{j_l}}(a_{j_l})\right)}_{\delta(\pi)}.
\end{align*}
Notice that $\lambda(\pi)=\lambda(\rho)$. In addition, it holds that
{\thinmuskip=0mu
\medmuskip=0mu plus 0mu minus 0mu
\thickmuskip=0mu plus 0mu
\begin{align}\label{eq:ind 2 step 1}
\delta(\pi)&=\bl p_{{i_1},{j_1}}(a_{i_1})\prod_{l=1}^{k} \bl p_{{i_2},{j_l}}(a_{j_l})\nonumber\\
&\qquad +\bl p_{{i_1},{j_1}}(a_{j_1})\sum_{f=2}^{k} \left(\prod_{l=2}^{f-1} \bl p_{{i_1},{j_l}}(a_{j_l})\right) \bl p_{{i_1},{j_f}}(a_{i_1}) \left(\prod_{l=f}^{k} \bl p_{{i_2},{j_l}}(a_{j_l})\right)\nonumber\\
& =\bl p_{{i_1},{j_1}}(a_{i_1})\bl p_{{i_2},{j_1}}(a_{j_1})\bl p_{{i_2},{j_2}}(a_{j_2})\prod_{l=3}^{k} \bl p_{{i_2},{j_l}}(a_{j_l})\nonumber\\
&\qquad+ \bl p_{{i_1},{j_1}}(a_{j_1})\Bigg[\bl p_{{i_1},{j_2}}(a_{i_1})\bl p_{{i_2},{j_2}}(a_{j_2})\prod_{l=3}^{k} \bl p_{{i_2},{j_l}}(a_{j_l})\nonumber\\
&\qquad \qquad+\bl p_{{i_1},{j_2}}(a_{j_2})  \sum_{f=3}^{k} \left(\prod_{l=3}^{f-1} \bl p_{{i_1},{j_l}}(a_{j_l})\right) \bl p_{{i_1},{j_f}}(a_{i_1}) \left(\prod_{l=f}^{k} \bl p_{{i_2},{j_l}}(a_{j_l})\right) \Bigg]\nonumber\\
&=\bl p_{{i_2},{j_2}}(a_{j_2})\left( \bl p_{{i_1},{j_1}}(a_{i_1})\bl p_{{i_2},{j_1}}(a_{j_1})+\bl p_{{i_1}
,{j_1}}(a_{j_1})\bl p_{{i_1},{j_2}}(a_{i_1})\right)\prod_{l=3}^{k} \bl p_{{i_2},{j_l}}(a_{j_l})\nonumber\\
&\qquad+ \bl p_{{i_1},{j_1}}(a_{j_1})\bl p_{{i_1},{j_2}}(a_{j_2})  \sum_{f=3}^{k} \left(\prod_{l=3}^{f-1} \bl p_{{i_1},{j_l}}(a_{j_l})\right) \bl p_{{i_1},{j_f}}(a_{i_1}) \left(\prod_{l=f}^{k} \bl p_{{i_2},{j_l}}(a_{j_l})\right).
\end{align}}%
We show that
\begin{claim}\label{claim:triplets}
It holds that
\begin{align*}
&\bl p_{{i_2},{j_2}}(a_{j_2})\left( \bl p_{{i_1},{j_1}}(a_{i_1})\bl p_{{i_2},{j_1}}(a_{j_1})+\bl p_{{i_1},{j_1}}(a_{j_1})\bl p_{{i_1},{j_2}}(a_{i_1})\right)\\
&=\bl p_{{i_2},{j_1}}(a_{j_1})\left(\bl p_{{i_1},{j_2}}(a_{i_1})\bl p_{{i_2},{j_2}}(a_{j_2})+\bl p_{{i_1},{j_2}}(a_{j_2})\bl p_{{i_1},{j_1}}(a_{i_1})  \right).
\end{align*}
\end{claim}
Now, set $\sigma:\mathbb N \rightarrow \mathbb N$ such that $\sigma(1)=2, \sigma(2)=1$, and $\sigma(i)=i$ for $i\geq 3$; hence, using Claim \ref{claim:triplets},
\begin{align}
\textnormal{Eq. (\ref{eq:ind 2 step 1})}
& =\bl p_{{i_1},{j_2}}(a_{i_1})\bl p_{{i_2},{j_2}}(a_{j_2})\bl p_{{i_2},{j_1}}(a_{j_1})\prod_{l=3}^{k} \bl p_{{i_2},{j_l}}(a_{j_l})\nonumber\\
&\qquad+ \bl p_{{i_1},{j_2}}(a_{j_2}) \Bigg[\bl p_{{i_1},{j_1}}(a_{i_1})\bl p_{{i_2},{j_1}}(a_{j_1}) \prod_{l=3}^{k} \bl p_{{i_2},{j_l}}(a_{j_l})\nonumber\\
&\qquad \qquad + \bl p_{{i_1},{j_1}}(a_{j_1})\sum_{f=3}^{k} \left(\prod_{l=3}^{f-1} \bl p_{{i_1},{j_l}}(a_{j_l})\right) \bl p_{{i_1},{j_f}}(a_{i_1}) \left(\prod_{l=f}^{k} \bl p_{{i_2},{j_l}}(a_{j_l})\right) \Bigg]\nonumber\\
&= \sum_{f=1}^{k} \left(\prod_{l=1}^{f-1} \bl p_{{i_1},{j_{\sigma(l)}}}(a_{j_{\sigma(l)}})\right) \bl p_{{i_1},{j_{\sigma (f)}}}(a_{i_1}) \left(\prod_{l=f}^{k} \bl p_{{i_2},{j_{\sigma(l)}}}(a_{j_{\sigma(l)}})\right)\nonumber\\
&=\delta(\rho)
\end{align}
\paragraph{Step 2} Assume that $i'=i_2\neq i_1$ and $j'=j_1$. We construct two policies, $\pi$ that orders $\ug$ as $\sigl_\pi=(a_{i_1},a_{i_2})$, and $\rho$ that orders $\ug$ as  $\sigl_\rho=(a_{i_2},a_{i_1})$. Both policies order $\ul$ by $\sigr_\pi=\sigr_\rho=(a_{j_1},a_{j_2}, \dots, a_{j_k})$. In addition, we introduce a third policy, $\tilde \rho$, that has the same order as $\rho$ on $\ug$, namely $\sigl_{\tilde \rho}=\sigl_\rho=(a_{i_2},a_{i_1})$, and orders $\ul$ by $\sigr_{\tilde \rho}=(a_{j_k},a_{j_2},\dots ,a_{j_{k-1}},a_{j_1})$.
It holds that
{\thinmuskip=0mu
\medmuskip=0mu plus 0mu minus 0mu
\thickmuskip=0mu plus 0mu
\begin{align}\label{eq: multiple transitions}
&Q(\pi,\ug,\ul)=\prod_{l=1}^k \bl p_{{i_1},{j_l}}(a_{j_l})+\sum_{f=1}^{k} \left(\prod_{l=1}^{f-1} \bl p_{{i_1},{j_l}}(a_{j_l})\right) \bl p_{{i_1},{j_f}}(a_{i_1}) \left(\prod_{l=f}^{k} \bl p_{{i_2},{j_l}}(a_{j_l})\right)\nonumber\\
&=\bl p_{{i_1},{j_1}}(a_{j_1})\bl p_{{i_1},{j_k}}(a_{j_k})\overbrace{\prod_{l=2}^{k-1} \bl p_{{i_1},{j_l}}(a_{j_l})}^{I_1}+\bl p_{{i_2},{j_k}}(a_{j_k})\Bigg[\bl p_{{i_1},{j_1}}(a_{j_1})\bl p_{{i_1},{j_k}}(a_{i_1})\underbrace{\prod_{l=2}^{k-1} \bl p_{{i_1},{j_l}}(a_{j_l})}_{I_1}\nonumber\\
&\quad+\underbrace{\sum_{f=1}^{k-1} \left(\prod_{l=1}^{f-1} \bl p_{{i_1},{j_l}}(a_{j_l})\right) \bl p_{{i_1},{j_f}}(a_{i_1}) \left(\prod_{l=f}^{k-1} \bl p_{{i_2},{j_l}}(a_{j_l})\right)}_{I_2}\Bigg]\nonumber\\
&= \bl p_{{i_1},{j_1}}(a_{j_1})\bl p_{{i_1},{j_k}}(a_{j_k})\bl p_{{i_2},{j_k}}(a_{j_k})I_1 +\bl p_{{i_1},{j_1}}(a_{j_1})\bl p_{{i_1},{j_k}}(a_{j_k})\bl p_{{i_2},{j_k}}(a_{i_2})I_1\nonumber\\
&\qquad+\bl p_{{i_2},{j_k}}(a_{j_k})\left[\bl p_{{i_1},{j_1}}(a_{j_1})\bl p_{{i_1},{j_k}}(a_{i_1})I_1+I_2\right]\nonumber\\
&=\bl p_{{i_1},{j_1}}(a_{j_1})\bl p_{{i_1},{j_k}}(a_{j_k})\bl p_{{i_2},{j_k}}(a_{i_2})I_1\nonumber\\
&\qquad +\bl p_{{i_2},{j_k}}(a_{j_k})\left[\bl p_{{i_1},{j_1}}(a_{j_1})\bl p_{{i_1},{j_k}}(a_{j_k})I_1+\bl p_{{i_1},{j_1}}(a_{j_1})\bl p_{{i_1},{j_k}}(a_{i_1})I_1+I_2\right]\nonumber\\
&=\bl p_{{i_1},{j_1}}(a_{j_1})\bl p_{{i_1},{j_k}}(a_{j_k})\bl p_{{i_2},{j_k}}(a_{i_2})I_1+\bl p_{{i_2},{j_k}}(a_{j_k})\left[\bl p_{{i_1},{j_1}}(a_{j_1})I_1+I_2\right].
\end{align}}%
Next, observe that
{\thinmuskip=0mu
\medmuskip=0mu plus 0mu minus 0mu
\thickmuskip=0mu plus 0mu
\begin{equation}\label{eq:ind step with k-1}
\left[\bl p_{{i_1},{j_1}}(a_{j_1})I_1+I_2\right]=\prod_{l=1}^{k-1} \bl p_{{i_1},{j_l}}(a_{j_l})+\sum_{f=1}^{k-1} \left(\prod_{l=1}^{f-1} \bl p_{{i_1},{j_l}}(a_{j_l})\right) \bl p_{{i_1},{j_f}}(a_{i_1}) \left(\prod_{l=f}^{k-1} \bl p_{{i_2},{j_l}}(a_{j_l})\right).
\end{equation}}%
Notice that the latter is precisely $Q(\pi,\left\{a_{i_1},a_{i_2},a_{j_1},\dots a_{j_{k-1}}\right\})$; thus, the inductive step implies that it is order invariant. Let $\sigma:\mathbb N \rightarrow \mathbb N$ such that $\sigma(1)=k,\sigma(k)=1$, and $\sigma(i)=i$ for $1<i<k$. Since $Q(\pi,\left\{a_{i_1},a_{i_2},a_{j_1},\dots a_{j_{k-1}}\right\})=Q(\tilde \rho,\left\{a_{i_1},a_{i_2},a_{j_1},\dots a_{j_{k-1}}\right\})$, we conclude that the expression in Equation (\ref{eq:ind step with k-1}) equals 
\[
\prod_{l=2}^{k} \bl p_{{i_2},{j_{\sigma(l)}}}(a_{j_{\sigma(l)}})+\sum_{f=2}^{k} \left(\prod_{l=2}^{f-1} \bl p_{{i_2},{j_{\sigma(l)}}}(a_{j_{\sigma(l)}})\right) \bl p_{{i_2},{j_{\sigma(f)}}}(a_{i_2}) \left(\prod_{l=f}^{k} \bl p_{{i_1},{j_{\sigma(l)}}}(a_{j_{\sigma(l)}})\right).
\]
Combining this with Equation (\ref{eq: multiple transitions}),
{\thinmuskip=0mu
\medmuskip=0mu plus 0mu minus 0mu
\thickmuskip=0mu plus 0mu
\begin{align}
&\textnormal{Eq. }\refeq{eq: multiple transitions}=\bl p_{{i_1},{j_1}}(a_{j_1})\bl p_{{i_1},{j_k}}(a_{j_k})\bl p_{{i_2},{j_k}}(a_{i_2})\prod_{l=2}^{k-1} \bl p_{{i_1},{j_{\sigma(l)}}}(a_{j_{\sigma(l)}})\nonumber\\
&\quad +\bl p_{{i_2},{j_k}}(a_{j_k})\Bigg[
\prod_{l=2}^{k} \bl p_{{i_2},{j_{\sigma(l)}}}(a_{j_{\sigma(l)}})+\sum_{f=2}^{k} \left(\prod_{l=2}^{f-1} \bl p_{{i_2},{j_{\sigma(l)}}}(a_{j_{\sigma(l)}})\right) \bl p_{{i_2},{j_{\sigma(f)}}}(a_{i_2}) \left(\prod_{l=f}^{k} \bl p_{{i_1},{j_{\sigma(l)}}}(a_{j_{\sigma(l)}})\right)
\Bigg]\nonumber\\
&=\bl p_{{i_2},{j_k}}(a_{i_2})\prod_{l=1}^{k} \bl p_{{i_1},{j_{\sigma(l)}}}(a_{j_{\sigma(l)}})+\bl p_{{i_2},{j_k}}(a_{j_k})
\prod_{l=2}^{k} \bl p_{{i_2},{j_{\sigma(l)}}}(a_{j_{\sigma(l)}})\nonumber\\
&\quad +\bl p_{{i_2},{j_k}}(a_{j_k})\left[
\sum_{f=2}^{k} \left(\prod_{l=2}^{f-1} \bl p_{{i_2},{j_{\sigma(l)}}}(a_{j_{\sigma(l)}})\right) \bl p_{{i_2},{j_{\sigma(f)}}}(a_{i_2}) \left(\prod_{l=f}^{k} \bl p_{{i_1},{j_{\sigma(l)}}}(a_{j_{\sigma(l)}})\right)
\right]\nonumber\\
&=\prod_{l=1}^{k} \bl p_{{i_2},{j_{\sigma(l)}}}(a_{j_{\sigma(l)}})+
\sum_{f=1}^{k} \left(\prod_{l=1}^{f-1} \bl p_{{i_2},{j_{\sigma(l)}}}(a_{j_{\sigma(l)}})\right) \bl p_{{i_2},{j_{\sigma(f)}}}(a_{i_2}) \left(\prod_{l=f}^{k} \bl p_{{i_1},{j_{\sigma(l)}}}(a_{j_{\sigma(l)}})\right)\nonumber\\
&=Q(\tilde \rho ,U),
\end{align}}%
where the last equality follows from the definition of $\tilde \rho$ (orders the arms precisely so). Finally, $Q(\tilde \rho, U)=Q(\rho ,U)$ follows from the previous Step 1.
\paragraph{Step 3}
The two previous steps imply that for any $a_{i'}\in \ug, a_{j'}\in \ul$, it holds that
\[
Q_{i',j'}(U)=Q_{i',j_1}(U)=Q_{i_1,j_1}(U).
\]
This completes the proof of Proposition \ref{prop:case of two strong}.
\end{proofof}
\begin{proofof}{Claim \ref{claim:triplets}}
To ease readability, let $\tmu{i}\defeq\abs{\mu(a_i)}$ for every $a_i\in A$. It holds that
\begin{align*}
&\bl p_{{i_2},{j_2}}({j_2})\left( \bl p_{{i_1},{j_1}}({i_1})\bl p_{{i_2},{j_1}}({j_1})+\bl p_{{i_1},{j_1}}({j_1})\bl p_{{i_1},{j_2}}({i_1})\right)\\
&=\frac{\tmu{i_2}}{\tmu{i_2}+\tmu{j_2}}\left( \frac{\tmu{j_1}}{\tmu{i_1}+\tmu{j_1}}\frac{\tmu{i_2}}{\tmu{i_2}+\tmu{j_1}}+\frac{\tmu{i_1}}{\tmu{i_1}+\tmu{j_1}}\frac{\tmu{j_2}}{\tmu{i_1}+\tmu{j_2}}\right) \\
&= \frac{\tmu{i_2}\tmu{j_1}\tmu{i_2}(\tmu{i_1}+\tmu{j_2})+\tmu{i_2}\tmu{i_1}\tmu{j_2}(\tmu{i_2}+\tmu{j_1})}{(\tmu{i_1}+\tmu{j_1})(\tmu{i_1}+\tmu{j_2})(\tmu{i_2}+\tmu{j_1})(\tmu{i_2}+\tmu{j_2})}\\
&= \frac{\overbrace{\tmu{i_2}\tmu{j_1}\tmu{i_2}\tmu{i_1}}^{I}+\overbrace{\tmu{i_2}\tmu{j_1}\tmu{i_2}\tmu{j_2}}^{II}+\overbrace{\tmu{i_2}\tmu{i_1}\tmu{j_2}\tmu{i_2}}^{III}+\overbrace{\tmu{i_2}\tmu{i_1}\tmu{j_2}\tmu{j_1}}^{IV}}{(\tmu{i_1}+\tmu{j_1})(\tmu{i_1}+\tmu{j_2})(\tmu{i_2}+\tmu{j_1})(\tmu{i_2}+\tmu{j_2})}\\
&= \frac{\overbrace{\tmu{i_2}\tmu{j_2}\tmu{i_2}\tmu{i_1}}^{III}+\overbrace{\tmu{i_2}\tmu{j_2}\tmu{i_2}\tmu{j_1}}^{II}+\overbrace{\tmu{i_2}\tmu{i_1}\tmu{j_1}\tmu{i_2}}^{I}+\overbrace{\tmu{i_2}\tmu{i_1}\tmu{j_1}\tmu{j_2}}^{IV}}{(\tmu{i_1}+\tmu{j_1})(\tmu{i_1}+\tmu{j_2})(\tmu{i_2}+\tmu{j_1})(\tmu{i_2}+\tmu{j_2})}\\
&= \frac{\tmu{i_2}\tmu{j_2}\tmu{i_2}(\tmu{i_1}+\tmu{j_1})+\tmu{i_2}\tmu{i_1}\tmu{j_1}(\tmu{i_2}+\tmu{j_2})}{(\tmu{i_1}+\tmu{j_1})(\tmu{i_1}+\tmu{j_2})(\tmu{i_2}+\tmu{j_1})(\tmu{i_2}+\tmu{j_2})}\\
&=\bl p_{{i_2},{j_1}}({j_1})\left(\bl p_{{i_1},{j_2}}({i_1})\bl p_{{i_2},{j_2}}({j_2})+\bl p_{{i_1},{j_2}}({j_2})\bl p_{{i_1},{j_1}}({i_1})  \right)
\end{align*}
\end{proofof}

\begin{proposition}\label{prop:case of two}
Let $U$ be an arbitrary state, such that $\ug\defeq\above(U)\geq 2$ and $\ul\defeq\below(U) = 2$. For any pair of $\mP$-valid policies $\pi$ and $\rho$, it holds that $Q(\pi,U)=Q(\rho,U)$.
\end{proposition}
\begin{proofof}{Proposition \ref{prop:case of two}}
The proof of this proposition goes along the lines of the proof of Proposition \ref{prop:case of two strong}, but we provide the details here for completeness. For simplicity, we let $\overline Q = 1-Q$, and prove that for any two policies $\pi, \rho$ it holds that $\overline Q(\pi,U)=\overline Q(\rho,U)$.

We prove the claim by induction, with Proposition \ref{prop:case of one strong} serving as the base case. Assume the claim holds for $\abs{\ug}=k-1$. It is enough to show that if $\abs{\ug}=k$, for any $a_i\in \ug,a_j\in \ul$, $Q^*_{i,j}(U)=Q^*(U)$. Assume that $Q^*_{{i_1},{j_1}}(U)=Q^*(U)$, and define a policy $\pi$ such that $\sigr_\pi=(a_{j_1},a_{j_2})$ and $\sigl_\pi=(a_{i_1},a_{i_2},\dots a_{i_k})$. Next, fix any $a_{i'}\in \ug, a_{j'}\in \ul$. 
\paragraph{Remark} We do not use Assumption \ref{assumption:dominance} here.
\paragraph{Step 1} Assume that $i'\neq i_1$ and $j' = j_1$. W.l.o.g. $i'=i_2$. We construct the ordered policy $\rho$ that orders $\ul$ by $\sigr_\rho=\sigr_\pi=(a_{j_1},a_{j_2})$ and $\ug$ by $\sigl_\rho=(a_{i_2},a_{i_1},\dots, a_{i_k})$. Due to the inductive step and our assumption that  $Q^*_{i_1,j_1}(U)=Q^*(U)$, we have that $\overline Q(\pi,U)=\overline Q^*(U)$. For brevity, we introduce the following notations. For $r\in \{1,2,3\}$, let
{\thinmuskip=.2mu
\medmuskip=0mu plus .2mu minus .2mu
\thickmuskip=1mu plus 1mu
\[
\lambda^r_{j_1}=\prod_{l=r}^{k} \bl p_{{i_l},{j_1}}(a_{i_l}),\lambda^r_{j_2}=\prod_{l=r}^{k} \bl p_{{i_l},{j_2}}(a_{i_l}), \delta^r= \sum_{f=r}^{k} \left(\prod_{l=r}^{f-1} \bl p_{{i_l},{j_1}}(a_{i_l})\right) \bl p_{{i_f},{j_1}}(a_{j_1}) \left(\prod_{l=f}^{k} \bl p_{{i_l},{j_2}}(a_{i_l})\right).
\]
}%
Observe that
\begin{align}\label{eq:from pi to ghjghj}
\overline Q(\pi,U)&=\lambda^1_{j_1}+\delta^1=\lambda^1_{j_1}+\bl p_{i_1,j_1}(a_{j_1})\lambda^1_{j_2}+\bl p_{i_1,j_1}(a_{i_1})\delta^2\nonumber\\
&=\lambda^1_{j_1}+\bl p_{i_1,j_1}(a_{j_1})\bl p_{i_1,j_2}(a_{i_1})\bl p_{i_2,j_2}(a_{i_2})\lambda^3_{j_2}\nonumber\\
&\qquad \qquad+\bl p_{i_1,j_1}(a_{i_1})\left[
\bl p_{i_2,j_1}(a_{j_1})\bl p_{i_2,j_2}(a_{i_2})\lambda^3_{j_2}+\bl p_{i_2,j_1}(a_{i_2})\delta^3\right].
\end{align}
Next, we show that 
\begin{claim}\label{claim:triplets additional}
It holds that
\begin{align*}
&\bl p_{{i_2},{j_2}}(a_{i_2})\left( \bl p_{{i_1},{j_1}}(a_{j_1})\bl p_{{i_1},{j_2}}(a_{i_1})+\bl p_{{i_1},{j_1}}(a_{i_1})\bl p_{{i_2},{j_1}}(a_{j_1})\right)\\
&=\bl p_{{i_1},{j_2}}(a_{i_1})\left(\bl p_{{i_2},{j_1}}(a_{j_1})\bl p_{{i_2},{j_2}}(a_{i_2})+\bl p_{{i_2},{j_1}}(a_{i_2})\bl p_{{i_1},{j_1}}(a_{j_1})  \right).
\end{align*}
\end{claim}
Combining Equation \refeq{eq:from pi to ghjghj} and Claim \ref{claim:triplets additional}, we get
\begin{align}\label{eq: flipping}
\textnormal{Eq. } \refeq{eq:from pi to ghjghj}
&=\lambda^1_{j_1}+\bl p_{i_2,j_1}(a_{j_1})\bl p_{i_2,j_2}(a_{i_2})\bl p_{i_1,j_2}(a_{i_1})\lambda^3_{j_2}\nonumber\\
&\qquad \qquad+\bl p_{i_2,j_1}(a_{i_2})\left[
\bl p_{i_1,j_1}(a_{j_1})\bl p_{i_1,j_2}(a_{i_1})\lambda^3_{j_2}+\bl p_{i_1,j_1}(a_{i_1})\delta^3\right]\nonumber \\
&=\lambda^1_{j_1}+\bl p_{i_2,j_1}(a_{j_1})\lambda^1_{j_2}+\bl p_{i_2,j_1}(a_{i_2})\left[
\bl p_{i_1,j_1}(a_{j_1})\bl p_{i_1,j_2}(a_{i_1})\lambda^3_{j_2}+\bl p_{i_1,j_1}(a_{i_1})\delta^3\right]
\end{align}
Now, set $\sigma:\mathbb N \rightarrow \mathbb N$ such that $\sigma(1)=2, \sigma(2)=1$, and $\sigma(i)=i$ for $i\geq 3$; hence,
{\thinmuskip=.2mu
\medmuskip=0mu plus .2mu minus .2mu
\thickmuskip=1mu plus 1mu
\begin{align*}
\textnormal{Eq. }\refeq{eq: flipping}&=\lambda^1_{j_1}+\bl p_{i_2,j_1}(a_{j_1})\lambda^1_{j_2}\\
&\qquad \qquad +\bl p_{i_2,j_1}(a_{i_2})\left[ 
\sum_{f=2}^{k} \left(\prod_{l=2}^{f-1} \bl p_{{i_{\sigma(l)}},{j_1}}(a_{i_{\sigma(l)}})\right) \bl p_{{i_{\sigma (f)}},{j_1}}(a_{j_1}) \left(\prod_{l=f}^{k} \bl p_{{i_{\sigma(l)}},{j_2}}(a_{i_{\sigma(l)}})\right)
 \right] \\
&= \lambda^1_{j_1}+ \sum_{f=1}^{k} \left(\prod_{l=1}^{f-1} \bl p_{{i_{\sigma(l)}},{j_1}}(a_{i_{\sigma(l)}})\right) \bl p_{{i_{\sigma (f)}},{j_1}}(a_{j_1}) \left(\prod_{l=f}^{k} \bl p_{{i_{\sigma(l)}},{j_2}}(a_{i_{\sigma(l)}})\right)\\
&= \overline Q(\rho,U).
\end{align*}
}
This concludes Step 1.
\paragraph{Step 2} Assume that $i'=i_1$ and $j' = j_2 \neq j_1$. We construct two ordered policies, $\rho$ and $\tilde \rho$ such that $\sigr_\rho=\sigr_{\tilde \rho}=(a_{j_2},a_{j_1})$ and $\sigl_\rho=(a_{i_1},a_{i_2},\dots, a_{i_k})$, $\sigl_{\tilde \rho}=(a_{i_k},a_{i_2},\dots, a_{i_{k-1}},a_{i_1})$. The previous Step 1 implies that $\overline Q(\rho,U)=\overline Q(\tilde \rho, U)$; thus, it suffices to show that $\overline Q(\pi,U)=\overline Q(\tilde \rho, U)$. Notice that
{\thinmuskip=.2mu
\medmuskip=0mu plus .2mu minus .2mu
\thickmuskip=1mu plus 1mu
\begin{align}\label{jnknkfdas}
\overline Q(\pi,U)&=\bl p_{i_k,j_1}(a_{i_k})\overbrace{\prod_{l=1}^{k-1} \bl p_{{i_l},{j_1}}(a_{i_l})}^{I_1}\nonumber \\
&+\bl p_{i_k,j_2}(a_{i_k})\left[\bl p_{i_k,j_1}(a_{j_1})\underbrace{\prod_{l=1}^{k-1} \bl p_{{i_l},{j_1}}(a_{i_l})}_{I_1}+\underbrace{\sum_{f=1}^{k-1} \left(\prod_{l=1}^{f-1} \bl p_{{i_l},{j_1}}(a_{i_l})\right) \bl p_{{i_f},{j_1}}(a_{j_1}) \left(\prod_{l=f}^{k-1} \bl p_{{i_l},{j_2}}(a_{i_l})\right)}_{I_2}  \right]
\end{align}}%
Rearranging,
\begin{align}\label{eq:kijnjvfvs}
\textnormal{Eq. }\refeq{jnknkfdas} &=\bl p_{i_k,j_1}(a_{i_k})I_1+\bl p_{i_k,j_2}(a_{i_k})\left[\bl p_{i_k,j_1}(a_{j_1})I_1+I_2  \right]. \nonumber \\
&= \bl p_{i_k,j_1}(a_{i_k})\left( \bl p_{i_k,j_2}(a_{i_k})+\bl p_{i_k,j_2}(a_{j_2}) \right) I_1+\bl p_{i_k,j_2}(a_{i_k})\left[\bl p_{i_k,j_1}(a_{j_1})I_1+I_2  \right]. \nonumber \\
&= \bl p_{i_k,j_1}(a_{i_k})\bl p_{i_k,j_2}(a_{j_2})I_1+\bl p_{i_k,j_2}(a_{i_k})\left[I_1+I_2  \right]. 
\end{align}
Recall that
\begin{align} \label{porenkjnf}
I_1+I_2=\prod_{l=1}^{k-1} \bl p_{{i_l},{j_1}}(a_{i_l})+\sum_{f=1}^{k-1} \left(\prod_{l=1}^{f-1} \bl p_{{i_l},{j_1}}(a_{i_l})\right) \bl p_{{i_f},{j_1}}(a_{j_1}) \left(\prod_{l=f}^{k-1} \bl p_{{i_l},{j_2}}(a_{i_l})\right),
\end{align}
which is precisely $\overline Q(\pi,\left\{a_{i_1},\dots, a_{i_{k-1}},a_{j_1},a_{j_2}\right\})$; thus, the inductive step implies that it is order invariant. Let $\sigma:\mathbb N \rightarrow \mathbb N$ such that $\sigma(1)=k,\sigma(k)=1$, and $\sigma(i)=i$ for $1<i<k$. Since $\overline Q(\pi,\left\{a_{i_1},\dots, a_{i_{k-1}},a_{j_1},a_{j_2}\right\})=\overline Q(\tilde \rho,\left\{a_{i_1},\dots, a_{i_{k-1}},a_{j_1},a_{j_2}\right\}))$, we conclude that the expression in Equation (\ref{porenkjnf}) equals 
\[
\prod_{l=2}^{k} \bl p_{{\sigma (i_l)},{j_2}}(a_{\sigma (i_l)})+\sum_{f=2}^{k} \left(\prod_{l=2}^{f-1} \bl p_{{\sigma (i_l)},{j_2}}(a_{\sigma (i_l)})\right) \bl p_{{\sigma (i_f)},{j_2}}(a_{j_2}) \left(\prod_{l=f}^{k} \bl p_{{\sigma (i_l)},{j_1}}(a_{\sigma (i_l)})\right).
\]
Notice that in the above expression, we first try to explore $a_{j_2}$, and only then $a_{j_1}$. Combining this with Equation \refeq{eq:kijnjvfvs},
{\thinmuskip=.2mu
\medmuskip=0mu plus .2mu minus .2mu
\thickmuskip=1mu plus 1mu
\begin{align*}
&\textnormal{Eq. } \refeq{eq:kijnjvfvs}=\bl p_{i_k,j_1}(a_{i_k})\bl p_{i_k,j_2}(a_{j_2})I_1\\
&\quad + \bl p_{i_k,j_2}(a_{i_k})\left[\prod_{l=2}^{k} \bl p_{{\sigma (i_l)},{j_2}}(a_{\sigma (i_l)})+\sum_{f=2}^{k} \left(\prod_{l=2}^{f-1} \bl p_{{\sigma (i_l)},{j_2}}(a_{\sigma (i_l)})\right) \bl p_{{\sigma (i_f)},{j_2}}(a_{j_2}) \left(\prod_{l=f}^{k} \bl p_{{\sigma (i_l)},{j_1}}(a_{\sigma (i_l)})\right) \right]. \\
&= \bl p_{i_k,j_1}(a_{i_k})\bl p_{i_k,j_2}(a_{j_2})\prod_{l=1}^{k-1} \bl p_{{i_l},{j_1}}(a_{i_l})+\bl p_{i_k,j_2}(a_{i_k})\prod_{l=2}^{k} \bl p_{{\sigma (i_l)},{j_2}}(a_{\sigma (i_l)})\\
&\quad + \bl p_{i_k,j_2}(a_{i_k})\left[\sum_{f=2}^{k} \left(\prod_{l=2}^{f-1} \bl p_{{\sigma (i_l)},{j_2}}(a_{\sigma (i_l)})\right) \bl p_{{\sigma (i_f)},{j_2}}(a_{j_2}) \left(\prod_{l=f}^{k} \bl p_{{\sigma (i_l)},{j_1}}(a_{\sigma (i_l)})\right) \right]\\
&=\bl p_{i_k,j_2}(a_{j_2})\prod_{l=1}^{k} \bl p_{{\sigma(i_l)},{j_1}}(a_{\sigma(i_l)})+\prod_{l=1}^{k} \bl p_{{\sigma (i_l)},{j_2}}(a_{\sigma (i_l)})\\
&\quad + \bl p_{i_k,j_2}(a_{i_k})\left[\sum_{f=2}^{k} \left(\prod_{l=2}^{f-1} \bl p_{{\sigma (i_l)},{j_2}}(a_{\sigma (i_l)})\right) \bl p_{{\sigma (i_f)},{j_2}}(a_{j_2}) \left(\prod_{l=f}^{k} \bl p_{{\sigma (i_l)},{j_1}}(a_{\sigma (i_l)})\right) \right]\\
&=\prod_{l=1}^{k} \bl p_{{\sigma (i_l)},{j_2}}(a_{\sigma (i_l)})+\sum_{f=1}^{k} \left(\prod_{l=1}^{f-1} \bl p_{{\sigma (i_l)},{j_2}}(a_{\sigma (i_l)})\right) \bl p_{{\sigma (i_f)},{j_2}}(a_{j_2}) \left(\prod_{l=f}^{k} \bl p_{{\sigma (i_l)},{j_1}}(a_{\sigma (i_l)})\right), 
\end{align*}}%
and the latter is precisely $\overline Q(\tilde \rho, U)$.
\paragraph{Step 3}
The two previous steps imply that for any $a_{i'}\in \ug, a_{j'}\in \ul$, it holds that
\[
\overline Q_{i',j'}(U)=\overline Q_{i',j_1}(U)= \overline Q_{i_1,j_1}(U).
\]
This completes the proof of Proposition \ref{prop:case of two}.
\end{proofof}
\begin{proofof}{Claim \ref{claim:triplets additional}}
To ease readability, let $\tmu{i}\defeq\abs{\mu(a_i)}$ for every $a_i\in A$. It holds that
\begin{align*}
&\bl p_{{i_2},{j_2}}(a_{i_2})\left( \bl p_{{i_1},{j_1}}(a_{j_1})\bl p_{{i_1},{j_2}}(a_{i_1})+\bl p_{{i_1},{j_1}}(a_{i_1})\bl p_{{i_2},{j_1}}(a_{j_1})\right)\\
&=\frac{\tmu{j_2}}{\tmu{i_2}+\tmu{j_2}}\left( \frac{\tmu{i_1}}{\tmu{i_1}+\tmu{j_1}}\frac{\tmu{j_2}}{\tmu{i_1}+\tmu{j_2}}+\frac{\tmu{j_1}}{\tmu{i_1}+\tmu{j_1}}\frac{\tmu{i_2}}{\tmu{i_2}+\tmu{j_1}}\right) \\
&= \frac{\tmu{j_2}\tmu{i_1}\tmu{j_2}(\tmu{i_2}+\tmu{j_1})+\tmu{j_2}\tmu{j_1}\tmu{i_2}(\tmu{i_1}+\tmu{j_2})}{(\tmu{i_1}+\tmu{j_1})(\tmu{i_1}+\tmu{j_2})(\tmu{i_2}+\tmu{j_1})(\tmu{i_2}+\tmu{j_2})}\\
&= \frac{\overbrace{\tmu{j_2}\tmu{i_1}\tmu{j_2}\tmu{i_2}}^{I}+\overbrace{\tmu{j_2}\tmu{i_1}\tmu{j_2}\tmu{j_1}}^{II}+\overbrace{\tmu{j_2}\tmu{j_1}\tmu{i_2}\tmu{i_1}}^{III}+\overbrace{\tmu{j_2}\tmu{j_1}\tmu{i_2}\tmu{j_2}}^{IV}}{(\tmu{i_1}+\tmu{j_1})(\tmu{i_1}+\tmu{j_2})(\tmu{i_2}+\tmu{j_1})(\tmu{i_2}+\tmu{j_2})}\\
&= \frac{\overbrace{\tmu{j_2}\tmu{i_1}\tmu{j_2}\tmu{i_2}}^{I}+\overbrace{\tmu{j_2}\tmu{j_1}\tmu{i_2}\tmu{j_2}}^{IV}+\overbrace{\tmu{j_2}\tmu{j_1}\tmu{i_2}\tmu{i_1}}^{III}+\overbrace{\tmu{j_2}\tmu{i_1}\tmu{j_2}\tmu{j_1}}^{II}}{(\tmu{i_1}+\tmu{j_1})(\tmu{i_1}+\tmu{j_2})(\tmu{i_2}+\tmu{j_1})(\tmu{i_2}+\tmu{j_2})}\\
&= \frac{\tmu{j_2}\tmu{i_2}\tmu{j_2}(\tmu{i_1}+\tmu{j_1})+\tmu{j_2}\tmu{j_1}\tmu{i_1}(\tmu{i_2}+\tmu{j_2})}{(\tmu{i_1}+\tmu{j_1})(\tmu{i_1}+\tmu{j_2})(\tmu{i_2}+\tmu{j_1})(\tmu{i_2}+\tmu{j_2})}\\
&=\bl p_{{i_1},{j_2}}(a_{i_1})\left(\bl p_{{i_2},{j_1}}(a_{j_1})\bl p_{{i_2},{j_2}}(a_{i_2})+\bl p_{{i_2},{j_1}}(a_{i_2})\bl p_{{i_1},{j_1}}(a_{j_1})  \right).
\end{align*}
\end{proofof}

\section{Proof of Theorem \ref{thm:holy grail}}\label{sec:proof of thm}
\begin{proofof}{Theorem \ref{thm:holy grail}}
Fix an arbitrary instance. We prove the claim by a two-dimensional induction on the size of $\above(s),\below(s)$, for states $s\in \mS$. The base cases are 
\begin{itemize}
\item $\abs{\above(s)}\geq 2 $ and $\abs{\below(s)} = 1$ (Proposition \ref{prop:W case of one}), and
\item $\abs{\above(s)}=1$ and $\abs{\below(s)}\geq 2$ (Proposition \ref{prop:W case of one strong}),
\end{itemize}
which we relegate to Section \ref{sec:for theorem}. Next, assume the statement holds for all $s\in \mS$ such that $\abs{\above(s)}\leq K_1$, $\abs{\below(s)}\leq K_2$ and $\abs{\above(s)}+\abs{\below(s)}< K_1+K_2$. Let $U\in\mS$ denote a state with $\abs{\above(U)}=K_1$ and $\abs{\below(U)}=K_2$. For abbreviation, let $\ug\defeq\above(U)=\{a_{i_1},a_{i_2},\dots ,a_{i_{K_1}}\}$ and $\ul\defeq\below(U)=\{a_{j_1},a_{j_2},\dots ,a_{j_{K_2}}\}$, and assume the indices follow the stochastic order. Further, for every $a_i \in \ug, a_j\in \ul$ let 
\[
W^*_{i,j}(\ug,\ul)\defeq\bl p_{i,j}(a_j)W^*(\ug,\ul\setminus \{a_j\}) +\bl p_{i,j}(a_i)W^*(\ug\setminus\{a_i\},\ul).
\]
We need to prove that $W^*_{{i_1},{j_1}}(s)=W^*(s)$.
\paragraph{Remarks} Notice that if $X_{a_i'}>0$ for $a_{i'}\in \above(A)$, any $\mP$-valid policy gets $W^*(s)$. To see this, recall that $\mP$-valid policies reach terminate states only after exploring all arms in $\above(A)$. Consequently, we assume for the rest of the proof that $X_{a_i'}\leq 0$ for $a_{i'}\in \above(A)$. 
\paragraph{Step 1}
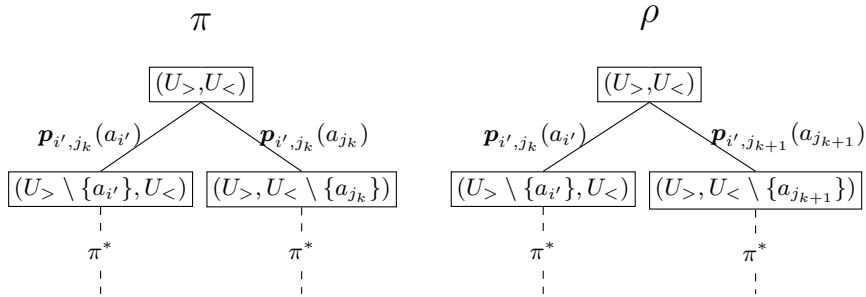
\begin{figure}[htbp]
\centering
\forestset{
 strongedge label/.style 2 args={
    edge label={node[midway,left, #1]{#2}},
  }, 
 weakedge label/.style 2 args={
    edge label={node[midway,right, #1]{#2}},
  }, 
   straightedge label/.style 2 args={
    edge label={node[midway, #1]{#2}},
  }, 
  important/.style={draw={red,thick,fill=red}}
}
\begin{forest} 
[{\Large $\pi$},
 [{$(U_>{,}U_<)$}, edge={white},l*=.05, for tree=
	{
	draw,
	font=\sffamily,
	l+=.5cm,
	inner sep=2pt,
	l sep=5pt,
	s sep=5pt,
	parent anchor=south,
	child anchor=north
    }
 	[${(U_>\setminus\{ a_{i'}\},U_<)}$, strongedge label={left}{$\bl p_{{i'},{j_k}}(a_{i'})$}
 		[,white, edge={dashed}, straightedge label={fill=white}{$\pi^*$}
 		]
 	]
	[${(U_>,U_<\setminus\{ a_{j_k}\})}$, weakedge label={right}{$\bl p_{{i'},{j_k}}(a_{j_k})$}
 		[,white, edge={dashed}, straightedge label={fill=white}{$\pi^*$}
 		]
 	]
 ]
]
\end{forest}
\qquad
\begin{forest} 
[{\Large $\rho$},
 [{$(U_>{,}U_<)$}, edge={white},l*=.05, for tree=
	{
	draw,
	font=\sffamily,
	l+=.5cm,
	inner sep=2pt,
	l sep=5pt,
	s sep=5pt,
	parent anchor=south,
	child anchor=north
    }
 	[${(U_>\setminus\{ a_{i'}\},U_<)}$, strongedge label={left}{$\bl p_{{i'},{j_k}}(a_{i'})$}
 		[,white, edge={dashed}, straightedge label={fill=white}{$\pi^*$}
 		]
 	]
	[${(U_>,U_<\setminus\{ a_{j_{k+1}}\})}$, weakedge label={right}{$\bl p_{{i'},{j_{k+1}}}(a_{j_{k+1}})$}
 		[,white, edge={dashed}, straightedge label={fill=white}{$\pi^*$}
 		]
 	]
 ]
]
\end{forest}
\label{fig:helping for thm two}
\caption{Illustration of the policies $\pi,\rho$ from Step 1 of Theorem \ref{thm:holy grail}. Notice that the left sub-trees of $\pi$ and $\rho$ are identical.
}
\end{figure}

Fix $a_{i'} \in \ug$. We show that for every $k$, $1\leq k<K_2$ it holds that $W^*_{{i'},{j_k}}(s) \geq W^*_{{i'},{j_{k+1}}}(s)$. We define the ordered, $\mP$-valid policy $\pi^*$ by $\sigl_\pi = (a_{i'},\dots)$  namely, $\sigl_{\pi^*}$ ranks $a_{i'}$ first and all the other arms in $\ug$ arbitrarily, and $\sigr_{\pi^*}=(a_{j_1},a_{j_2},\dots,a_{j_{K_2}})$. Due to the inductive assumption, for every state $s$ with $\abs{s}<\abs{U}$, $W^*(s) = W(\pi^*,s)$. Next, we define the policies $\pi,\rho$ explicitly, as follows:
\[
\pi(s) = 
\begin{cases}
\bl p_{{i'},{j_k}} & \textnormal{if $s=(\ug,\ul)$} \\
\pi^*(s)& \textnormal{otherwise} 
\end{cases},\quad 
\rho(s) = 
\begin{cases}
\bl p_{{i'},{j_{k+1}}} & \textnormal{if $s=(\ug,\ul)$} \\
\pi^*(s)& \textnormal{otherwise}
\end{cases}.
\]
We illustrate $\pi$ and $\rho$ in Figure \ref{fig:helping for thm two}. Note that both $\pi,\rho$ have on-path states that are off-path for $\pi^*$. For instance, $\rho$ reaches the state $(\ug,\ul\setminus\{a_{j_{k+1}}\})$ with positive probability, while $\pi^*$ cannot reach it at all. In addition, $\pi$ and $\rho$ are left-ordered with $\sigl_\pi=\sigl_\rho=\sigl_{\pi^*}$. Due to the inductive assumption, it is enough to show that $W(\pi,U)-W(\rho,U)\geq 0$, as this implies $W^*_{{i'},{j_k}}(s) \geq W^*_{{i'},{j_{k+1}}}(s)$.

Next, we factor $W(\pi,U)$ as follows: for every state $U'\subset U$, we factor $W(\pi,U')$ as long as $a_{j_k},a_{j_{k+1}}\in U'$. Once we reach a term $W(\pi,U')$ with $a_{j_k},a_{j_{k+1}}\notin U'$, we stop. Let $\Psi\defeq \prefix(a_{j_1},a_{j_2}\dots, a_{j_{k-1}})$ be the set of (possibly empty) prefixes of the first $k-1$ arms in $\ul$ according to $\pi^*$. Observe that\footnote{The reader can think of $\psi$ as the set of arms from $(a_{j_1},a_{j_2}\dots, a_{j_{k-1}})$ that were explored.}  
{\thinmuskip=.2mu
\medmuskip=0mu plus .2mu minus .2mu
\thickmuskip=1mu plus 1mu
\begin{align}\label{eq:w of pi}
W(\pi,U)&=\sum_{\psi\in \Psi}
\Pr(\pathto{s}{(\ul\setminus \psi)})R(\ul\setminus \psi)+\Pr(\pathto{s}{(\ul\setminus (\psi \cup \{a_{j_k}\})})R(\ul\setminus (\psi \cup \{a_{j_k}\})\nonumber \\
&\qquad +\underbrace{\sum_{\substack{Z\in \suff(\sigl_\pi)}}
f^\pi_Z \cdot W(\pi,Z,\ul\setminus \{a_{j_1},a_{j_2},\dots,a_{j_{k}},a_{j_{k+1}} \})}_{I^\pi}
\end{align}}
The coefficients $(f^\pi_Z)$ follow from the factorization process. Using similar factorization,
{\thinmuskip=.2mu
\medmuskip=0mu plus .2mu minus .2mu
\thickmuskip=1mu plus 1mu
\begin{align}\label{eq:w of rho}
W(\rho,U)&=\sum_{\psi\in \Psi}
\Pr(\pathtorho{s}{(\ul\setminus \psi)})R(\ul\setminus \psi)+\Pr(\pathtorho{s}{(\ul\setminus (\psi \cup \{a_{j_{k+1}}\}))})R(\ul\setminus (\psi \cup \{a_{j_{k+1}}\}))\nonumber\\
&\qquad +\underbrace{\sum_{\substack{Z\in \suff(\sigl_\rho)}}
f^\rho_Z \cdot W(\rho,Z,\ul\setminus \{a_{j_1},a_{j_2},\dots,a_{j_{k}},a_{j_{k+1}} \})}_{I^\rho}. 
\end{align}}%
Next, we express $(f^\pi_Z)_Z$ in terms of $Q$. By relying on the Equivalence lemma, we show that
\begin{claim}\label{claim:f are equal}
For every $Z\in \suff(\sigl_\pi)=\suff(\sigl_\rho)$, it holds that $f^\pi_Z=f^\rho_Z$.
\end{claim}
To see why Claim \ref{claim:f are equal} holds, notice that we can represent $f^\pi_Z$ as 
{\thinmuskip=.2mu
\medmuskip=0mu plus .2mu minus .2mu
\thickmuskip=1mu plus 1mu
\[
Q(\pi,\ug\setminus Z \cup \{a_{i(Z)}\}, \ul\setminus \{a_{j_1},a_{j_2},\dots,a_{j_{k}},a_{j_{k+1}} \})-Q(\pi,\ug\setminus Z, \ul\setminus \{a_{j_1},a_{j_2},\dots,a_{j_{k}},a_{j_{k+1}} \}),
\]}%
where $a_{i(Z)}$ is the minimal element in $Z$ according to $\sigl_\pi$. By invoking the Equivalence lemma, we can replace $\pi$ in the above expression with $\rho$ and thus for every $Z\in \suff(\sigl_\pi)=\suff(\sigl_\rho)$, it holds that $f^\pi_Z=f^\rho_Z$. Combining Claim \ref{claim:f are equal} with the inductive step and $\sigl_\pi,\sigl_\rho$ being equal, we conclude that $I^\pi=I^\rho$.

Next, we focus on the first sum of $W(\pi,U)$ in Equation \refeq{eq:w of pi}. For every $\psi\in \Psi$, we denote the event $E^\pi_\psi$ as a shorthand for
\[
E^\pi_\psi \defeq \left(\pathto{s}{(\ul\setminus \psi)}\right)\cup\left( \pathto{s}{(\ul\setminus (\psi \cup \{a_{j_k}\}))}\right).
\]
In words, $E^\pi_\psi$ is the event that the GMDP reaches the final state $(\ul\setminus \psi)$ or $(\ul\setminus (\psi \cup \{a_{j_k}\}))$ after starting in $s_0$ and following $\pi$. We use conditional expectation to simplify the summands in the first sum of Equation (\ref{eq:w of pi}),
{\thinmuskip=.2mu
\medmuskip=0mu plus .2mu minus .2mu
\thickmuskip=1mu plus 1mu
\begin{align}\label{eq:pi with alpha}
&\Pr(\pathto{s}{(\ul\setminus \psi)})R(\ul\setminus \psi)+\Pr(\pathto{s}{(\ul\setminus (\psi \cup \{a_{j_k}\})})R(\ul\setminus (\psi \cup \{a_{j_k}\})\nonumber\\
&=\Pr\left(E^\pi_\psi\right)\cdot \underbrace{\left( \Pr(\pathto{s}{(\ul\setminus \psi)}\mid E^\pi_\psi)R(\ul\setminus \psi)+\Pr(\pathto{s}{(\ul\setminus (\psi \cup \{a_{j_k}\})}\mid E^\pi_\psi)R(\ul\setminus (\psi \cup \{a_{j_k}\}))\right)}_{\alpha^\pi_\psi}\nonumber\\
&=\Pr(E^\pi_\psi) \alpha^\pi_\psi. 
\end{align}}%
Similarly for $\rho$, by letting
\[
E^\rho_\psi \defeq \left(\pathtorho{s}{(\ul\setminus \psi)}\right)\cup\left( \pathtorho{s}{(\ul\setminus (\psi \cup \{a_{j_{k+1}}\}))}\right)
\]
and following the derivation in Equation \refeq{eq:pi with alpha} for $\rho$, we get
{\thinmuskip=.2mu
\medmuskip=0mu plus .2mu minus .2mu
\thickmuskip=1mu plus 1mu
\begin{align}\label{eq:rho with alpha}
&\Pr(\pathtorho{s}{(\ul\setminus \psi)})R(\ul\setminus \psi)+\Pr(\pathtorho{s}{(\ul\setminus (\psi \cup \{a_{j_{k+1}}\})})R(\ul\setminus (\psi \cup \{a_{j_{k+1}}\})\nonumber\\
&=\Pr\left(E^\rho_\psi\right)\cdot \underbrace{\left( \Pr(\pathtorho{s}{(\ul\setminus \psi)}\mid E^\rho_\psi)R(\ul\setminus \psi)+\Pr(\pathtorho{s}{(\ul\setminus (\psi \cup \{a_{j_{k+1}}\})}\mid E^\rho_\psi)R(\ul\setminus (\psi \cup \{a_{j_{k+1}}\})) \right)}_{\alpha^\rho_\psi}\nonumber\\
&=\Pr(E^\rho_\psi) \alpha^\rho_\psi. 
\end{align}}%
Due to the fact that $I^\pi=I^\rho$ and relying on Equations (\ref{eq:w of pi})-(\ref{eq:rho with alpha}), we have  
\begin{align}\label{eq:w minus w}
W(\pi,U)-W(\rho,U)=\sum_{\psi\in \Psi}\left(\Pr(E^\pi_\psi) \alpha^\pi_\psi-\Pr(E^\rho_\psi) \alpha^\rho_\psi\right).
\end{align}
At this point, we might be tempted to show that every summand of the sum in Equation \refeq{eq:w minus w} is non-negative. Unfortunately, this approach would not work---stochastic order does not mean $\Pr(E^\pi_\psi) \alpha^\pi_\psi\geq \Pr(E^\rho_\psi) \alpha^\rho_\psi$. Instead, we take a different approach. For every $l, 1\leq l \leq k-1$ let $E^\pi_{l}$ denote the event that $a_{j_l}$ was observed in the final state. Formally,
\[
E^\pi_{l} \defeq  \bigcup_{\substack{\psi' \in \Psi\\a_{j_l} \in \psi' }} \left(\pathto{s}{(\ul\setminus \psi')}\right)\cup\left(\pathto{s}{(\ul\setminus (\psi' \cup \{a_{j_k}\})}\right).
\]
In addition, we let $E^\pi_{k}$ denote the empty event, and $E^\pi_{0}$ be the full event (that occurs w.p.~$1$).   For every non-empty $\psi \in \Psi$, i.e. $\abs{\psi}\geq 1$, let $\max(\psi)\defeq\argmax_{j_l:{a_{j_l}} \in \psi} \sigr_{\pi^*}(a_{j_l})$. According to our assumption about the index of the arms, $\max(\psi)$ is simply the maximal index of an arm in $\psi$ (that is well-defined when $\abs{\psi}\geq 1$). In addition, for completeness, if $\psi=\emptyset$ we let $\max(\emptyset)=0$. We can use the terms $(E^\pi_{l})_{l=0}^k$ to provide an alternative form for $E^\pi_\psi$:
\[
E^\pi_\psi=E^\pi_{\max(\psi)}\setminus E^\pi_{\max(\psi)+1}.
\]
Put differently, for $\psi$ with $\max(\psi)<k-1$, $E^\pi_\psi$ can be viewed as the collection of all of events in which arm $a_{j_{\max(\psi)}}$ was explored, but arm $a_{j_{\max(\psi)+1}}$ was not (where the ``arm'' $a_{j_0}$ for $\psi=\emptyset$ refers to $E^\pi_{0}$). Recall that for $\psi$ with $\max(\psi)=k-1$, $E^\pi_\psi=E^\pi_{\max(\psi)}$ since $E^\pi_{k}$ is the empty event.

Since $E^\pi_{\max(\psi)+1}\subseteq E^\pi_{\max(\psi)}$, we have that for every $\psi\in \Psi$,
\begin{align}\label{eq:e pi with Z}
\Pr(E^\pi_\psi)=\Pr(E^\pi_{\max(\psi)})-\Pr(E^\pi_{\max(\psi)+1}),
\end{align}
taking care of edge cases too. By renaming $\alpha^\pi_\psi$ to $\alpha^\pi_{\max(\psi)}$, i.e., $\alpha^\pi_{\max(\psi)} \defeq \alpha^\pi_\psi $, and rearranging Equation \refeq{eq:w of pi} using Equations \refeq{eq:pi with alpha} and \refeq{eq:e pi with Z},
{\thinmuskip=.2mu
\medmuskip=0mu plus .2mu minus .2mu
\thickmuskip=1mu plus 1mu
\begin{align}\label{eq:w pi with diff}
W(\pi,U)&= I^\pi+\sum_{\psi\in \Psi} \Pr\left(E^\pi_\psi\right) \alpha^\pi_\psi=I^\pi+\sum_{l=0}^{k-1}\left(\Pr(E^\pi_{l})-\Pr(E^\pi_{l+1})\right) \alpha^\pi_l.
\end{align}}%
By defining $E^\rho_l$ analogously,
\[
E^\rho_{l} \defeq  \bigcup_{\substack{\psi' \in \Psi\\a_{j_l} \in \psi' }} \left(\pathtorho{s}{(\ul\setminus \psi')}\right)\cup\left(\pathtorho{s}{(\ul\setminus (\psi' \cup \{a_{j_{k+1}}\})}\right),
\]
and following similar arguments, we conclude that
{\thinmuskip=.2mu
\medmuskip=0mu plus .2mu minus .2mu
\thickmuskip=1mu plus 1mu
\begin{align}\label{eq:w rho with diff}
W(\rho,U)&= I^\rho+\sum_{\psi\in \Psi} \Pr\left(E^\rho_\psi\right) \alpha^\rho_\psi=I^\rho+\sum_{l=0}^{k-1}\left(\Pr(E^\rho_{l})-\Pr(E^\rho_{l+1})\right) \alpha^\rho_l.
\end{align}}%
By rephrasing Equation (\ref{eq:w minus w}) using Equations (\ref{eq:w pi with diff}) and (\ref{eq:w rho with diff}),
\begin{align}\label{eq: thm step 1 good}
W(\pi,U)-W(\rho,U)=\sum_{l=0}^{k-1}\left(\Pr(E^\pi_{l})-\Pr(E^\pi_{l+1})\right) \alpha^\pi_l-\left(\Pr(E^\rho_{l})-\Pr(E^\rho_{l+1})\right) \alpha^\rho_l.
\end{align}
Next, we show two monotonicity properties.
\begin{proposition}\label{prop: monotonicity in thm}
Under Assumption \ref{assumption:dominance},
\begin{enumerate}
\item for every $l\in \{0,1,\dots,k-1 \}$, it holds that $\alpha^\pi_l \geq \alpha^\rho_l$. \label{item:prop alphas rho pi}
\item for every $l\in \{0,1,\dots,k-2\}$, it holds that $\alpha^\pi_{l+1} \geq \alpha^\pi_{l}$ and $\alpha^\rho_{l+1} \geq \alpha^\rho_{l}$. \label{item:prop alphas alpha pi}
\end{enumerate}
\end{proposition}
In fact, this is the only place in the proof of Theorem \ref{thm:holy grail} where we rely on Assumption \ref{assumption:dominance}. Equipped with Proposition \ref{prop: monotonicity in thm}, we can make the final argument. For every $r,r\in \{1\dots,k-1\}$ let
\[
f(r) \defeq \left( \Pr(E^\pi_{r})-\Pr(E^\rho_{r}) \right)\alpha^\pi_{r-1}.
\]
In addition, let
\[
g(r)\defeq \sum_{l=0}^{r}\left(\Pr(E^\pi_{l})-\Pr(E^\pi_{l+1})\right) \alpha^\pi_l-\left(\Pr(E^\rho_{l})-\Pr(E^\rho_{l+1})\right) \alpha^\rho_l
\]
We shall show that for every $r,r\in \{0,\dots,k-2\}$ it holds that
\begin{align}\label{eq:thm last argument}
W(\pi,U)-W(\rho,U)\geq f(r+1)+g(r).
\end{align}
For $r=k-2$, we have
{
\begin{align*}
\textnormal{Eq. (\ref{eq: thm step 1 good})}&=\left(\Pr(E^\pi_{k-1})-\Pr(E^\pi_{k})\right) \alpha^\pi_{k-1}-\left(\Pr(E^\rho_{k-1})-\Pr(E^\rho_{k})\right) \alpha^\rho_{k-1}+g(k-2)\\
&\stackrel{\substack{E^\pi_{k},E^\rho_{k}\\\textnormal{are empty}}}{=}\Pr(E^\pi_{k-1})\alpha^\pi_{k-1}-\Pr(E^\rho_{k-1}) \alpha^\rho_{k-1}+g(k-2)\\
&\stackrel{\textnormal{Prop. \ref{prop: monotonicity in thm}.\ref{item:prop alphas rho pi}}}{\geq}\left(\Pr(E^\pi_{k-1})-\Pr(E^\rho_{k-1}) \right)\alpha^\pi_{k-1}+g(k-2)\\
&\stackrel{\textnormal{Prop. \ref{prop: monotonicity in thm}.\ref{item:prop alphas alpha pi}}}{\geq}\left(\Pr(E^\pi_{k-1})-\Pr(E^\rho_{k-1}) \right)\alpha^\pi_{k-2}+g(k-2)\\
&=f(k-1)+g(k-2).
\end{align*}}%
Assume Inequality \refeq{eq:thm last argument} holds for $r+1$. Then, for $r$ we have 
{\thinmuskip=.2mu
\medmuskip=0mu plus .2mu minus .2mu
\thickmuskip=1mu plus 1mu
\begin{align*}
f(r+1)+g(r) &= \left( \Pr(E^\pi_{r+1})-\Pr(E^\rho_{r+1}) \right)\alpha^\pi_{r}+g(r-1)\\
&\qquad \qquad +\left(\Pr(E^\pi_{r})-\Pr(E^\pi_{r+1})\right) \alpha^\pi_r-\underbrace{\left(\Pr(E^\rho_{r})-\Pr(E^\rho_{r+1})\right)}_{\geq 0,\textnormal{ Eq. \refeq{eq:e pi with Z}}}  \alpha^\rho_r\\
&\stackrel{\textnormal{Prop. \ref{prop: monotonicity in thm}.\ref{item:prop alphas rho pi}}}{\geq} \left( \Pr(E^\pi_{r+1})-\Pr(E^\rho_{r+1}) \right)\alpha^\pi_{r}+g(r-1)\\
&\qquad \qquad +\left(\Pr(E^\pi_{r})-\Pr(E^\pi_{r+1})\right) \alpha^\pi_r-\left(\Pr(E^\rho_{r})-\Pr(E^\rho_{r+1})\right) \alpha^\pi_r\\
&=\Pr(E^\pi_{r})\alpha^\pi_r-\Pr(E^\rho_{r}) \alpha^\pi_r +g(r-1)\\
&\stackrel{\textnormal{Prop. \ref{prop: monotonicity in thm}.\ref{item:prop alphas alpha pi}}}{\geq}\Pr(E^\pi_{r})\alpha^\pi_{r-1}-\Pr(E^\rho_{r}) \alpha^\pi_{r-1} +g(r-1)\\
&= f(r)+g(r-1).
\end{align*}}%
Ultimately, by setting $r=0$ in Inequality \refeq{eq:thm last argument},
\begin{align*}
W(\pi,U)-W(\rho,U)&\geq f(1)+g(0)\\
&=\left( \Pr(E^\pi_{1})-\Pr(E^\rho_{1}) \right)\alpha^\pi_{0}\\
&\qquad \qquad +\left(\Pr(E^\pi_{0})-\Pr(E^\pi_{1})\right) \alpha^\pi_0-\left(\Pr(E^\rho_{0})-\Pr(E^\rho_{1})\right) \alpha^\rho_0\\
&\geq\Pr(E^\pi_{0}) \alpha^\pi_0-\Pr(E^\rho_{0}) \alpha^\pi_0\\
&=0.
\end{align*}
This concludes the first step of the theorem.
\paragraph{Step 2}  In this step, we show that for every $k$, $1\leq k < K_1$ it holds that $W^*_{{i_k},{j_1}}(s) = W^*_{{i_{k+1}},{j_1}}(s)$. Define an ordered, $\mP$-valid policy $\pi^*$ by $\sigl_\pi = (a_{i_1},a_{i_2},\dots,a_{i_{K_1}})$,  namely, $\sigl_{\pi^*}$ ranks the elements of $\ug$ according to the stochastic order, and $\sigr_{\pi^*}=(a_{j_1},a_{j_2},\dots,a_{j_{K_2}})$. Due to the inductive assumption, for every state $s$ with $\abs{s}<\abs{U}$, $W^*(s) = W(\pi^*,s)$. Next, we define the policies $\pi,\rho$ explicitly, as follows:
\[
\pi(s) = 
\begin{cases}
\bl p_{{i_k},{j_1}} & \textnormal{if $s=(\ug,\ul)$} \\
\pi^*(s)& \textnormal{otherwise} 
\end{cases},\quad 
\rho(s) = 
\begin{cases}
\bl p_{{i_{k+1}},{j_{1}}} & \textnormal{if $s=(\ug,\ul)$} \\
\pi^*(s)& \textnormal{otherwise}
\end{cases}.
\]
As in the previous step, the inductive step suggests that showing $W(\pi,s)=W(\rho,s)$ is suffice. However, unlike the previous step, here the set of reachable terminal state is the same for $\pi$ and $\rho$; hence, this equality is almost immediate due to the Equivalence lemma. Let $\Psi'\defeq \prefix(a_{j_1},a_{j_2}\dots, a_{j_{K_2-1}})$ be the set of (possibly empty) prefixes of the arms in $\ul \setminus \{a_{j_{K_2}}\}$ according to $\pi^*$. Observe that we can factor $W(\pi,s)$ as follows:
\begin{align}\label{eq:w for pi with prob}
W(\pi,s) &= Q(\pi,U)\cdot R(\emptyset)+ \sum_{\psi \in \Psi}\Pr(\pathto{s}{(\ul \setminus \psi)})R(\ul \setminus \psi).
\end{align}
The next Claim \ref{claim:thm:step 2 claim} suggests we can replace probabilities with functions of $Q$.
\begin{claim}\label{claim:thm:step 2 claim}
For every $\psi \in \Psi$, it holds that 
\[
\Pr(\pathto{s}{(\ul \setminus \psi)})=Q(\pi,\ug, \psi)-Q(\pi,\ug,\psi \cup\{a_{j_{max(\psi)+1}} \}).
\]
\end{claim}
By applying the Equivalence lemma on the statement of  Claim \ref{claim:thm:step 2 claim}, we obtain that for every $\psi \in \Psi$
\[
\Pr(\pathto{s}{(\ul \setminus \psi)})=\Pr(\pathtorho{s}{(\ul \setminus \psi)});
\]
hence, we can rewrite Equation \refeq{eq:w for pi with prob} as 
\begin{align*}
W(\pi,s) &= Q(\rho,U)\cdot R(\emptyset)+ \sum_{\psi \in \Psi}\Pr(\pathtorho{s}{(\ul \setminus \psi)})R(\ul \setminus \psi)\\
&=W(\rho,s).
\end{align*}
This concludes the second step of the theorem.
\paragraph{Step 3 (final)} We are ready to prove the theorem. Fix arbitrary $a_{\tilde i}$ and $a_{\tilde j}$ such that $a_{\tilde i} \in \ug$ and  $a_{\tilde j} \in \ul$. By the previous steps, we know that
\[
W^*_{{i_1},{j_1}}(U)\stackrel{\textnormal{Step 2}}{=}W^*_{{\tilde i},{j_1}}(U)\stackrel{\textnormal{Step 1}}{\geq}W^*_{{\tilde i},{\tilde j}}(U).
\]
This ends the proof of Theorem \ref{thm:holy grail}.
\end{proofof}

\section{Statements for Theorem \ref{thm:holy grail}}\label{sec:for theorem}
\begin{proposition}\label{prop:W case of one}
Let $U$ be a state such that $\abs{\above(U)}\geq 2$ and $\abs{\below(U)} =1$. 
It holds that $W(\pi^\star,U)=W^*(U)$.
\end{proposition}
\begin{proofof}{Proposition \ref{prop:W case of one}}
Denote $\below(U)=\{a_{j}\}$. We can assume w.l.o.g. that the realization of all arms in $\above(U)$ are non-positive, as otherwise every $\mP$-valid policy will explore all the arms; thus, $W^*(U)=Q^*(U)\cdot \max\{0, X_{a_{j}}  \}$. Finally, the Equivalence lemma suggests that   $Q^*(U)$ is policy invariant; hence, $W(\pi,U)=W^*(U)$ holds for any $\mP$-valid policy $\pi$, and in particular for $\pi=\pi^\star$.
\end{proofof}

\begin{proposition}\label{prop:W case of one strong}
Let $U$ be a state such that $\abs{\above(U)}=1$ and $\abs{\below(U)} \geq 2$. It holds that $W(\pi^\star,U)=W^*(U)$.
\end{proposition}
\begin{proofof}{Proposition \ref{prop:W case of one strong}}
This statement is a special case of Proposition \ref{prop:index with ugeq one} for instances satisfying Assumption \ref{assumption:dominance}.
\end{proofof}
\begin{proposition}\label{prop:index with ugeq one}
Let $U$ be a state such that $\abs{\above(U)}=1$ and $\abs{\below(U)} \geq 2$. Let $f^*$ be a real-valued function, $f^*:\below(U)\rightarrow \mathbb R$, such that for every $a_l \in \below(U)$,
\[
f^*(a_{l})=\frac{\Pr(X_{a_l} > 0 )\E(\max_{a_j \in {U} } X_{a_j}\mid X_{a_l}>0 )}{\abs{\mu(a_l)}}.
\]  
Denote by $\pi_{f^*}$ the right-ordered policy that orders $\below(U)$ according to decreasing order of $f^*$. Then, $W(\pi_{f^*},U)=W^*(U)$. 
\end{proposition}

\begin{proofof}{Proposition \ref{prop:index with ugeq one}}
Denote $\above(U)=\ug=\{a_{i}\}$  and $\below(U)=\ul=\{a_{j_1},\dots a_{j_k}\}$ for $k=\abs{\ul}$. Let $\pi$ be any right-ordered policy with the matching $\sigr_\pi$, such that $\pi\neq \pi_{f^*}$. Assume that there are indices $r,l$, for $1\leq r,l \leq k$, such that $\sigr_\pi(a_{l})<\sigr_\pi(a_{r})$ yet $f(a_{l})<f(a_{{r}})$, for arms $a_l,a_r\in \ul$. Moreover, if such a pair $(l,r)$ exists,  assume w.l.o.g. that $\sigr_\pi$ orders them consequentially, i.e., for every arm $a\in U$ such that $a\notin \{ a_{l}, a_{r} \}$,  either $\sigr_\pi(a)< \sigr_\pi(a_{l})$ or $\sigr_\pi(a)> \sigr_\pi(a_{r})$.

Denote by $\pi'$ the right-ordered policy that swaps $a_{l}$ and $a_{r}$. If we show that $\pi'$ yields a better reward than $\pi$, we could swap the order of $\pi$ one pair at a time, thereby showing that $\pi_{f^*}$ is indeed optimal. To simplify notation, we denote by $\ind_{a}$ the event that $X_a>0$ for arm $a\in U$.

Let $\ul'$ be the set of all arms in $\ul$ such that $\ul'=\{a\in \ul\mid \sigr_\pi(a)< \sigr_\pi(a_{l})  \}$. We divide the analysis into two cases: in case $X_{a_{i}} >0$ or $\max_{a\in \ul'} X_a > 0$, both $\pi,\pi'$ obtain the same reward. Otherwise, assume that $X_{a_{i}} \leq 0$ and $\max_{a\in \ul'} X_a \leq 0$ ; hence, in case arm $a_{i}$ is selected, the policy reaches a terminal state with a reward of 0. The reward of $\pi$ is given by 
{\small
\thinmuskip=.2mu
\medmuskip=0mu plus .2mu minus .2mu
\thickmuskip=1mu plus 1mu
\begin{align*}
W(\pi,U)=C_1 \left( \bl p_{i,l}(l)\E(\ind_l \max_{a_j \in {U_< \setminus \ul'}} X_{a_j} )+\bl p_{i,l}(l)\bl p_{i,r}(r)\E((1-\ind_l)\ind_r \max_{a_j \in {U_< \setminus \ul'}} X_{a_j})+\bl p_{i,l}(l)\bl p_{i,r}(r)C_2 \right),
\end{align*}}%
where $C_1$ and $C_2$ are constants that depend on $\sigr_\pi$. Similarly, the reward of $\pi'$ is given by 
{\small
\thinmuskip=.2mu
\medmuskip=0mu plus .2mu minus .2mu
\thickmuskip=1mu plus 1mu
\begin{align*}
W(\pi',U) = C_1\left( \bl p_{i,r}(r)\E(\ind_r \max_{a_j \in {U_< \setminus \ul'}} X_{a_j} )+\bl p_{i,r}(r)\bl p_{i,l}(l)\E((1-\ind_r)\ind_l \max_{a_j \in {U_< \setminus \ul'}} X_{a_j})+\bl p_{i,r}(r)\bl p_{i,l}(l)C_2 \right),
\end{align*}}%
where $C_1$ and $C_2$ are the same constants.
If $W(\pi,U)\geq W(\pi',U)$, then
\begin{align*}
& \bl p_{i,l}(l)\E(\ind_l \max_{a_j \in {U_< \setminus \ul'}} X_{a_j} )+\bl p_{i,l}(l)\bl p_{i,r}(r)\E((1-\ind_l)\ind_r \max_{a_j \in {U_< \setminus \ul'}} X_{a_j}) \\
&\qquad \geq \bl p_{i,r}(r)\E(\ind_r \max_{a_j \in {U_< \setminus \ul'}} X_{a_j} )+\bl p_{i,r}(r)\bl p_{i,l}(l)\E((1-\ind_r)\ind_l \max_{a_j \in {U_< \setminus \ul'}} X_{a_j}),
\end{align*}
implying that
\begin{align*}
&\bl p_{i,l}(l)\E(\ind_l \max_{a_j \in  {U_< \setminus \ul'}} X_{a_j} )(1-\bl p_{i,r}(r)) \geq \bl p_{i,r}(r)\E(\ind_r \max_{a_j \in  {U_< \setminus \ul'}} X_{a_j} )(1-\bl p_{i,l}(l)).
\end{align*}
Stated otherwise,
\begin{align*}
\frac{\bl p_{i,l}(l)\E(\ind_l \max_{a_j \in {U_< \setminus \ul'}} X_{a_j} )}{ 1-\bl p_{i,l}(l)} > \frac{\bl p_{i,r}(r)\E(\ind_r \max_{a_j \in {U_< \setminus \ul'}} X_{a_j} )}{1-\bl p_{i,r}(r) }.
\end{align*}
Finally, due to the definitions of $\bl p_{i,r},\ind_l$ and $\bl p_{i,l}, \ind_r$,
\begin{align*}
\frac{\Pr(X_{a_l} > 0 )\E(\max_{a_j \in {U_< \setminus \ul'} } X_{a_j}\mid X_{a_l}>0 )}{\abs{\mu(a_l)}} \geq \frac{\Pr(X_{a_l} > 0 )\E( \max_{a_j \in {U_< \setminus \ul'}} X_{a_j} \mid X_{a_r} >0 )}{\abs{\mu(a_r)}},
\end{align*}
which contradicts our assumption that $f(a_{l})<f(a_{r})$.
\end{proofof}

\begin{proofof}{Claim \ref{claim:f are equal}}
Notice that we can represent $f^\pi_Z$ as 
{\thinmuskip=.2mu
\medmuskip=0mu plus .2mu minus .2mu
\thickmuskip=1mu plus 1mu
\[
Q(\pi,\ug\setminus Z \cup \{a_{i(Z)}\}, \ul\setminus \{a_{j_1},a_{j_2},\dots,a_{j_{k}},a_{j_{k+1}} \})-Q(\pi,\ug\setminus Z, \ul\setminus \{a_{j_1},a_{j_2},\dots,a_{j_{k}},a_{j_{k+1}} \}),
\]}%
where $a_{i(Z)}$ is the minimal element in $Z$ according to $\sigl_\pi$. This process is similar in spirit to the proof of Proposition \ref{prop:coef c} and is hence omitted. Then, we can mirror the same arguments for $f^\rho_Z$. Finally, the Equivalence lemma suggests that the two representations are equal.
\end{proofof}

\begin{proofof}{Proposition \ref{prop: monotonicity in thm}}
We address the two parts separately below.
\paragraph{Part \ref{item:prop alphas rho pi}} 
Notice that the terminal state $(\ul\setminus \psi)$ is reachable from the left sub-tree of $\pi$ and $\rho$ solely (see Figure \ref{fig:helping for thm two} for illustration). Due to the construction of $\pi$ and $\rho$,
\begin{align}\label{eq:pi to rho}
\Pr(\pathto{s}{(\ul\setminus \psi)})&= \bl p_{{i_1},{j_k}}(a_{i_1})\Pr(\pathto{s\setminus \{a_{i_1}\}}{(\ul\setminus \psi)})\nonumber\\
&=\frac{-\mu(a_{j_{k}})}{- \mu(a_{j_{k}})+ \mu(a_{i_{1}})} \Pr(\pathto{s\setminus \{a_{i_1}\}}{(\ul\setminus \psi)})\nonumber\\
&=\frac{-\mu(a_{j_{k}})}{- \mu(a_{j_{k}})+ \mu(a_{i_{1}})} \Pr(\pathtorho{s\setminus \{a_{i_1}\}}{(\ul\setminus \psi)})\nonumber\\
&\leq\frac{-\mu(a_{j_{k+1}})}{- \mu(a_{j_{k+1}})+ \mu(a_{i_{1}})} \Pr(\pathtorho{s\setminus \{a_{i_1}\}}{(\ul\setminus \psi)})\nonumber\\
&\leq \bl p_{{i_1},{j_{k+1}}}(a_{i_1})\Pr(\pathtorho{s\setminus \{a_{i_1}\}}{(\ul\setminus \psi)})\nonumber \\
& = \Pr(\pathtorho{s}{(\ul\setminus \psi)}),
\end{align}
since $\mu(a_{j_{k+1}})\geq \mu(a_{j_{k}})$, and due to monotonicity of $f(x)=\frac{x}{x+c}$ for positive $c$. Using similar arguments, 
\begin{align}\label{eq: pi with rho again}
\Pr(\pathto{s}{(\ul\setminus (\psi \cup \{a_{j_k}\})}) \geq \Pr(\pathtorho{s}{(\ul\setminus (\psi \cup \{a_{j_{k+1}}\})}).
\end{align}
Now, observe that
\begin{align}\label{eq: with Ez}
\Pr(\pathto{s}{(\ul\setminus \psi)}\mid E^\pi_\psi)
&=\frac{\Pr(\pathto{s}{(\ul\setminus \psi)})}{\Pr(\pathto{s}{(\ul\setminus \psi)}) +\Pr(\pathto{s}{(\ul\setminus (\psi \cup \{a_{j_k}\})})} \nonumber\\
&\stackrel{\textnormal{Eq. }(\ref{eq:pi to rho}),(\ref{eq: pi with rho again})}{\leq} \frac{\Pr(\pathtorho{s}{(\ul\setminus \psi)})}{\Pr(\pathtorho{s}{(\ul\setminus \psi)}) +\Pr(\pathtorho{s}{(\ul\setminus (\psi \cup \{a_{j_{k+1}}\})})} \nonumber\\
&=\Pr(\pathtorho{s}{(\ul\setminus \psi)}\mid E^\rho_\psi),
\end{align}
where the second to last step follows again from monotonicity of $f(x)=\frac{x}{x+c}$ for positive $c$. Further, due to monotonicity of the reward function $R$,
\[
R(\ul\setminus (\psi\cup \{a_{j_k}\})) \geq R(\ul\setminus \psi), \qquad R(\ul\setminus (\psi\cup \{a_{j_{k+1}}\})) \geq R(\ul\setminus \psi).
\]
In addition, due to Assumption \ref{assumption:dominance}, $R(\ul\setminus (\psi\cup \{a_{j_k}\})) \geq  R(\ul\setminus (\psi\cup \{a_{j_{k+1}}\}))$. Wrapping up,
{\thinmuskip=.2mu
\medmuskip=0mu plus .2mu minus .2mu
\thickmuskip=1mu plus 1mu
\begin{align*}
\alpha^\pi_\psi&=\Pr(\pathto{s}{(\ul\setminus \psi)}\mid E^\pi_\psi)R(\ul\setminus \psi)+(1-\Pr(\pathto{s}{(\ul\setminus \psi)}\mid E^\pi_\psi))R(\ul\setminus (\psi \cup \{a_{j_k}\}))\\
&\geq \Pr(\pathto{s}{(\ul\setminus \psi)}\mid E^\pi_\psi)R(\ul\setminus \psi)+(1-\Pr(\pathto{s}{(\ul\setminus \psi)}\mid E^\pi_\psi))R(\ul\setminus (\psi \cup \{a_{j_{k+1}}\}))\\
&\stackrel{\textnormal{Eq. }(\ref{eq: with Ez})}{\geq} \Pr(\pathtorho{s}{(\ul\setminus \psi)}\mid E^\rho_\psi)R(\ul\setminus \psi)+(1-\Pr(\pathtorho{s}{(\ul\setminus \psi)}\mid E^\rho_\psi))R(\ul\setminus (\psi \cup \{a_{j_{k+1}}\}))\\
&=\alpha^\rho_\psi.
\end{align*}}%
This completes the proof of the first part.

\paragraph{Part \ref{item:prop alphas alpha pi}} We prove the claim for $\alpha^\rho_{l+1} \geq \alpha^\rho_{l}$, and the other part is symmetrical. Let $\psi=\psi(l)$ such that $\max(\psi)=l$. Notice that the reward function $R$ is a set function, and is, by definition monotonically decreasing; hence, $R(\ul\setminus \psi) \leq R(\ul\setminus (\psi \cup \{a_{j_{k+1}}\}))$. Consequently,
{\thinmuskip=.2mu
\medmuskip=0mu plus .2mu minus .2mu
\thickmuskip=1mu plus 1mu
\begin{align*}
\alpha^\rho_\psi&=\Pr(\pathtorho{s}{(\ul\setminus \psi)}\mid E^\rho_\psi)R(\ul\setminus \psi)+\Pr(\pathtorho{s}{(\ul\setminus (\psi \cup \{a_{j_{k+1}}\})}\mid E^\rho_\psi)R(\ul\setminus (\psi \cup \{a_{j_{k+1}}\}))\nonumber\\
&\leq R(\ul\setminus (\psi \cup \{a_{j_{k+1}}\})).
\end{align*}}%
Further, let $\psi'$ such that $\max(\psi')=l+1$, namely $\psi'=\psi\cup\{a_{j_{l+1}} \}$. Using the same properties of $R$, we have that $R(\ul\setminus \psi') \leq R(\ul\setminus (\psi' \cup \{a_{j_{k+1}}\}))$; thus,
{\thinmuskip=.2mu
\medmuskip=0mu plus .2mu minus .2mu
\thickmuskip=1mu plus 1mu
\begin{align*}
\alpha^\rho_{\psi'}&=\Pr(\pathtorho{s}{(\ul\setminus \psi')}\mid E^\rho_{\psi'})R(\ul\setminus \psi')+\Pr(\pathtorho{s}{(\ul\setminus (\psi' \cup \{a_{j_{k+1}}\})}\mid E^\rho_{\psi'})R(\ul\setminus (\psi' \cup \{a_{j_{k+1}}\}))\nonumber\\
&\geq R(\ul\setminus \psi' ).
\end{align*}}%
Next, let $V_\psi$ denote the event that $(X_{a})_{a\in \psi}$ attain value below $\alpha$. Observe that 
{\thinmuskip=.2mu
\medmuskip=0mu plus .2mu minus .2mu
\thickmuskip=1mu plus 1mu
\begin{align*}
R(\ul\setminus (\psi \cup \{a_{j_{k+1}}\})) &=\Pr(V_\psi)R(\ul\setminus (\psi \cup \{a_{j_{k+1}}\}))+(1-\Pr(V_\psi))R(\ul\setminus (\psi \cup \{a_{j_{k+1}}\}))\nonumber\\
&=\Pr(V_\psi)\max\{\alpha,X_{a_{j_{k+1}}}\}+(1-\Pr(V_\psi))\max\{a_{j_1},\dots ,a_{j_l},a_{j_{k+1}}  \}\nonumber\\
&\leq \Pr(V_\psi)\max\{\alpha,X_{a_{j_{l+1}}}\}+(1-\Pr(V_\psi))\max\{a_{j_1},\dots ,a_{j_l},a_{j_{l+1}}  \}\\
&=R(\ul\setminus \psi' ) ,
\end{align*}}%
where the second to last inequality is due to Assumption \ref{assumption:dominance} and independence of $(X_{a_i})_{i=1}^K$. Ultimately, 
\[
\alpha^\rho_{l+1}=\alpha^\rho_{\psi'} \geq R(\ul\setminus \psi' ) \geq R(\ul\setminus (\psi \cup \{a_{j_{k+1}}\})) \geq \alpha^\rho_{\psi}=\alpha^\rho_{l}.
\]
This completes the proof of the second part.
\end{proofof}

\begin{proofof}{Claim \ref{claim:thm:step 2 claim}}
The proof goes along the lines of Claim \ref{claim:f are equal}, and is hence omitted.
\end{proofof}

\section{Proofs of Observations and Propositions from Sections \ref{sec:infinite} and \ref{sec:policy to algorithm}}\label{sec:appendix main body}
\begin{proofof}{Observation \ref{obs:eventually will explore}}
Let $x(a_i)>0$ for some $i\in [K]$, let $j$ be an index of unexplored arm, and let $\mI$ be the information of the algorithm. We overload the notation $\bl p_{i,j}$ to acknowledge the realized value $x(a_i)$; that is,
\begin{align*}
\bl p_{i,j}(a) =
\begin{cases}
\frac{-\mu(a_j)}{x(a_i)-\mu(a_j)} & \textnormal{if } a=a_i\\
\frac{x(a_i)}{x(a_i)-\mu(a_j)} & \textnormal{if } a=a_j\\
0 & \textnormal{otherwise}
\end{cases}.
\end{align*}
Notice that 
\begin{align*}
\sum_{a\in A}\bl p(a)\E\left[X(a)\mid \mI\right] &= \bl  p_{i,j}(a_i)x(a_i)  + \bl  p_{i,j}(a_j)\mu(a_j) \\
&= x(a_i)\cdot \frac{-\mu(a_j)}{x(a_i)-\mu(a_j)} + \mu(a_j)\cdot \frac{x(a_i)}{x(a_i)-\mu(a_j)} = 0;  
\end{align*}
hence, $\bl p_{i,j}$ is MIR w.r.t. to $\mI$. After selecting $\bl p_{i,j}$, either $a_i$ was realized or $a_j$. In the former, the information remains the same, and we can repeat this experiment again. The probability of $a_j$ realizing is positive and constant, and hence, after a finite time, we will eventually realize it. Once we do, the number of unexplored armed decreases by one. We can follow this process until all arms are explored.
\end{proofof}

\begin{proofof}{Observation~\ref{obs: U leq W*}}
The proof of this observation relies on constructing a policy $\pi$ that simulates $\ALG$. Since by definition $W(\pi,A) \leq  W^\star(A)$, it is enough to show that $\lim_{T \rightarrow \infty }\mU_T(\ALG) \leq W(\pi,A)$. In every round, $\pi$ selects precisely what $\ALG$ selects, and if the realized arm was already explored by $\ALG$, $\pi$ ignores it. The infinite time expected value of $\ALG$ cannot exceed $ W(\pi,A)$. The full details are similar to \cite[Theorem 3]{Fiduciary} and are hence omitted. 
\end{proofof}

\begin{proofof}{Observation \ref{obs: U get W}}
Fix any policy $\pi$. Let $\ALG(\pi)$ be the modification of Algorithm~\ref{alg:alg of pi} that uses $\pi$ instead of $\OGP$ in Lines \ref{algpi:while}-\ref{algpi:play with ogp}. Once $\pi$ reaches a terminal state, $\ALG(\pi)$ secures the reward of $\pi$ in finite time. Overall, $\lim_{T \rightarrow \infty }\mU_T(\ALG(\pi)) = W(\pi,A)$.
\end{proofof}

\begin{proofof}{Proposition \ref{prop:bernoulli opt}}
Fix an $\ise$ instance such that $(X(a_i))_i \in \{x^-,x^+\}$ (for $x^- \leq x^+$) almost surely. For the problem to be non-trivial, we must have $x^- <0$ and $x^+ >0$. Otherwise, if  $x^-,x^+ <0$ the only MIR action is $a_0$, and if $x^-,x^+ \geq 0$, we can explore all arms using the singleton portfolios $(\bl p_{ii})_{i \in [K]}$. From here on, we assume w.l.o.g. that $x^- = -1$ and $x^+=H$. For convenience, we state $\SEGB'$ explicitly in Algorithm~\ref{alg:alg of pi two supported}. Before we prove the proposition, we remark that
\begin{enumerate}
    \item Since $(X(a_i))_i$ take either $-1$ or $H$, Assumption~\ref{assumption:dominance} implies a stochastic order on all arms, not only on $\below(A)$. 
    \item Any asymptotically optimal algorithm conducts at most $K$ exploration rounds before it exploits. This implies an immediate crude bound of
    $\mU_T(\SEGB') \geq \left(  1-\frac{KH}{T}\right) \OPT_T$.
    \item This proof uses the analysis presented in Section~\ref{sec:thm1 outline}.
\end{enumerate}

The proof is composed of two steps. In the first step, we show that if $T > T_0$ for some $T_0$, any optimal algorithm must explore the arms according to a policy that admits the same structure of $\OGP$. In the second step, we show that all such policies have an identical exploration time, and hence all achieve the same social welfare.\\
\textbf{Step 1:} In the case of realizing a positive reward, any algorithm would stop exploring and exploit that realized reward. Consequently, we can separate exploration rounds from exploitation rounds. Notice, however, that the exploration policy can select portfolios different that $\OGP$ for finite $T$. To illustrate, reconsider Example~\ref{example with four}. In the extreme case of $T=1$, there is no point in selecting $\bl p_{1,3}$, since exploring $a_3$ is futile; we only care about maximizing the current round's reward.

However, if $T$ is large \textit{enough}, any optimal algorithm must employ an asymptotically optimal policy. To see this, let $(\pi,\ALG^{\pi})$ be a pair of exploration policy and the algorithm that employs it, and assume $\pi$ does not admit the structure of $\OGP$. The social welfare of $\ALG^\pi$ satisfies
\begin{align}\label{eq:alg of pi is not optimal}
\mU_T(\ALG^\pi) \leq KH+ (T-K)W(\pi,A).    
\end{align}
Similarly, taking into account the optimality of $\OGP$,
\begin{align}\label{eq:alg of pi with ogp}
\mU_T(\ALG^{\OGP}) \geq -K +(T-K)W^\star(A).
\end{align}
Using similar arguments to those in Proposition~\ref{prop:optimal p valid}, we can assume w.l.o.g. that $\pi$ belongs to $\{2^A \rightarrow \mP \cup \mP'\}$ (recall the definition of $\mP$ and $\mP'$ from Subsection~\ref{subsec:bin}). To see this, observe that any MIR policy can be formulated as a convex combination of policies that use $\mP \cup \mP'$ solely, and therefore we can assume that $\pi$ is the one for which $\ALG^{\pi}$ gets the highest social welfare. Furthermore, due to the proof of Theorem~\ref{thm:holy grail} (precisely Equation~\eqref{eq:rho with alpha}) it follows that if $\pi$ does not admit the structure of $\OGP$, then it is strictly sub-optimal. Next, let
\begin{align}\label{eq:alg of pi little omega}
\omega  \defeq \min_{\substack{\rho \in \{2^A \rightarrow \mP \cup \mP'\},\\W(\rho, A) < W^\star(A)}} W^\star(A) - W(\rho,A) > 0.
\end{align}
The quantity $\omega$ concerns the distributions of $(X(a_i))_i$ and is completely independent of the time $T$. We can further quantify or bound $\omega$ but this abstract and simple form is sufficient for our purposes. Let $T_0 \defeq K + \frac{K(H+1)}{\omega}$. Combining Inequalities~ \eqref{eq:alg of pi is not optimal},\eqref{eq:alg of pi with ogp}, and \eqref{eq:alg of pi little omega} we get
\begin{align*}
\mU_T(\ALG^{\OGP})-\mU_T(\ALG^\pi) &\geq  -K +(T-K)W^\star(A) - KH - (T-K)W(\pi,A) \\
& \geq -K(H+1) + (T-K)\omega\\
& > 0,
\end{align*}
provided that $T>T_0$. To conclude this step, we know that $\pi$ is a variation of $\OGP$.

\textbf{Step 2:}
Notice, however, that $\OGP$ is a class of policies differing from one another in the choices of arms from $\above(A)$ (Line~\ref{policy:pick arbitrary} in Policy~\ref{policy:pi star}); hence, one policy may attain a better reward than the other by reaching exploitation faster. 

As we commented in Line~\ref{algpi:two:play with ogp} of Algorithm~\ref{alg:alg of pi two supported}, when the state $s$ does not contain arms from $\below(A)$ we prioritize singleton portfolios according to the stochastic order. That is, we favor $\bl p_{i,i}$ over $\bl p_{i',i'}$ if $\mu(a_i) > \mu(a_{i'})$. This modification ensures that the time to exploitation from such states is minimal. 

Nevertheless, we might face a problem in states for which $\below(s)\neq\emptyset$. To illustrate, consider the action $\bl p_{i,j}$ for some $a_i\in \above(A), a_j\in \below(A)$. The probability we discover a reward of $H$ when selecting $\bl p_{i,j}$ is
\begin{align}\label{eq:q is identical to all}
\bl p_{i,j}(a_i)\Pr(X(a_i)=H)  + \bl p_{i,j}(a_j)\Pr(X(a_j)=H).
\end{align}
Consequently, we might favor $\bl p_{i,j}$ over $\bl p_{i',j}$ (for $a_{i'}\in \above(A)$) if it allows faster discovery of a positive reward (which is necessarily $H$). However, as we show next, the probability in Equation~\eqref{eq:q is identical to all} is the same regardless of the selection the arm from $\above(s)$. Observe that 
\[
\mu(a_i) = H\Pr(X(a_j)=H) +(-1)\cdot (1-\Pr(X(a_j)=H)) = (H+1)\Pr(X(a_j)=H)-1; 
\]
thus, by reformulating Equation~\eqref{eq:q is identical to all} we get
{
\thinmuskip=.2mu
\medmuskip=0mu plus .2mu minus .2mu
\thickmuskip=1mu plus 1mu
\begin{align*}\label{eq:q is identical to all elaborate}
\textnormal{Eq.}\eqref{eq:q is identical to all}&=\frac{-\mu(a_j)}{\mu(a_i)-\mu(a_j)}\Pr(X(a_i)=H)+ \frac{\mu(a_i)}{\mu(a_i)-\mu(a_j)}\Pr(X(a_j)=H)\\
& =\frac{-(H+1)\Pr(X(a_j)=H)+1}{\mu(a_i)-\mu(a_j)}\Pr(X(a_i)=H)+ \frac{(H+1)\Pr(X(a_i)=H)-1}{\mu(a_i)-\mu(a_j)}\Pr(X(a_j)=H)\\
& = \frac{\Pr(X(a_i)=H)-\Pr(X(a_j)=H)}{\mu(a_i)-\mu(a_j)}\\
& = \frac{\Pr(X(a_i)=H)-\Pr(X(a_j)=H)}{(H+1)\Pr(X(a_i)=H)-1-(H+1)\Pr(X(a_i)=H)+1}\\
& = \frac{1}{H+1}.
\end{align*}
}%

\begin{algorithm}[t]
\renewcommand{\algorithmiccomment}[1]{\texttt{\kibitz{blue}{\##1}}}
\caption{$\SEGB$ for Two-Supported Distributions \label{alg:alg of pi two supported}}
\begin{algorithmic}[1]
\REQUIRE the $\OGP$ policy
\STATE $s\gets A$
\WHILE[$s$ is not a terminal state] {$\OGP(s)\neq \emptyset$\label{algpi:two:while}}{
\STATE select $\OGP(s)$, and denote the realized action by $a_k$\label{algpi:two:play with ogp} \COMMENT{if $\below(s)=\emptyset$, prioritize the arms in $\above(s)$ according to the stochastic order}
\IF[a reward of $H$ was realized]{$x_{a_k}>0$} {
		\STATE \textbf{break}
}
\ENDIF
\STATE $s\gets s\setminus \{a_k\}$\label{algpi:two:update s}
}
\ENDWHILE
\IF{$x(a_k)=H$ for some explored arm $a_k$}{
\STATE exploit $a_k$ forever
}
\ELSE{
\STATE exploit $a_0$ forever
}
\ENDIF
\end{algorithmic}
\end{algorithm}
Since all $\OGP$ have the same expected exploration time, they all achieve the same social welfare. This completes the proof of Proposition~\ref{prop:bernoulli opt}.
\end{proofof}

\begin{proofof}{Proposition \ref{prop:i-d bounds}}
To prove the claim, we show that $\SEGB$ explores for at most $K(1+\frac{\eta}{\delta})$ rounds, and then exploits. First, $\SEGB$ uses $K_1 \leq K$ rounds following $\OGP$ until it reaches a terminal state (the while loop in Line~\ref{algpi:while} breaks). If all the realized rewards are negative, the exploration ends after that. Otherwise, if it discovers a positive reward, the value of that reward is at least $\delta$. Next, the Bernoulli trials will explore every unexplored arm w.p. of at least $ \frac{\delta}{\delta+\eta}$, or $ \frac{\delta+\eta}{\delta}$ rounds in expectation. Therefore, after $(K-K_1)(1+\frac{\eta}{\delta})$ rounds in expectation we explore all the remaining arms. Overall, the expected number of rounds devoted to exploration is
\[
K_1 + (K-K_1)\left(1+\frac{\eta}{\delta}\right) \leq K \left(1+\eta\E\left[{\frac{1}{\delta}}\right]\right).
\]
Ultimately, recall that $\lim_{T \rightarrow \infty }\mU_T(\SEGB) = \OPT_{\infty}$, so $\SEGB$ exploits an expected reward of $\OPT_{\infty}$ after it completes its exploration. Therefore,
\begin{align*}
\mU_T(\SEGB)& \geq  \frac{1}{T} \left[K \left(1++\eta\E\left[{\frac{1}{\delta}}\right]\right) \cdot 0 + \left(T-K \left(1+\eta\E\left[{\frac{1}{\delta}}\right]\right) \right)\OPT_{\infty}\right]. \\
&\geq\left(1-\frac{K \left(1+\eta\E\left[{\frac{1}{\delta}}\right]\right)}{T} \right)\OPT_{\infty}. 
\end{align*}
This completes the proof of Proposition \ref{prop:i-d bounds}.
\end{proofof}

\section{Proof of Statements from Section \ref{sec:thm1 outline}}\label{sec:aux}
\begin{proofof}{Proposition \ref{prop:optimal p valid}}
Fix a non-terminal state $U\subseteq A$. Further, for simplicity of notation, denote $V(a)\defeq W^*(U\setminus \{a\})$ for every $a\in U$. Due to Equation (\ref{eq:W elaborated}), the action that maximizes the reward at state $U$ is the solution $\bl p \in \Delta(U)$ of the following linear program:
\begin{equation}
\tag{P1} \label{eq:lp w}
\begin{array}{ll@{}ll}
\max \limits_{\bl p} &\sum_{a\in U} \bl p(a)V(a)   & \\
\text{subject to} & \sum_{a\in U} \bl p(a)\mu(a) \geq 0 &     \\
& \sum_{a\in U} \bl p(a)=1 &  \\
& 0\leq \bl p(a) \leq 1 & \textnormal{for all }a\in U
\end{array}
\end{equation}
Observe that for every $\bl p$ such that $ \sum_{a\in U} \bl p(a)\mu(a) \geq 0$, there exist coefficients $(\alpha_{i,j})_{i,j}$ such that for every $a_i\in U$, $\bl p(a_i)=\sum_{a_j\in U}\alpha_{i,j}\bl p_{i,j}(i)$, and
\[
\sum_{a\in U} \bl p(a)\mu(a) = \sum_{i,j}\alpha_{i,j}\left(\bl p_{i,j}(i)\mu(a_i)+\bl p_{i,j}(j)\mu(a_j)\right).
\]
Hence, an equivalent form of Problem (\ref{eq:lp w}) is
\begin{equation}
\tag{P2} \label{eq:lp w with alpha}
\begin{array}{ll@{}ll}
\max \limits_{\bl \alpha} &\sum_{a\in U} \alpha_{i,j}\left(\bl p_{i,j}(i)V(a_i)+\bl p_{i,j}(j)V(a_j)\right)  & \\
\text{subject to} & \sum_{i,j}\alpha_{i,j}\left(\bl p_{i,j}(i)\mu(a_i)+\bl p_{i,j}(j)\mu(a_j)\right) \geq 0 &     \\
&\sum_{i,j}\alpha_{i,j}=1 &  \\
& 0\leq \alpha_{i,j} \leq 1 \qquad  \textnormal{for all }(i,j)\in \{(i',j')\mid \bl p_{i',j'} \in \mP\cup \mP'\textnormal{ and } i',j' \in U \} &
\end{array}
\end{equation}
Finally, notice that the constraint $\sum_{i,j}\alpha_{i,j}\left(\bl p_{i,j}(i)\mu(a_i)+\bl p_{i,j}(j)\mu(a_j)\right) \geq 0$ holds for every selection of $(\alpha_{i,j})_{i,j}$ by the way we defined $\mP\cup \mP'$; thus, the maximum of Problem (\ref{eq:lp w with alpha}) is obtained when we set $\alpha_{i,j}=1$ for the pair $(i,j)$ that maximizes $\left(\bl p_{i,j}(i)V(a_i)+\bl p_{i,j}(j)V(a_j)\right) $.
\end{proofof}

\begin{claim}\label{claim:ass is not for W}
Consider a state $U\in\mS$, such that $\below(U) \geq 2$. Let  $a_j = \argmin_{a_{j'}\in\below(U)}\sigr_\pi(a_{j'})$, and let $a_{\tilde j}\in \below(U), a_{\tilde j} \neq a_j$. Under Assumption \ref{assumption:dominance}, it might be the case that $W^*(U\setminus \{a_j\}) < W^*(U\setminus \{a_{\tilde j}\})$.
\end{claim}
\begin{proofof}{Claim \ref{claim:ass is not for W}}
We prove the claim by providing an example, that could be easily extended to a family of infinitely many examples. Consider $K=3$, $A=\{a_1,a_2,a_3\}$ such that 
\[
X_1=\begin{cases}
-1 & \textnormal{w.p. 0.45}\\
1 & \textnormal{w.p. 0.55}
\end{cases}, \qquad
X_2=\begin{cases}
-10^6-2\epsilon & \textnormal{w.p. 0.5}\\
10^6 & \textnormal{w.p. 0.5}
\end{cases}, \qquad
X_3=\begin{cases}
-10^{\frac{1}{\epsilon}} & \textnormal{w.p. 0.5}\\
10^6 & \textnormal{w.p. 0.5}
\end{cases}
\]
For small $\epsilon$, say $\epsilon<\frac{1}{7}$, it is clear that $X_2$ stochastically dominates $X_3$. The resulting expected values are $\mu({a_1})=0.1,\mu({a_2})=-\epsilon,$ and $\mu({a_3})= -\Theta(10^{\frac{1}{\epsilon}})$. The intuition behind our selection of rewards is that arm $a_2$ could have high reward, and can be explored with probability $\bl p_{1,2}(2)=1-O(\epsilon)$. On the other hand, arm $a_3$ has a high reward with the same probability, but it is highly unlikely to explore it. More precisely, $\bl p_{1,3}(3)=\Theta(10^{-\frac{1}{\epsilon}})$. To finalize the argument, notice that
\[
W^*(A\setminus \{a_2\})=\bl p_{1,3}(1)\cdot R(\{a_3\})+\bl p_{1,3}(3)\cdot \bl p_{1,1}(1)\cdot R(\emptyset)=0.5\cdot 10^6+0.5\cdot 0.55\cdot 1 +O(\epsilon)
\]
while
\begin{align*}
W^*(A\setminus \{a_3\})&=\bl p_{1,2}(1)\cdot R(\{a_2\})+\bl p_{1,2}(2)\cdot \bl p_{1,1}(1)\cdot R(\emptyset)\\
&=0.75\cdot 10^6+0.25\cdot 0.55\cdot 1+O(\epsilon).
\end{align*}
The proof is completed by taking $\epsilon$ to zero.
\end{proofof}

\section{Bayesian Incentive Compatibility}\label{sec:IC in app}
\begin{proofof}{Proposition~\ref{prop:ic for uniform}}
The proof follows closely the proof of~\citet[Section E]{Fiduciary}; hence, we omit the details.
\end{proofof}

\begin{proofof}{Theorem~\ref{theorem: ic fee}}
$\ICSEGB$ is MIR by design, as it only recommends the default arm, a greedy arm, or according to $\SEGB$. To analyze its social welfare, observe that after $K(1+{\eta}\E\left[\frac{1}{\delta}\right])$ phases in expectation, we only explore. This follows directly from Proposition~\ref{prop:i-d bounds}. Since the length of each phase is $B=\ceil*{\frac{H}{\xi \gamma}}+1$, there are at most $O\left(\frac{K \eta H \E\left[\frac{1}{\delta}\right]}{\xi \gamma}  \right)$ exploration rounds in expectation. To show it is BIC, we divide the analysis according to the agent's index. The analysis closely follows the proof of \citet[Theorem 3]{Fiduciary}, and is hence omitted.
\end{proofof}
}

\fi}


\begin{thebibliography}{37}
\providecommand{\natexlab}[1]{#1}
\providecommand{\url}[1]{\texttt{#1}}
\expandafter\ifx\csname urlstyle\endcsname\relax
  \providecommand{\doi}[1]{doi: #1}\else
  \providecommand{\doi}{doi: \begingroup \urlstyle{rm}\Url}\fi

\bibitem[Amani et~al.(2019)Amani, Alizadeh, and Thrampoulidis]{amani2019linear}
S.~Amani, M.~Alizadeh, and C.~Thrampoulidis.
\newblock Linear stochastic bandits under safety constraints.
\newblock In \emph{Advances in Neural Information Processing Systems}, pages 9252--9262, 2019.

\bibitem[Atsidakou et~al.(2024)Atsidakou, Caramanis, Gergatsouli, Papadigenopoulos, and Tzamos]{atsidakou2024contextual}
A.~Atsidakou, C.~Caramanis, E.~Gergatsouli, O.~Papadigenopoulos, and C.~Tzamos.
\newblock Contextual pandora’s box.
\newblock In \emph{Proceedings of the AAAI Conference on Artificial Intelligence}, volume~38, pages 10944--10952, 2024.

\bibitem[Badanidiyuru et~al.(2013)Badanidiyuru, Kleinberg, and Slivkins]{badanidiyuru2013bandits}
A.~Badanidiyuru, R.~Kleinberg, and A.~Slivkins.
\newblock Bandits with knapsacks.
\newblock In \emph{2013 IEEE 54th Annual Symposium on Foundations of Computer Science}, pages 207--216. IEEE, 2013.

\bibitem[Bahar et~al.(2016)Bahar, Smorodinsky, and Tennenholtz]{Bahar2016}
G.~Bahar, R.~Smorodinsky, and M.~Tennenholtz.
\newblock Economic recommendation systems: One page abstract.
\newblock In \emph{Proceedings of the 2016 ACM Conference on Economics and Computation}, EC '16, pages 757--757, New York, NY, USA, 2016. ACM.
\newblock ISBN 978-1-4503-3936-0.

\bibitem[Bahar et~al.(2019)Bahar, Smorodinsky, and Tennenholtz]{bahar2019recommendation}
G.~Bahar, R.~Smorodinsky, and M.~Tennenholtz.
\newblock Social learning and the innkeeper challenge.
\newblock In \emph{ACM Conf. on Economics and Computation (EC)}, 2019.

\bibitem[Bahar et~al.(2020)Bahar, Ben-Porat, Leyton-Brown, and Tennenholtz]{Fiduciary}
G.~Bahar, O.~Ben-Porat, K.~Leyton-Brown, and M.~Tennenholtz.
\newblock Fiduciary bandits.
\newblock In \emph{International Conference on Machine Learning}, pages 518--527. PMLR, 2020.

\bibitem[Barto et~al.(1995)Barto, Bradtke, and Singh]{barto1995learning}
A.~G. Barto, S.~J. Bradtke, and S.~P. Singh.
\newblock Learning to act using real-time dynamic programming.
\newblock \emph{Artificial Intelligence}, 72\penalty0 (1-2):\penalty0 81--138, 1995.

\bibitem[Ben{-}Porat and Tennenholtz(2018)]{BenPoratT18}
O.~Ben{-}Porat and M.~Tennenholtz.
\newblock A game-theoretic approach to recommendation systems with strategic content providers.
\newblock In \emph{Advances in Neural Information Processing Systems 31: Annual Conference on Neural Information Processing Systems 2018, NeurIPS 2018, 3-8 December 2018, Montr{\'{e}}al, Canada.}, pages 1118--1128, 2018.

\bibitem[Ben-Porat and Torkan(2023)]{ben2023learning}
O.~Ben-Porat and R.~Torkan.
\newblock Learning with exposure constraints in recommendation systems.
\newblock In \emph{Proceedings of the ACM Web Conference 2023}, pages 3456--3466, 2023.

\bibitem[Berger et~al.(2023)Berger, Ezra, Feldman, and Fusco]{berger2023pandora}
B.~Berger, T.~Ezra, M.~Feldman, and F.~Fusco.
\newblock Pandora's problem with combinatorial cost.
\newblock In \emph{Proceedings of the 24th ACM Conference on Economics and Computation}, pages 273--292, 2023.

\bibitem[Boodaghians et~al.(2020)Boodaghians, Fusco, Lazos, and Leonardi]{boodaghians2020pandora}
S.~Boodaghians, F.~Fusco, P.~Lazos, and S.~Leonardi.
\newblock Pandora's box problem with order constraints.
\newblock In \emph{Proceedings of the 21st ACM Conference on Economics and Computation}, pages 439--458, 2020.

\bibitem[Braverman et~al.(2019)Braverman, Mao, Schneider, and Weinberg]{braverman2019multi}
M.~Braverman, J.~Mao, J.~Schneider, and S.~M. Weinberg.
\newblock Multi-armed bandit problems with strategic arms.
\newblock In \emph{Conference on Learning Theory}, pages 383--416. PMLR, 2019.

\bibitem[Chen et~al.(2018)Chen, Frazier, and Kempe]{chen18a}
B.~Chen, P.~Frazier, and D.~Kempe.
\newblock Incentivizing exploration by heterogeneous users.
\newblock In S.~Bubeck, V.~Perchet, and P.~Rigollet, editors, \emph{Proceedings of the 31st Conference On Learning Theory}, volume~75 of \emph{Proceedings of Machine Learning Research}, pages 798--818. PMLR, 06--09 Jul 2018.

\bibitem[Cohen and Mansour(2019)]{cohen2019optimal}
L.~Cohen and Y.~Mansour.
\newblock Optimal algorithm for bayesian incentive-compatible exploration.
\newblock In \emph{Proceedings of the 2019 ACM Conference on Economics and Computation}, pages 135--151, 2019.

\bibitem[Feng et~al.(2020)Feng, Parkes, and Xu]{feng2020intrinsic}
Z.~Feng, D.~Parkes, and H.~Xu.
\newblock The intrinsic robustness of stochastic bandits to strategic manipulation.
\newblock In \emph{International Conference on Machine Learning}, pages 3092--3101. PMLR, 2020.

\bibitem[Gergatsouli and Tzamos(2024)]{gergatsouli2024weitzman}
E.~Gergatsouli and C.~Tzamos.
\newblock Weitzman's rule for pandora's box with correlations.
\newblock \emph{Advances in Neural Information Processing Systems}, 36, 2024.

\bibitem[Hu et~al.(2023)Hu, Jagadeesan, Jordan, and Steinhardt]{hu2023incentivizing}
X.~Hu, M.~Jagadeesan, M.~I. Jordan, and J.~Steinhardt.
\newblock Incentivizing high-quality content in online recommender systems.
\newblock \emph{arXiv preprint arXiv:2306.07479}, 2023.

\bibitem[Immorlica et~al.(2019)Immorlica, Mao, Slivkins, and Wu]{immorlica2019bayesian}
N.~Immorlica, J.~Mao, A.~Slivkins, and Z.~S. Wu.
\newblock Bayesian exploration with heterogeneous agents, 2019.

\bibitem[Immorlica et~al.(2024)Immorlica, Jagadeesan, and Lucier]{immorlica2024clickbait}
N.~Immorlica, M.~Jagadeesan, and B.~Lucier.
\newblock Clickbait vs. quality: How engagement-based optimization shapes the content landscape in online platforms.
\newblock In \emph{Proceedings of the ACM on Web Conference 2024}, pages 36--45, 2024.

\bibitem[Joseph et~al.(2016)Joseph, Kearns, Morgenstern, and Roth]{MatthewKearnsMorgensternRothNIPS2016}
M.~Joseph, M.~Kearns, J.~H. Morgenstern, and A.~Roth.
\newblock Fairness in learning: Classic and contextual bandits.
\newblock In D.~D. Lee, M.~Sugiyama, U.~V. Luxburg, I.~Guyon, and R.~Garnett, editors, \emph{Advances in Neural Information Processing Systems 29}, pages 325--333. Curran Associates, Inc., 2016.

\bibitem[Kazerouni et~al.(2017)Kazerouni, Ghavamzadeh, Abbasi, and Van~Roy]{kazerouni2017conservative}
A.~Kazerouni, M.~Ghavamzadeh, Y.~Abbasi, and B.~Van~Roy.
\newblock Conservative contextual linear bandits.
\newblock In \emph{Advances in Neural Information Processing Systems}, pages 3910--3919, 2017.

\bibitem[Kremer et~al.(2013)Kremer, Mansour, and Perry]{KremerMP13}
I.~Kremer, Y.~Mansour, and M.~Perry.
\newblock Implementing the "wisdom of the crowd".
\newblock In \emph{Proceedings of the fourteenth {ACM} Conference on Electronic Commerce, {EC} 2013, Philadelphia, PA, USA, June 16-20, 2013}, pages 605--606, 2013.

\bibitem[Kremer et~al.(2014)Kremer, Mansour, and Perry]{Kremer2014}
I.~Kremer, Y.~Mansour, and M.~Perry.
\newblock Implementing the wisdom of the crowd.
\newblock \emph{Journal of Political Economy}, 122:\penalty0 988--1012, 2014.

\bibitem[Liu et~al.(2017)Liu, Radanovic, Dimitrakakis, Mandal, and Parkes]{liu2017calibrated}
Y.~Liu, G.~Radanovic, C.~Dimitrakakis, D.~Mandal, and D.~C. Parkes.
\newblock Calibrated fairness in bandits, 2017.

\bibitem[Mansour et~al.(2015)Mansour, Slivkins, and Syrgkanis]{mansour2015bayesian}
Y.~Mansour, A.~Slivkins, and V.~Syrgkanis.
\newblock Bayesian incentive-compatible bandit exploration.
\newblock In \emph{Proceedings of the Sixteenth ACM Conference on Economics and Computation}, pages 565--582, 2015.

\bibitem[Mansour et~al.(2016)Mansour, Slivkins, Syrgkanis, and Wu]{Mansour2016Slivkins}
Y.~Mansour, A.~Slivkins, V.~Syrgkanis, and Z.~S. Wu.
\newblock Bayesian exploration: Incentivizing exploration in bayesian games.
\newblock In \emph{Proceedings of the 2016 ACM Conference on Economics and Computation}, EC '16, pages 661--661, New York, NY, USA, 2016. ACM.

\bibitem[Mansour et~al.(2020)Mansour, Slivkins, and Syrgkanis]{mansour2020bayesian}
Y.~Mansour, A.~Slivkins, and V.~Syrgkanis.
\newblock Bayesian incentive-compatible bandit exploration.
\newblock \emph{Operations Research}, 68\penalty0 (4):\penalty0 1132--1161, 2020.

\bibitem[Mladenov et~al.(2020)Mladenov, Creager, Ben-Porat, Swersky, Zemel, and Boutilier]{googleOmer20}
M.~Mladenov, E.~Creager, O.~Ben-Porat, K.~Swersky, R.~Zemel, and C.~Boutilier.
\newblock Optimizing long-term social welfare in recommender systems: A constrained matching approach.
\newblock In \emph{Proceedings of the 37th International Conference on Machine Learning}, volume 119, pages 6987--6998. PMLR, 13--18 Jul 2020.

\bibitem[Moldovan and Abbeel(2012)]{Moldovan:2012}
T.~M. Moldovan and P.~Abbeel.
\newblock Safe exploration in markov decision processes.
\newblock In \emph{Proceedings of the 29th International Conference on Machine Learning}, ICML'12, pages 1451--1458, 2012.
\newblock ISBN 978-1-4503-1285-1.

\bibitem[Nisan and Ronen(1999)]{nisan1999algorithmic}
N.~Nisan and A.~Ronen.
\newblock Algorithmic mechanism design.
\newblock In \emph{Proceedings of the thirty-first annual ACM Symposium on Theory of Computing (STOC)}, pages 129--140. ACM, 1999.

\bibitem[Nisan et~al.(2007)Nisan, Roughgarden, Tardos, and Vazirani]{nisan2007algorithmic}
N.~Nisan, T.~Roughgarden, E.~Tardos, and V.~V. Vazirani.
\newblock \emph{Algorithmic game theory}, volume~1.
\newblock Cambridge University Press Cambridge, 2007.

\bibitem[Simchowitz and Slivkins(2024)]{simchowitz2024exploration}
M.~Simchowitz and A.~Slivkins.
\newblock Exploration and incentives in reinforcement learning.
\newblock \emph{Operations Research}, 72\penalty0 (3):\penalty0 983--998, 2024.

\bibitem[Slivkins(2019)]{slivkins2019introduction}
A.~Slivkins.
\newblock Introduction to multi-armed bandits.
\newblock \emph{Foundations and Trends{\textregistered} in Machine Learning}, 12\penalty0 (1-2):\penalty0 1--286, 2019.

\bibitem[Tadelis(2013)]{tadelis2013game}
S.~Tadelis.
\newblock \emph{Game theory: an introduction}.
\newblock Princeton university press, 2013.

\bibitem[Wachi and Sui(2020)]{wachi2020safe}
A.~Wachi and Y.~Sui.
\newblock Safe reinforcement learning in constrained markov decision processes.
\newblock In \emph{International Conference on Machine Learning}, pages 9797--9806. PMLR, 2020.

\bibitem[Weitzman(1978)]{weitzman1978optimal}
M.~Weitzman.
\newblock \emph{Optimal search for the best alternative}, volume~78.
\newblock Department of Energy, 1978.

\bibitem[Wu et~al.(2016)Wu, Shariff, Lattimore, and Szepesv{\'a}ri]{wu2016conservative}
Y.~Wu, R.~Shariff, T.~Lattimore, and C.~Szepesv{\'a}ri.
\newblock Conservative bandits.
\newblock In \emph{International Conference on Machine Learning}, pages 1254--1262, 2016.

\end{thebibliography}
\end{document}